\documentclass[12pt]{article}
\usepackage[letterpaper, margin=1in]{geometry}

\usepackage{setspace}
\usepackage{fancyhdr}
\usepackage{authblk}

\usepackage{amsmath}
\usepackage{amsthm}
\usepackage{amssymb}
\usepackage{mathrsfs}

\usepackage{hyperref}[colorlinks, citecolor=blue, urlcolor=blue]
\usepackage{enumitem}

\usepackage{longtable}
\usepackage{adjustbox}
\usepackage{threeparttable}
\usepackage{booktabs}

\DeclareMathOperator{\sgn}{sgn}
\DeclareMathOperator{\cossimilarity}{cossimilarity}
\DeclareMathOperator{\var}{var}
\DeclareMathOperator{\cutsin}{cutsin}
\DeclareMathOperator{\Imb}{Imb}

\numberwithin{equation}{section}

\theoremstyle{plain}
\newtheorem{theorem}{Theorem}[section]
\newtheorem{lemma}{Lemma}[section]
\newtheorem*{lemma*}{Lemma}
\newtheorem{corollary}{Corollary}[section]
\newtheorem{assumption}{Assumption}[section]
\newtheorem{condition}{Condition}[section]

\theoremstyle{definition}

\newtheorem{remark}{Remark}[section]
\newtheorem{example}{Example}[section]

\begin{document}

\title{A General (Non-Markovian) Framework for Covariate Adaptive Randomization: Achieving Balance While Eliminating the Shift}
\author[1]{Hengjia Fang\thanks{hengjiafang@ruc.edu.cn}}
\author[1]{Wei Ma\thanks{mawei@ruc.edu.cn}}
\affil[1]{Institute of Statistics and Big Data, Renmin University of China, Beijing, China}

\date{} 

\maketitle

\begin{abstract}
    Emerging applications increasingly demand flexible covariate adaptive randomization (CAR) methods that support unequal targeted allocation ratios.
    While existing procedures can achieve covariate balance, they often suffer from the shift problem.
    This occurs when the allocation ratios of some additional covariates deviate from the target.
    We show that this problem is equivalent to a mismatch between the conditional average allocation ratio and the target among units sharing specific covariate values, revealing a failure of existing procedures in the long run.
    To address it, we derive a new form of allocation function by requiring that balancing covariates ensures the ratio matches the target.
    Based on this form, we design a class of parameterized allocation functions.
    When the parameter roughly matches certain characteristics of the covariate distribution, the resulting procedure can balance covariates.
    Thus, we propose a feasible randomization procedure that updates the parameter based on collected covariate information, rendering the procedure non-Markovian.
    To accommodate this, we introduce a CAR framework that allows non-Markovian procedure.
    We then establish its key theoretical properties, including the boundedness of covariate imbalance in probability and the asymptotic distribution of the imbalance for additional covariates.
    Ultimately, we conclude that the feasible randomization procedure can achieve covariate balance and eliminate the shift.
\end{abstract}
\textbf{Keywords:} Covariate Adaptive Randomization, Unequal Allocation Ratio, Covariate Balance, Shift Problem, Non-Markovian

\section{Introduction}

In comparative studies, such as clinical trials and economic field experiments, covariates play an important role.
In the design stage, a major concern for experimenters is the balance between treatment groups with respect to key covariates.
The imbalance of covariates may lead to confounding, which undermines the credibility of treatment comparisons and reduces statistical efficiency (\cite{bruhnPursuitBalanceRandomization2009,zhaoStatisticsPrecisionHealth2024}).

Despite its simplicity, complete randomization often fails to balance key baseline covariates across treatments.
To achieve covariate balance, covariate adaptive randomization (CAR) procedures are routinely employed to assign treatment status in randomized controlled studies.
In a survey of randomized clinical trials published in leading medical journals, 82\% were found to have used certain forms of CAR procedures (\cite{linPursuitBalanceOverview2015}).
In economics, CAR procedures are widely used and accepted in development economics, as evidenced by a growing body of empirical studies (\cite{bugniInferenceCovariateadaptiveRandomization2019,bruhnPursuitBalanceRandomization2009,liuTestingHeterogeneousTreatment2024,dufloUsingRandomizationDevelopment2007}).
The most commonly used CAR methods across disciplines are stratified randomization, in which units are first stratified by baseline covariates and then randomized within strata to ensure balanced treatment allocation.

Although stratified randomization is commonly used to achieve balance (\cite{maCaratPackageCovariateAdaptive2023,zelenRandomizationStratificationPatients1974}), minimization becomes especially advantageous when many covariates are involved, as it directly balances covariate margins and maintains good performance (\cite{zhaoConsistentCovariancesEstimation2024, kangIssuesOutcomesResearch2008,tavesMinimizationNewMethod1974}).
Early minimization procedures were limited to handling only discrete covariates (\cite{tavesMinimizationNewMethod1974,pocockSequentialTreatmentAssignment1975}).
More recently, the principle of minimizing imbalance measure have led to the development of new CAR procedures with available theoretical guarantees.
These procedures can directly balance continuous covariates marginally without requiring discretization (\cite{maNewUnifiedFamily2024,yangSequentialCovariateadjustedRandomization2024,qinAdaptiveRandomizationMahalanobis2024,liuPropertiesCovariateadaptiveRandomization2025}).
The underlying theory for covariate balance in these procedures is based on the Markov chain theory, similar to the theoretical foundations of early minimization procedures (\cite{huAsymptoticPropertiesCovariateadaptive2012,huTheoryCovariateAdaptiveDesigns2020}).
However, the principle of minimization encounters several issues under unequal allocation ratios.
Problems related to the rerandomization test were identified earlier (\cite{proschanMinimizeUseMinimization2011}), and when applied to balancing the margins of continuous covariates, it has recently been found to cause severe imbalances in additional covariates, i.e., variables that are not directly used in the randomization (\cite{liuPropertiesCovariateadaptiveRandomization2025}).
The focus of this paper is on the imbalance in additional covariates.
We address the severe imbalances in additional covariates through a theoretical analysis and introduce a special CAR procedure, which is distinct from minimization and operates within a non-Markovian framework.

In Subsection \ref{subsec_general_framework}, we establish a general framework for randomization procedures that is broad enough to cover most existing CAR procedures under arbitrary targeted allocation ratio.
As illustrated in Section \ref{sec_applications_existing_procedures}, existing Markovian CAR procedures can be viewed as special cases of this framework (\cite{maNewUnifiedFamily2024,zhangAsymptoticPropertiesMultitreatment2023,yangSequentialCovariateadjustedRandomization2024,huAsymptoticPropertiesCovariateadaptive2012,huTheoryCovariateAdaptiveDesigns2020}).
In addition, our framework also allows for a varying allocation function by updating the parameter of the function during the randomization procedure.
Some parameter update methods may produce a non-Markovian parameter sequence, thereby violating the Markov property of the randomization procedure.

Within the general framework, Subsection \ref{subsec_boundedness_in_probability} extends existing results on covariate balance.
Under an arbitrary targeted allocation ratio $\rho:(1-\rho)$, the marginal imbalance of covariates in this article can be measured by the imbalance vector $\Lambda_n = \sum_{i=1}^n (T_i-\rho)\phi(X_{\mathrm{origin},i})$, where $T_i$ is the assignment of the $i$th unit, $\phi$ is a feature map and $X_{\mathrm{origin},i}$ is the baseline covariate vector of the $i$th unit (\cite{liuPropertiesCovariateadaptiveRandomization2025}).
The best available convergence rate of it is $O_P(1)$, compared to $O_P(\sqrt{n})$ under complete randomization.
However, in the literature on CAR procedures, due to the irreducibility of the corresponding Markov chains, the sufficient conditions for achieving $O_P(1)$ convergence are usually derived separately for discrete and continuous covariates (\cite{huAsymptoticPropertiesCovariateadaptive2012,huTheoryCovariateAdaptiveDesigns2020,maNewUnifiedFamily2024,zhangAsymptoticPropertiesMultitreatment2023}).
In Subsection \ref{subsec_boundedness_in_probability}, we show that, without relying on the Markov chain theory, the best convergence rate of $O_P(1)$ can be achieved in the unified setting under a negative feedback condition on the randomization procedure.
This result unifies previous analyses that treated discrete and continuous covariates separately and provides a new theoretical approach to the analysis of covariate balance.

In the context of the unequal targeted allocation ratio, the shift problem in additional covariates is a novel problem which was discovered recently (\cite{liuPropertiesCovariateadaptiveRandomization2025}), for which we provide a theoretical explanation in Section \ref{sec_shift_problem}.
To be more specific, Liu, Hu and Ma found that in the context of the unequal targeted allocation ratio, the imbalance vectors associated with additional covariates no longer center around $0$ when their new randomization procedure of minimizing the imbalance measure is applied.
This phenomenon is similar to the finding of Proschan et al., who showed that applying minimization under unequal allocation can introduce substantial bias at the mean of the rerandomization distribution, thereby reducing the validity and statistical power of the rerandomization test for treatment effect (\cite{proschanMinimizeUseMinimization2011}).
In this article, we show that, when the parameter of allocation function is fixed, the solution to the shift problem requires a specific relationship between the allocation function and the invariant probability of the corresponding Markov chain. Such a requirement is difficult to meet without a carefully designed randomization procedure.
Moreover, we show that the shift problem is equivalent to a discrepancy between the nominal target $\rho:(1-\rho)$ and the actual targeted allocation ratio, which is the conditional average allocation ratio for units with the same covariate value in the long run, as shown in Subsection \ref{subsec_shift_properties_fixed}.
When the targeted allocation ratio is equal, the discrepancy is automatically eliminated by the symmetry of the allocation function and the corresponding stochastic process for the imbalance vector $\Lambda_n$.
The absence of symmetry precludes a natural mechanism for eliminating the shift, making the problem particularly challenging.
Furthermore, Subsection \ref{subsec_shift_properties_convergent} shows that if the parameter converges and the stepwise parameter updates are small, the imbalance of the additional covariate, $\sum_{i=1}^n (T_i - \rho) Y_i$, retains the same asymptotic properties as in the case where the parameter is fixed at its limit.
Thus, it suffices to consider the shift problem under the fixed parameters.

In Section \ref{sec_applications_existing_procedures}, under arbitrary targeted allocation ratio, we apply the results from previous sections to show that the imbalance vector $\Lambda_n = O_P(1)$ under the CAR procedures that minimize a pre-specified imbalance measure at each allocation step, regardless of whether the covariate is continuous or discrete.
This feasibility requires only verifying the negative feedback condition.
For a special case of procedures that balance discrete covariates, such as Pocock and Simon's minimization procedure, the result $\Lambda_n = O_P(1)$ implies that their generalizations under unequal allocation can still achieve balance for the discrete covariates.
Moreover, through a constructive argument, we show that when the covariate to be balanced is continuous and unbounded, the procedure in (\cite{liuPropertiesCovariateadaptiveRandomization2025}) necessarily leads to a discrepancy between the nominal target $\rho$ and the actual targeted allocation ratio $\tilde{\rho}_\theta(x)$, implying that the shift problem also exists in theory, even though the procedure guarantees that $\Lambda_n = O_P(1)$.

In Subsection \ref{subsec_new_allocation_function_form}, motivated by the results in Section \ref{sec_shift_problem}, we derive a new form of the allocation function by requiring that balancing covariates leads to an invariant probability measure under which the actual targeted allocation ratio equals $\rho:(1-\rho)$.
This form is markedly distinct from the allocation function induced by minimization.
Based on this form, we construct a class of parameterized allocation functions.
For each covariate distribution, we define an oracle parameter as a functional of the distribution.
When the parameter of the allocation function is fixed and sufficiently close to the oracle parameter, the resulting procedure can achieve marginal covariate balance, and the imbalance vector associated with any additional covariate centers around $0$ asymptotically.
Furthermore, within the general framework, we allow the parameter to be updated adaptively during the randomization process.
Specifically, the procedure is implemented via a sequence of estimates that continuously approximate the oracle parameter based on accumulating covariate information.
The properties of this procedure are presented in Subsection \ref{subsec_randomization_procedure_convergent_parameter}.

The rest of the article is organized as follows.
We first introduce the general framework and discuss the covariate balance properties in Section \ref{sec_general_framework}.
In Section \ref{sec_shift_problem}, we analyze the shift problem and demonstrate its presence in several representative randomization procedures.
Section \ref{sec_addressing_shift_problem} proposes a randomization procedure with convergent parameters to address the shift problem and establishes the asymptotic properties under convergent parameters.
Numerical studies are presented in Section \ref{sec_experiments}.
Finally, we conclude our article and provide directions for future work in Section \ref{sec_conclusion}.

\subsection*{Notation}

Denote the Lebesgue measure by $\mu_{\mathrm{leb}}$.
Let $\mathbb{E}_\theta$ denote the expectation under the randomization procedure with a fixed parameter sequence $\{\theta\}_{n \in \mathbb{N}}$.
Let $\mu f = \int f(\Lambda) \mu(d\Lambda)$ and $\mu \left[ g(\cdot, x) \right] = \int g(\Lambda,x) \mu(d\Lambda)$.
We denote the norm of the parameter by $\|\cdot\|_F$, which refers to the Frobenius norm when the parameter is a matrix.

For a function $V: \mathrm{X} \to [1,+\infty)$, define the $V$-norm of a function $f: \mathrm{X} \to \mathbb{R}$ by
\begin{equation*}
    |f|_V := \sup_{x \in \mathrm{X}} \frac{|f|(x)}{V(x)}
\end{equation*}
When $V=1$, the $V$-norm is the supremum norm denoted by $|f|_{\infty}$. For $\mu$ a finite signed measure on $(\mathrm{X}, \mathcal{X})$ and $V: \mathrm{X} \to [1, \infty)$ such that $|\mu|(V)<\infty$, where $|\mu|$ is the variation of $\mu$, we define $\|\mu\|_V$ the $V$-norm of $\mu$ as
\begin{equation*}
    \|\mu\|_V := \sup_{|f|_V \leq 1}|\mu(f)|
\end{equation*}
When $V \equiv 1$, the $V$-norm corresponds to the total variation norm.
Denote the state space by $\mathrm{X}$.
Let $P$ be a finite signed kernel on $(\mathrm{X}, \mathcal{X})$ such that $|P(x,\cdot)|(V) < \infty$ for any $x \in \mathrm{X}$.
If it exists, let $\pi_\theta$ denote the invariant probability measure corresponding to $P_\theta$.
For any measurable $f:\mathrm{X} \to \mathbb{R}$, we write $(Pf)(x) = P(x,f) := \int f(y) P(x,dy)$, which is the integral of $f$ with respect to the signed measure $P(x,\cdot)$.
Denote $(\mu P)(A) := \int \mu(\mathrm{d} x) P(x, A)$ and $P^n(x, A) := \int P(x, d y) P^{n-1}(y, A)$ for $n$-step transition kernel $P^n$.
The $V$-norm of $P$ is then defined by $\|P\|_V := \sup_{x \in \mathrm{X}} V^{-1}(x)\|P(x, \cdot)\|_V$.
Define $D_V$ that $D_V(\theta, \theta^\prime) := \| P_\theta - P_{\theta^\prime} \|_V$.

\section{General Framework of the Randomization Procedure}
\label{sec_general_framework}

\subsection{General Framework}

\label{subsec_general_framework}

For the $i$th unit, we denote its baseline covariate vector by $X_{\mathrm{origin},i}$, and define the transformed covariate used for balance as $X_i = \phi(X_{\mathrm{origin},i})$, where $\phi$ is a feature map.
Let $X_i$ be the $d$-dimensional covariate vector.
Denote some additional covariate by $Y_i$ (possibly unobserved by the experimenter), and we suppose that the vectors $\{(X_i, Y_i)\}_{i \in \mathbb{N}^*}$ are i.i.d. random vectors. Let $\Gamma$ be the distribution of $X_i$.

In this article, we only consider the two-treatment case. Let $T_i$ be the assignment of the $i$th unit, such that $T_i=1$ for the treatment and $T_i=0$ for the control.
Suppose the targeted allocation ratio for the treatment group is $\rho \in (0,1)$ and the ratio for the control group is $1-\rho$.
For brevity, we refer to ``the targeted allocation ratio for the treatment group'' as ``the targeted allocation ratio''.
The main targets of the CAR procedure in the context of an unequal targeted allocation ratio are as follows:
first, to achieve marginal balance of $X$, that is, to control the imbalance vector $\Lambda_n = \sum_{i=1}^n (T_i-\rho)X_i = o_P(\sqrt{n})$, ideally $\Lambda_n = O_P(1)$; and second, to ensure that the asymptotic distribution of $\sum_{i=1}^n (T_i-\rho)Y_i$ for any additional covariate $Y$ centers around zero.
If these goals have been achieved, we can conclude that the baseline covariates are well-balanced.

Let $\{\mathcal{F}_n\}_{n \in \mathbb{N}}$ be a filtration, where $\mathcal{F}_n$ represents all information collected by the experimenters after assigning the $n$th unit.
We define $\theta_n \in \mathcal{F}_n$ as a summary statistic of the population covariate distribution, based on information collected up to the $n$th step.
We next formally state an assumption that, through the filtration $\{\mathcal{F}_n\}_{n \in \mathbb{N}}$, restricts the treatment assignment and parameter update at each step to depend only on the information collected up to that point, excluding any influence from future covariates.
\begin{assumption}
    \label{assumption_filtration}
    
    The filtration $\{\mathcal{F}_n\}_{n \in \mathbb{N}}$ satisfies that $\theta_0 \in \mathcal{F}_0$ and for any $n \in \mathbb{N}^*$, $\{(X_n, Y_n, T_n, \theta_n)\}$ is $\mathcal{F}_n$-adapted.
    In addition, the information of the subsequent units, $\left\{ (X_{n+k},Y_{n+k}) \mid k \in \mathbb{N}^* \right\}$, is independent of $\mathcal{F}_n$.
\end{assumption}

In the framework of randomization procedures in our article, we consider a data-adaptive allocation mechanism such that the allocation of $(n+1)$th unit depends on the history $\mathcal{F}_n$ and the information of the unit through the parameter $\theta_n$, the imbalance vector $\Lambda_n$ and the covariate vector $X_{n+1}$.
Accordingly, for $(n+1)$th unit, the conditional probability of treatment assignment $T_{n+1}$ is
\begin{align}
    P(T_{n+1} = 1 \mid \mathcal{F}_n, X_{n+1}, Y_{n+1}) &= P(T_{n+1} = 1 \mid \theta_n, \Lambda_n, X_{n+1}) = g_{\theta_n}(\Lambda_n, X_{n+1}) \label{eq_T_conditional_1} \\
    P(T_{n+1} = 0 \mid \mathcal{F}_n, X_{n+1}, Y_{n+1}) &= 1-P(T_{n+1} = 1 \mid \mathcal{F}_n, X_{n+1}) . \label{eq_T_conditional_0}
\end{align}
In (\ref{eq_T_conditional_1}) and (\ref{eq_T_conditional_0}), the function $g_\theta(\Lambda, X)$ is referred to as the allocation function.
Combining with the constraint on the codomain of the allocation function $g_\theta$, we can now formally state the assumption on the treatment assignment:
\begin{assumption}
    \label{assumption_sampling}
    The conditional probability of the treatment assignment $T_n$ satisfies (\ref{eq_T_conditional_1}) and (\ref{eq_T_conditional_0}).
    In addition, there exists some $\iota > 0$ such that the allocation function $g_\theta(\Lambda, X) \in [\iota, 1 - \iota]$.
\end{assumption}

Throughout this article, we always assume that these two assumptions hold.

Based on the setting of the allocation, we can define the transition probability kernel $P_\theta$ of the imbalance vector $\Lambda$ by a form of mapping between functions that
\begin{equation}
    P_\theta(\Lambda,h) = \int \left[ g_\theta(\Lambda, X)h(\Lambda+(1-\rho)X) + [1-g_\theta(\Lambda, X)]h(\Lambda-\rho X) \right] \Gamma(dX) \label{eq_transition_kernel} ,
\end{equation}
for any bounded measurable function $h$. If the function $h$ is taken as an indicator function $\mathbb{I}_A(\Lambda)$, the transition probability kernel can be written as
\begin{equation*}
    P_\theta(\Lambda,A) = \int_{(A-\Lambda)/(1-\rho)} g_\theta(\Lambda, X) \Gamma(dX) + \int_{(\Lambda-A)/\rho} [1-g_\theta(\Lambda, X)] \Gamma(dX) ,
\end{equation*}
which represents the conditional probability that the next state belongs to the set $A$ conditional on the last state $\Lambda$ and the parameter $\theta$. It is the traditional form of transition probability kernel.

With the definition of the kernel, we can establish the framework of the stochastic process for $\{(\Lambda_n,\theta_n)\}_{n \in \mathbb{N}}$ on the filtered probability space $(\Omega, \{\mathcal{F}_n, n \in \mathbb{N}\}, \mathbb{P})$.
\begin{align*}
    \mathbb{E} \left[ h(\Lambda_{n+1}) \mid \mathcal{F}_n \right]
    &= \mathbb{E} \left[ h(\Lambda_n + (T_{n+1}-\rho)X_{n+1}) \mid \mathcal{F}_n \right] \\
    &= \mathbb{E} \left[ \mathbb{E} \left[ h(\Lambda_n + (T_{n+1}-\rho)X_{n+1}) \mid \mathcal{F}_n, X_{n+1} \right] \mid \mathcal{F}_n \right] \\
    &= \mathbb{E} \left[ \left[ h(\Lambda_n + (1-\rho)X_{n+1})g_{\theta_n}(\Lambda_n, X_{n+1}) \right. \right. \\
    & \quad \left. \left. + h(\Lambda_n - \rho X_{n+1})(1-g_{\theta_n}(\Lambda_n, X_{n+1})) \right]
    \mid \mathcal{F}_n \right] \\
    &= \mathbb{E}_{X \sim \Gamma} \left[ h(\Lambda_n + (1-\rho)X)g_{\theta_n}(\Lambda_n, X)
    + h(\Lambda_n - \rho X)(1-g_{\theta_n}(\Lambda_n, X)) \right] \\
    &= P_{\theta_n}(\Lambda_n,h) ,
\end{align*}
where the fourth equality is from the independence between $X_{n+1}$ and $\mathcal{F}_n$.
The model behind this equality coincides with a common model of adaptive MCMC; See \cite{fortCentralLimitTheorem2014,laitinenInvitationAdaptiveMarkov2024} for the theory and examples.
When $h$ is taken to be an indicator function, the conditional probability of the next state given the past information $\mathcal{F}_n$ can be obtained.

This model resembles a Markov chain due to the similar form of the transition kernel in (\ref{eq_transition_kernel}).
However, since the conditional distribution of $\Lambda_{n+1}$ given the history $\mathcal{F}_n$ depends not only on $\Lambda_n$ but also on $\theta_n$, which may be stochastic and non-Markovian, there is no guarantee that the process $\{\Lambda_n\}$ or $\{(\Lambda_n, \theta_n)\}$ is a Markov chain.
Even when the parameter sequence $\{\theta_n\}_{n \in \mathbb{N}}$ is given, the process $\{\Lambda_n\}$ is not necessarily Markovian.
This is because each parameter $\theta_n$ may encode information up to the $n$th unit, including the past transitions $\Lambda_m \rightarrow \Lambda_{m+1}$ for $m<n$.
Consequently, conditional on the entire parameter sequence $\{\theta_n\}_{n \in \mathbb{N}}$, the transition probability of $\Lambda_{n+1}$ given $\Lambda_n$ may still implicitly depend on future parameters $\{\theta_m\}_{m=n+1:\infty}$.
Thus, the process does not follow the Markov transition kernel $P_{\theta_n}$ at time $n$.

Finally, we summarize a concise procedure of the general framework as follows:
\begin{enumerate}
    \item Initialize the parameter $\theta_0$ and the imbalance vector $\Lambda_0 = 0$.
    \item Suppose the first $n$ units have been allocated and the $(n+1)$th unit with the covariate vector $X_{n+1}$ is to be allocated ($n \geq 0$).
    \item Allocate the $(n+1)$th unit to the treatment group with the probability $g_{\theta_n}(\Lambda_n,X_{n+1})$, otherwise to the control group.
    \item Calculate the parameter $\theta_{n+1}$ and $\Lambda_{n+1} = \Lambda_n + (T_{n+1}-\rho)X_{n+1}$.
    \item Repeat the last three steps until all units are allocated.
\end{enumerate}

\subsection{Boundedness of the Imbalance Vector}

\label{subsec_boundedness_in_probability}

In this subsection, we explore the conditions under which the imbalance vector $\Lambda_n$ is bounded in probability, that is, $\Lambda_n = O_P(1)$.
We use both terms interchangeably, referring to them as ``boundedness'' when clear.

We begin by imposing the following assumption to control the covariates:
\begin{assumption}
    \label{assumption_x_sub_exponential_bound}
    For some $\lambda > 0$, $\mathbb{E} \left[ \exp(\lambda \|X\|) \right]  = C< \infty$. It is equivalent to that $X$ is sub-exponential.
\end{assumption}

Next, because the transition from $\Lambda_n$ to $\Lambda_{n+1}$, conditional on the past information $\mathcal{F}_n$, is governed by the transition kernel $P_{\theta_n}$, it suffices to impose a suitable assumption on $P_{\theta_n}$ to control the behavior of the sequence $\{\Lambda_n\}$.

Denote the $(M,\Delta)$-negative feedback set
\begin{equation*}
    K_{M,\Delta}
    = \left\{
        \theta \middle|
        \mathbb{E}_\theta \left[ (\Lambda_1 - \Lambda_0)^T \frac{\Lambda_0}{\|\Lambda_0\|} \mid \Lambda_0 = \Lambda \right]
        \leq - \Delta ,
        \forall \Lambda \in W_\Gamma, \|\Lambda\| \geq M ,
    \right\} ,
\end{equation*}
where $W_\Gamma$ be the linear subspace spanned by the support of $\Gamma$.
When $\theta \in K_{M,\Delta}$, $P_\theta$ exhibits negative feedback in that, for large initial value $\|\Lambda_0\|$, the component of $\Lambda$ along $\Lambda_0$ tends to point in the opposite direction, thereby reducing $\|\Lambda\|$.

\begin{theorem}
    \label{theorem_boundedness}
    Suppose Assumption \ref{assumption_x_sub_exponential_bound} holds. If for some positive numbers $M$ and $\Delta$, $\theta_n \in K_{M,\Delta}$ for sufficiently large $n$ almost surely, then the stochastic process $\{\Lambda_n\}$ is bounded in probability.
\end{theorem}
\begin{remark}

    While existing analyses depend on the Markov property and ergodicity of $\{\Lambda_n\}$, our proof of this corollary does not rely on the Markov chain theory.
    Our analysis directly demonstrates the boundedness of $\{\Lambda_n\}$, thus offering a more foundational explanation for why the CAR procedure induces covariate balance.
    Compared to the condition in previous works (\cite{zhangAsymptoticPropertiesMultitreatment2023,maNewUnifiedFamily2024}), the condition in Theorem \ref{theorem_boundedness} is more general and places fewer demands on extra assumptions.
    We do not need the assumption on the form of the distribution of covariates, such as spread-out condition for continuous covariates or lattice condition for discrete covariates, to obtain irreducibility of Markov chain (\cite{zhangAsymptoticPropertiesMultitreatment2023,maNewUnifiedFamily2024}).
Moreover, this corollary shows that besides minimizing the imbalance measure at each step, there exist other possible ways to control the measure in the randomization procedure.
\end{remark}

Theorem \ref{theorem_boundedness} is applicable in scenarios beyond the fixed allocation function.
It reveals that the control of the imbalance vector $\Lambda$ only requires the negative feedback condition for $P_{\theta_n}$ for sufficiently large $n$.

\section{The Shift Problem}
\label{sec_shift_problem}

\subsection{Shift Problem under Fixed Parameters}

\label{subsec_shift_properties_fixed}

In this subsection, we establish the properties of $\sum_{i=1}^n (T_i - \rho) Y_i$ under fixed parameters. 
To simplify the discussion on additional covariates, we call the expression $\mathbb{E} \left[ \sum_{i=1}^n (T_i-\rho)Y_i \right]$ by the shift on $Y$, the shift divided by $n$ by the average shift on $Y$, and the phenomenon $\frac{1}{n}\mathbb{E} \left[ \sum_{i=1}^n (T_i-\rho)Y_i \right] \rightarrow 0$ by the asymptotic nonexistence of the average shift on $Y$.
Under the fixed parameters, the stochastic process is a Markov chain, and thus we can apply the Markov chain theory.

When balancing covariates under unequal allocation, Liu, Hu and Ma discovered that the imbalance of the additional covariate $Y$, $\frac{1}{n}\sum_{i=1}^n (T_i-\rho)Y_i$, may no longer center around $0$ (\cite{liuPropertiesCovariateadaptiveRandomization2025}).
This novel phenomenon is referred to as the shift problem.
They encountered the shift problem in simulations of balancing continuous covariates by using their methods.
This phenomenon is analogous to a known issue in rerandomization tests.
When the minimization procedure is repeatedly applied to the fixed observed data to balance discrete covariates, the resulting rerandomization distribution of the test statistic is no longer centered at zero (\cite{proschanMinimizeUseMinimization2011}).

\begin{assumption}
    \label{assumption_x_spread_out}

    The distribution $\Gamma_\rho = \rho((1-\rho)\Gamma)+(1-\rho)(-\rho\Gamma)$, which is the mixture of the scaled distributions $(1-\rho)\Gamma$ and $-\rho\Gamma$, is spread out, i.e., there exists a positive integer $d_s$ and a nonnegative measurable function $f_s(\Lambda)$ with $\int f_s(\Lambda) \mu_{\mathrm{leb}}(d\Lambda) > 0$ such that
    \begin{equation*}
        \Gamma_\rho^{d_s*}(A) \geq \int_A f_s(\Lambda) \mu_{\mathrm{leb}}(d\Lambda) ,
    \end{equation*}
    for any set $A$, where $\mu_{\mathrm{leb}}$ is the Lebesgue measure and $\Gamma_\rho^{d_s*}$ is the $d_s$th convolution of $\Gamma_\rho$.
\end{assumption}

\begin{remark}

    Assumption \ref{assumption_x_spread_out} is a mild condition that naturally holds when the covariate distribution $\Gamma$ is continuous.
Based on Example 3.1 of Ma et al. (\cite{maNewUnifiedFamily2024}), one possible choice of the feature map $\phi$ is $\phi(X_{\mathrm{origin}}) = (X_{\mathrm{origin}}^i)_{i=1:p}$ when $X_{\mathrm{origin}}$ has a nonzero absolutely continuous part with respect to $\mu_{\mathrm{leb}}$.
    When $\phi$ is chosen as above, the assumption is also satisfied and the target of the randomization procedure is to balance the first $p$ moments of $X_{\mathrm{origin}}$.
\end{remark}

Let the function $f(x)$ denote the conditional expectation $\mathbb{E} \left[ Y \mid X=x \right]$. It follows that the conditional expectation of $(T_n - \rho)Y_n$ given $\mathcal{F}_{n-1}$ equals $h_{\theta_{n-1}}(\Lambda_{n-1})$, where the conditional mean function $h_\theta(\Lambda)$ is defined as
\begin{equation}
    h_\theta(\Lambda) = \mathbb{E}_{X \sim \Gamma} \left[ [g_\theta(\Lambda, X) - \rho] f(X) \right] \label{eq_h_definition} .
\end{equation}
\begin{theorem}
    \label{theorem_boundedness_shift}

Suppose that for some positive numbers $M$ and $\Delta$, $\theta \in K_{M,\Delta}$.
    If Assumptions \ref{assumption_x_sub_exponential_bound} and \ref{assumption_x_spread_out} hold, then $P_\theta$ is positive recurrent with the unique invariant probability $\pi_\theta$.
    In addition, if $Y$ has a finite expectation, then under fixed parameters ${\theta_n = \theta}{n \in \mathbb{N}}$ and any initial imbalance vector $\Lambda$,
\begin{equation*}
        \left| \mathbb{E}_\theta \left[ \sum_{i=1}^n (T_i-\rho)Y_i \mid \Lambda_0 = \Lambda \right] - n\pi_\theta h_\theta \right| \leq L^2 \left[\mathbb{E}|Y|\right] V(\Lambda) < \infty ,
    \end{equation*}
where $L$ is a positive constant uniform in $\theta \in K_{M,\Delta}$.
\end{theorem}

This theorem states that the difference between the shift and the expectation $n\pi_\theta h_\theta$ is $O(1)$. Thus, we can conclude that the average shift, $\frac{1}{n} \mathbb{E}_\theta \left[ \sum_{i=1}^n (T_i-\rho)Y_i \right]$, converges to $\pi_\theta h_\theta$ at a fast rate $O(\frac{1}{n})$.

\begin{remark}
\label{remark_equivalent_form}
    Suppose the fixed covariate distribution $\Gamma$ satisfies Assumptions \ref{assumption_x_sub_exponential_bound} and \ref{assumption_x_spread_out}, and that the associated kernel $P_\theta$ corresponds to a parameter $\theta \in K_{M,\Delta}$.
    A consequence of Theorem \ref{theorem_boundedness_shift} is that the asymptotic nonexistence of the average shift on any $Y$ under the covariate distribution $\Gamma$ is equivalent to
    \begin{equation*}
        0 = \pi_\theta h_\theta
        = \mathbb{E}_{\Lambda \sim \pi_\theta, X \sim \Gamma} \left[ [g_\theta(\Lambda, X) - \rho] f(X) \right]
        = \mathbb{E}_{X \sim \Gamma} \left[ [\pi_\theta \left[ g_\theta(\cdot,X) \right] - \rho] f(X) \right] ,
    \end{equation*}
    where $\pi_\theta \left[ g_\theta(\cdot,X) \right]$ denotes the expectation of $g_\theta(\Lambda, X)$ under the invariant probability $\pi_\theta$ for a fixed covariate $X$.
    By the arbitrariness of the conditional expectation function $f$, it is equivalent to $\pi_\theta \left[ g_\theta(\cdot, x) \right] = \rho$ for $\Gamma$-a.e. $x$.
\end{remark}

\begin{remark}
\label{remark_stronger_equivalent_form}
    A slightly stronger equivalent form is that the equality $\pi_\theta \left[ g_\theta(\cdot, x) \right] = \rho$ holds for any $x$.
    This form is equivalent to the asymptotic nonexistence of the average shift on any $Y$ under any covariate distribution $\Gamma$ that satisfies Assumptions \ref{assumption_x_sub_exponential_bound} and \ref{assumption_x_spread_out} and leads to $\theta \in K_{M,\Delta}$.
    The sufficiency is immediate.
    For necessity, if the adjustment is sufficiently small, the invariant distribution under the adjusted measure can approximate that under $\Gamma$ arbitrarily well.
    Adding an infinitesimal Dirac mass at any point $x$ to $\Gamma$ thus enforces $\pi_\theta \left[ g_\theta(\cdot, x) \right] = \rho$ and traversing all $x$ completes the argument for necessity.
    We discuss necessity for illustrative purposes, and the necessity itself does not play a role in the formal proofs.
\end{remark}

Let $\tilde{\rho}_\theta(x) = \pi_\theta \left[ g_\theta(\cdot,x) \right]$ denote the actual targeted allocation ratio for a unit with covariate $x$ under the invariant probability $\pi_\theta$.
This quantity also corresponds to the conditional average allocation ratio in the long run, as justified below.
Consider an additional covariate defined by $Y = \mathbb{I}(X \in A)$, where $A$ is a measurable set.
By Theorem \ref{theorem_boundedness_shift}, the limit of $\frac{1}{N} \sum_{n=1}^{N} \left[ T_n \mathbb{I}(X_n \in A) \right]$ is $\mathbb{E}_{X \sim \Gamma} \left[ \tilde{\rho}_\theta(X)\mathbb{I}(X \in A) \right]$.
Thus, by choosing $A = \{x\}$ or by taking measurable sets that concentrate around $x$, it follows that $\tilde{\rho}_\theta(x)$ represents the conditional average allocation ratio for units with covariate value $x$ in the long run.

As indicated by Remarks \ref{remark_equivalent_form} and \ref{remark_stronger_equivalent_form}, the condition $\tilde{\rho}_\theta(x) \equiv \rho$ implies the asymptotic nonexistence of the average shift.
It motivates the construction of a new allocation function form in Subsection \ref{subsec_new_allocation_function_form} to address the shift problem.

\subsection{Shift Problem under Convergent Parameters}

\label{subsec_shift_properties_convergent}

In this subsection, we characterize the asymptotic distribution of $\sum_{i=1}^n (T_i-\rho)Y_i$ under the randomization procedure under convergent parameters.
To ensure the stability of the allocation function against continuously updated parameters, we impose two key assumptions.
\begin{assumption}
    \label{assumption_g_lipschitz}
    The allocation function families $g_\theta$ is Lipschitz continuous on $\overline{B(\theta^*, r_g)}$ in the sense of $|g_\theta(\Lambda, X) - g_{\theta^\prime}(\Lambda, X)| \leq L_g \|\theta-\theta^\prime\|$ for any $\Lambda$, $X$ and some constant $L_g>0$.
\end{assumption}
This assumption uniformly restricts changes in the allocation function with respect to the parameter $\theta$ over all possible covariate and imbalance vector values, while the next assumption limits the variation in the parameter sequence.
Denote the oracle value of the parameter by $\theta^*$.
\begin{assumption}
    \label{assumption_theta_convergence}
    The parameter sequence $\{\theta_n\}$ converges to $\theta^*$ almost surely. Moreover, for some $\epsilon>1/2$, the difference $\|\theta_n-\theta_{n+1}\|$ satisfies that $\sup_n n^\epsilon\|\theta_n-\theta_{n+1}\| < \infty$ almost surely.
\end{assumption}

Under these assumptions, we can establish the asymptotic normality of $\sum_{i=1}^n (T_i-\rho)Y_i$ for the additional covariate $Y$ can be established, with center $\frac{1}{N}\sum_{n=0}^{N-1} \pi_{\theta_n}h_{\theta_n}$, in analogy with Theorem 2.2 of Fort et al. (\cite{fortCentralLimitTheorem2014}).
Denote
\begin{equation*}
    \mu_n =
    \begin{cases}
        \pi_{\theta_n}(h_{\theta_n}), & \text{if }
        \|\theta_n-\theta^*\| \leq \min\{r_d, r_g\} , \\
        \pi_{\theta^*}(h_{\theta^*}), & \text{otherwise} .
    \end{cases}
\end{equation*}
We refer to ``the variance/covariance matrix of the asymptotic distribution'' as ``the asymptotic variance/covariance matrix'' for brevity.

\begin{theorem}
\label{theorem_CLT_complex_center}
    
    Suppose that there exist positive constants $r_d$, $M$, and $\Delta$ such that $\overline{B(\theta^*, r_d)} \subset K_{M,\Delta}$.
Denote some other additional covariate distinct from $Y_n$ by $Z_n$.
    Suppose the distribution of the covariate $X$ satisfies Assumptions \ref{assumption_x_sub_exponential_bound} and \ref{assumption_x_spread_out}, the allocation function satisfies Assumption \ref{assumption_g_lipschitz} and the parameter sequence $\{\theta_n\}$ satisfies Assumption \ref{assumption_theta_convergence}.
    If $Y$ and $Z$ have a finite second moment, then there is a $\sigma^*_{Y,Z} \geq 0$ such that
    \begin{equation}
        \frac{1}{\sqrt{N}} \sum_{n=1}^{N} { \left[ (T_n-\rho) Y_n + Z_n
        -\mu_{n-1}
        -\mathbb{E}Z_n \right] } \xrightarrow{d} \mathcal{N}(0, {\sigma^*_{Y,Z}}^2) \label{eq_CLT} ,
    \end{equation}
    where ${\sigma^*_{Y,Z}}^2$ is equal to the asymptotic variance under the fixed parameter $\theta^*$.

\end{theorem}

\section{Existing Covariate Adaptive Randomization Procedures under Unequal Allocation}
\label{sec_applications_existing_procedures}

\subsection{Balance Continuous Covariates}

\label{subsec_main_liu2024}

The randomization procedure proposed by Liu, Hu and Ma (\cite{liuPropertiesCovariateadaptiveRandomization2025}) tends to minimize the imbalance measure at each step.
When $X$ is continuous, we will formally prove the asymptotic existence of the average shift on some additional covariate.
We begin by presenting the randomization procedure in detail. Let the imbalance measure $\mathrm{Imb}_n$ be the squared Euclidean norm of the imbalance vector $\Lambda_n$,
\begin{equation*}
    \mathrm{Imb}_n=\left\|\sum_{i=1}^n\left(T_i-\rho\right) X_i\right\|^2 .
\end{equation*}

The procedure sequentially allocates units to minimize $\mathrm{Imb}_n$ in each step using a more biased allocation probability. Specifically, the randomization procedure is defined as follows:
\begin{enumerate}
    \item Assign the first unit to treatment with probability $\rho$ and to control with probability $1-\rho$. Set the initial imbalance vector $\Lambda_1=(T_1-\rho)X_1$.
    \item Suppose the first $n$ units have been allocated and the $(n+1)$th unit with the covariate vector $X_{n+1}$ is to be allocated ($n > 0$).
    If the $(n+1)$th unit is allocated to the treatment, the potential imbalance measures will be $\mathrm{Imb}_{n+1}^{(1)}=\left\|\Lambda_n+(1-\rho) X_{n+1}\right\|^2$. Similarly, $\mathrm{Imb}_{n+1}^{(0)}=$ $\left\|\Lambda_n-\rho X_{n+1}\right\|^2$ if the $(n+1)$th unit is allocated to the control.
    \item Assign the $(n+1)$th unit to the treatment with probability
    \begin{align*}
        P \left(T_{n+1}=1 \mid T_1, \ldots, T_n, X_1, \ldots, X_{n+1}\right)=
        \begin{cases}
            \rho_1 & \text { if } \Imb_{n+1}^{(1)}<\Imb_{n+1}^{(0)}, \\
            1-\rho_1 & \text { if } \Imb_{n+1}^{(1)}>\Imb_{n+1}^{(0)}, \\
            \rho & \text { if } \Imb_{n+1}^{(1)}=\Imb_{n+1}^{(0)},
        \end{cases}
    \end{align*}
    where $\max \{\rho, 1-\rho\} < \rho_1 < 1$. After calculation, the expression $\mathrm{Imb}_{n+1}^{(1)}-\mathrm{Imb}_{n+1}^{(0)}=2 X_{n+1}^T \Lambda_n + (1-2 \rho) X_{n+1}^T X_{n+1}$.
    \item Calculate $\Lambda_{n+1} = \Lambda_n + (T_{n+1}-\rho)X_{n+1}$.
    \item Repeat the last three steps until all units are allocated.
\end{enumerate}

Under the general framework proposed in Subsection \ref{subsec_general_framework}, this procedure is equivalent to the procedure with the fixed allocation function
\begin{align}
    & g_\theta(\Lambda, X) = \rho_1 \mathbb{I}(2X^T\Lambda+(1-2\rho)X^TX < 0) \label{eq_allocation_function_minimization} \\
    & \quad + (1-\rho_1) \mathbb{I}(2X^T\Lambda+(1-2\rho)X^TX > 0)
    + \rho \mathbb{I}(2X^T\Lambda+(1-2\rho)X^TX = 0) \notag
\end{align}
for the $n$th unit ($n > 1$). According to Theorem \ref{theorem_boundedness_shift}, the specific allocation function yields the theorem below:
\begin{theorem}
    \label{theorem_shift_liu_xiao}
    If the covariate $X$ has a finite second moment, then there exist constants $M>0$ and $\Delta>0$ such that the parameter $\theta \in K_{M,\Delta}$, and so, $\Lambda_n = O_P(1)$.
    Furthermore, suppose $\rho \neq \frac{1}{2}$ and $X$ is not bounded. If Assumptions \ref{assumption_x_sub_exponential_bound} and \ref{assumption_x_spread_out} hold, then there exists a measurable function $f(x)$ such that for the additional covariate $Y_i = f(X_i) + \epsilon_i$, the average shift on it, $\frac{1}{n}\mathbb{E} \left[ \sum_{i=1}^n (T_i-\rho)Y_i \right]$, does not converge to zero, where $\epsilon_i$ can be any i.i.d. error with zero expectation.
\end{theorem}
\begin{remark}

    This unboundedness condition on $X$ is not a necessary condition for the asymptotic existence of the average shift. If we instead assume that $X$ is bounded, or modify the distribution of the error term in $Y_i$, then the asymptotic existence of the average shift is still possible.
    Remarks \ref{remark_equivalent_form} and \ref{remark_stronger_equivalent_form} show that the condition $\pi_\theta \left[ g_\theta(\cdot,x) \right] = \tilde{\rho}_\theta(x) \equiv \rho$ is equivalent to the asymptotic nonexistence of the average shift.
    However, it is difficult to know the explicit expression of $\pi_\theta$. Based on the equation $\pi_\theta=\pi_\theta P_\theta$ alone, we cannot determine whether $\pi_\theta \left[ g_\theta(\cdot, x) \right]$ deviates from $\rho$. Furthermore, we cannot tell if the allocation ratio $\frac{1}{n}\sum_{i=1}^n \mathbb{E}T_i$ equals the targeted allocation ratio $\rho$. In Liu, Hu and Ma's article (\cite{liuPropertiesCovariateadaptiveRandomization2025}), they only show the phenomenon numerically that $\frac{1}{n}\sum_{i=1}^n \mathbb{E}T_i$ deviates from $\rho$.
    Due to these difficulties, we do not attempt to prove the asymptotic existence of the average shift for any naturally occurring or pre-specified additional covariate.
    Rather, we construct a deliberately designed example of the additional covariate $Y$ in the proof of Theorem \ref{theorem_shift_liu_xiao}.
\end{remark}

\subsection{Balance Discrete Covariates}

Consider the case where the baseline covariate $X_{\mathrm{origin}} = (X_{1,\mathrm{origin}}, \dots, X_{p,\mathrm{origin}})$ is discrete.
Suppose that each component $X_{t,\mathrm{origin}}$ takes finitely many values $x_t^{(l_t)}$, $l_t = 1, \dots, L_t$, $t = 1, \dots, p$.

Some classical randomization procedures for balancing discrete covariates can be readily extended to accommodate any targeted allocation ratio.
Their randomization procedures can be viewed as special cases of the randomization procedure described in Subsection \ref{subsec_main_liu2024}, and the corresponding allocation functions can all be expressed in the form of \eqref{eq_allocation_function_minimization}.
Thus, we will talk about these procedures as examples.

\begin{example}[Stratified Randomization Procedure]
    Define the covariate transformation
    \begin{equation*}
        \phi_{\mathrm{S}}(X_{\mathrm{origin}}) = (
        \underbrace{ \dots, \mathbb{I}(X_{\mathrm{origin}} = (x_1^{(l_1)}, \dots, x_p^{(l_p)})), \dots }_{\prod L_t \text { within-stratum terms }}
        )^T ,
    \end{equation*}
    and let $X = \phi_{\mathrm{S}}(X_{\mathrm{origin}})$.  
    Under the general CAR framework, this corresponds to the classical stratified randomization procedure, which achieves exact balance within each stratum defined by $X_{\mathrm{origin}}$.  
    By Theorem \ref{theorem_shift_liu_xiao}, the within-stratum imbalance vector $\sum_{i=1}^n (T_i - \rho) \phi_{\mathrm{S}}(X_{\mathrm{origin},i})$ is $O_P(1)$.
\end{example}

\begin{example}[Pocock and Simon's Minimization Procedure]
    Define
    \begin{equation*}
        \phi_{\mathrm{PS}}(X_{\mathrm{origin}}) = (
        \underbrace{ \dots, \sqrt{w_{m, t}} \mathbb{I}(X_{t,\mathrm{origin}} = x_t^{(l_t)}), \dots}_{\sum L_t \text { marginal terms }}
        )^T ,
    \end{equation*}
    where $w_{m,t} > 0$ are pre-specified weights corresponding to each covariate $X_{t,\mathrm{origin}}$, and let $X = \phi_{\mathrm{PS}}(X_{\mathrm{origin}})$.  
    When the targeted allocation ratio $\rho = \frac{1}{2}$, This specification corresponds to Pocock and Simon's minimization procedure in (\cite{pocockSequentialTreatmentAssignment1975}).
    Theorem \ref{theorem_shift_liu_xiao} implies that $\sum_{i=1}^n (T_i - \rho) \phi_{\mathrm{PS}}(X_{\mathrm{origin},i}) = O_P(1)$.
    Consequently, the marginal imbalances are bounded in probability.
\end{example}

\begin{example}[Hu and Hu's Generalized Procedure]
    Define
    \begin{align*}
        \phi_{\mathrm{HH}}(X_{\mathrm{origin}})
        &= ( \sqrt{w_o},
        \underbrace{ \dots, \sqrt{w_{m, t}} \mathbb{I}(X_{t,\mathrm{origin}} = x_t^{(l_t)}), \dots}_{\sum L_t \text { marginal terms }}, \\
        & \quad \underbrace{ \dots, \sqrt{w_s} \mathbb{I}(X_{\mathrm{origin}} = (x_1^{(l_1)}, \dots, x_p^{(l_p)})), \dots }_{\prod L_t \text { within-stratum terms }}
        )^T ,
    \end{align*}
    and let $X = \phi_{\mathrm{HH}}(X_{\mathrm{origin}})$.  
    This corresponds to the generalized covariate adaptive randomization procedure proposed by (\cite{huAsymptoticPropertiesCovariateadaptive2012}), which combines overall, marginal, and stratified imbalance measures through weights $w_o, w_{m,t}, w_s \geq 0$.  
    By Theorem \ref{theorem_shift_liu_xiao}, $\sum_{i=1}^n (T_i - \rho)X_i = O_P(1)$, implying that the overall, marginal, and within-stratum imbalances are all bounded in probability.
\end{example}

\section{New Randomization Procedures Addressing Shift Problems}
\label{sec_addressing_shift_problem}

\subsection{New Allocation Function Form}

\label{subsec_new_allocation_function_form}

Suppose the parameter sequence of the randomization procedure is $\{\theta_n = \theta\}_{n \in \mathbb{N}}$.
In this case, $\{ \Lambda_n \}$ forms a Markov chain.
Let $\pi_\theta$ be the corresponding invariant probability.
Recall that $\tilde{\rho}_\theta(x) = \pi_\theta \left[ g_\theta(\cdot, x) \right]$ as defined in Subsection \ref{subsec_shift_properties_fixed}.
The relationship $\pi_\theta \left[ \mathbb{E}_\theta \left[ \Lambda_1 \mid \Lambda_0 = \cdot \right] \right] = \pi_\theta \left[ P_\theta(\cdot, \Lambda) \right] = \pi_\theta (\Lambda)$ implies that
\begin{equation*}
    0
    = \pi_\theta \left[ P_\theta(\cdot, \Lambda) - \Lambda  \right]
    = \mathbb{E}_{X \sim \Gamma} \left[ (\tilde{\rho}_\theta(X) - \rho) X \right] .
\end{equation*}

Let $\mathcal{C}$ denote the function set
\begin{equation*}
    \mathcal{C} = \{ \mu \left[ g_\theta(\cdot, x) \right] \mid \mu \text{ is a probability measure on } \mathbb{R}^d \} ,
\end{equation*}
which corresponds to the convex hull of the allocation functions $\{ g_{\theta,\Lambda}(x) \mid \Lambda \in \mathbb{R}^d \}$, where $g_{\theta,\Lambda}(x) := g_\theta(\Lambda, x)$.
Define the mapping $\Phi: \mathcal{C} \to \mathbb{R}^d$ by
\begin{equation*}
    \Phi\left(\tilde{\rho}(\cdot)\right)
    := \mathbb{E}_{X \sim \Gamma} \left[ (\tilde{\rho}(X) - \rho) X \right] ,
\end{equation*}
where $\tilde{\rho} \in \mathcal{C}$.
We now introduce the core argument for establishing $\tilde{\rho}(x) \equiv \rho$.
\begin{lemma}
    \label{lemma_Phi}
    Suppose the Markov chain $\{\Lambda_n\}$ is positive recurrent with the invariant probability $\pi_\theta$, the expectation $\pi_\theta(\Lambda) < \infty$ and the covariate $X$ has a finite expectation.
    If the following two conditions hold:
    \begin{enumerate}[label=(\Alph*)]
        \item $\Phi$ is injective,
        \item the function $\tilde{\rho}$, defined by $\tilde{\rho}(x) \equiv \rho$, belongs to the set $\mathcal{C}$,
    \end{enumerate}
    then $\tilde{\rho}_\theta(x) \equiv \rho$ for any $x$.
\end{lemma}
\begin{proof}
    The positive recurrence implies the existence of a probability measure $\pi_\theta$ satisfying $\pi_\theta P_\theta = \pi_\theta$.
    It implies that $\pi_\theta \left[ P_\theta(\cdot, \Lambda) \right] = \pi_\theta (\Lambda)$, which is equivalent to $\Phi\left(\tilde{\rho}_\theta(\cdot)\right) = 0$.
    Because $\Phi$ is injective, and the function $\tilde{\rho}$, defined by $\tilde{\rho}(x) \equiv \rho$, belongs to $\tilde{\rho} \in \mathcal{C}$ with $\Phi\left(\tilde{\rho}(\cdot)\right) = 0$, it follows that $\tilde{\rho}_\theta = \tilde{\rho}$, namely $\tilde{\rho}_\theta(x) \equiv \rho$.
\end{proof}

Based on the stronger equivalent form in Remark \ref{remark_stronger_equivalent_form}, the asymptotic nonexistence of the average shift follows from the conclusion in Lemma \ref{lemma_Phi}, namely $\tilde{\rho}_\theta(x) = \pi_\theta \left[ g_\theta(\cdot, x) \right] = \rho$ for all $x$.
In addition, the condition in Lemma \ref{lemma_Phi} does not involve the explicit form of $\pi_\theta$, making it easier to verify than directly checking whether $\pi_\theta h_\theta = 0$.

The most restrictive condition in Lemma \ref{lemma_Phi} is the injectivity of $\Phi$.
This condition is not satisfied by the randomization procedure in Subsection \ref{subsec_main_liu2024}. To address this issue, we aim to develop a new form of allocation function that satisfies this condition.

Define the mapping $\Phi_c: \mathcal{C} - \rho \to \mathbb{R}^d$ by
\begin{equation*}
    \Phi_c(\rho_c(\cdot)) := \Phi(\rho_c(\cdot)+\rho) .
\end{equation*}
Since $\Phi_c$ is linear and maps into $\mathbb{R}^d$, its injectivity implies that the dimension of the space $\mathcal{C}-\rho$ is at most $d$.

For simplicity, we assume that $\mathcal{C}-\rho$ is contained in the linear span of a set of linearly independent functions $\{\alpha_i(x) \mid i \in \{1, \dots, d^\prime\}\}$, where $d^\prime \leq d$.
Recall that $\mathcal{C} = \{ \mu \left[ g_\theta(\cdot, x) \right] \mid \mu \text{ is a probability measure on } \mathbb{R}^d \}$.
By taking $\mu$ to be the Dirac measure at $\Lambda$, we obtain a basic element $g_\theta(\Lambda,x) \in \mathcal{C}$.
Thus, it can be shown that $g_\theta(\Lambda, x) - \rho = \sum_{i=1}^{d^\prime} \beta_i(\Lambda) \alpha_i(x)$, which is equivalent to
\begin{equation*}
    g_\theta(\Lambda, x) = \rho + \sum_{i=1}^{d^\prime} \beta_i(\Lambda) \alpha_i(x) .
\end{equation*}
Since $\mathcal{C} \subset \{\rho + \sum_{i=1}^{d^\prime} c_i \alpha_i(x) \mid c_i \in \mathbb{R}\}$, we consider a generic element $\tilde{\rho} \in \mathcal{C}$ of the form $\tilde{\rho}(x) = \rho + \sum_{i=1}^{d^\prime} c_i \alpha_i(x)$. Under this setting,
\begin{equation*}
    \Phi\left(\tilde{\rho}(\cdot)\right)
    = \mathbb{E}_{X \sim \Gamma} \left[ (\tilde{\rho}(X) - \rho) X \right]
    = \sum_{i=1}^{d^\prime} c_i \mathbb{E}_{X \sim \Gamma} [\alpha_i(X)X] .
\end{equation*}
The linear independence of $\left\{ \mathbb{E}_{X \sim \Gamma} [\alpha_i(X)X] \mid 1 \leq i \leq d^\prime \right\}$ implies the injectivity of the mapping $\Phi$.

To establish properties in Subsection \ref{subsec_boundedness_in_probability}--\ref{subsec_shift_properties_convergent}, we need the condition $\theta \in K_{M,\Delta}$.
The condition $\theta \in K_{M,\Delta}$ is equivalent to
\begin{align}
    & \quad \mathbb{E}_\theta \left[ (\Lambda_1 - \Lambda_0)^T \frac{\Lambda_0}{\|\Lambda_0\|} \mid \Lambda_0 = \Lambda \right]
    = \mathbb{E}_{X \sim \Gamma} [(g_\theta(\Lambda, X) - \rho)X]^T \frac{\Lambda}{\|\Lambda\|} \label{eq_negative_condition_new_form} \\
    &= \sum_{i=1}^{d^\prime} \beta_i(\Lambda) \mathbb{E}_{X \sim \Gamma} [\alpha_i(X)X]^T \frac{\Lambda}{\|\Lambda\|}
    \leq - \Delta < 0 \notag ,
\end{align}
for any $\Lambda \in W_\Gamma$ with $\|\Lambda\| \geq M$.
Under Assumption \ref{assumption_x_spread_out}, $W_\Gamma = \mathbb{R}^d$ is guaranteed by Lemma A.1 in the Supplementary Material.
Thus, it is shown that for any $\Lambda \neq 0$, there exists $i \in \{1, \dots, d^\prime\}$ such that $\mathbb{E}_{X \sim \Gamma} [\alpha_i(X)X]^T \Lambda \neq 0$.
It further implies that $\{\mathbb{E}_{X \sim \Gamma} [\alpha_i(X)X] \mid i \in \{1, \dots, d^\prime\}\}$ can span the whole space $\mathbb{R}^d$.
This shows that $d^\prime \geq d$. In summary, we conclude that $d^\prime = d$.

To ensure that the allocation function satisfies (\ref{eq_negative_condition_new_form}) for all $\Lambda$ with sufficiently large $\|\Lambda\|$, we fix the function $\alpha_i$ and seek to make $\beta_i(\Lambda) \mathbb{E}_{X \sim \Gamma} [\alpha_i(X)X]^T \frac{\Lambda}{\|\Lambda\|}$ as small as possible for each $i$.
Assuming $\alpha_i, \beta_i \in [-1,1]$, the optimal choice is $\beta_i(\Lambda) = -\sgn(\mathbb{E}_{X \sim \Gamma} [\alpha_i(X)X]^T \Lambda)$.
However, the quantity $\mathbb{E}_{X \sim \Gamma} [\alpha_i(X)X]$ is unknown and must be estimated during the randomization procedure.
To mitigate the effect of estimation error, we adopt a smooth surrogate for $\beta_i$, as illustrated in (\ref{eq_oracle_allocation}) and Assumption \ref{assumption_theta_basic} in the next subsection.

\subsection{Oracle Randomization Procedure}

\label{subsec_main_oracle}

In this subsection, we explore a special class of randomization procedure with fixed parameters.
Following the discussion in Subsection \ref{subsec_new_allocation_function_form}, we set the allocation function $g_\theta$ to
\begin{equation}
    g_\theta(\Lambda, X) = \rho + \frac{p}{d} \sum_{i=1}^d { \alpha_i(X) \beta_{\xi_i}^{\epsilon(\theta)}(\Lambda) }
    \label{eq_oracle_allocation} ,
\end{equation}
where $0 < p < \min(\rho,1-\rho)$.
Here, $\alpha_i(X)$, $\beta_{\xi}^{\epsilon}(\Lambda)$, and $\epsilon(\theta)$ are functions taking values in $[-1, 1]$, $[-1, 1]$, and $[0, 1)$, respectively. The parameter $\theta = (\xi_1, \dots, \xi_d)$ is a matrix, and its norm is measured using the Frobenius norm.
Furthermore, the component $\xi_i$ in the parameter is a $d$-dimensional random vector endowed with $\ell_2$-norm. Here, we remind that $d$ is the dimension of the covariate $X$.

Define the oracle value of the parameter as $\theta^* = (\xi_1^*, \dots, \xi_d^*)$, where $\xi_i^* = \mathbb{E}_{X \sim \Gamma} \left[ \alpha_i(X) X \right]$. When the parameter $\theta$ is fixed near $\theta^*$, the randomization procedure under the fixed allocation function $g_\theta$ is also referred to as the oracle randomization procedure, because its properties are similar to the procedure under $g_{\theta^*}$.

Fixing the imbalance vector $\Lambda$, the allocation probability $g_\theta(\Lambda, x)$ can be viewed as an adjustment of the baseline probability $\rho$ via a weighted sum of $\{\alpha_i(x)\}$.
The function $\beta_{\xi_i}^{\epsilon(\theta)}(\Lambda)$ serves as a smooth surrogate for $-\sgn(\xi_i^T \Lambda)$, where the parameter $\epsilon = \epsilon(\theta)$ controls the degree of closeness between the surrogate and the original function.
For these three functions, additional assumptions are needed to ensure that the negative feedback condition $\theta \in K_{M,\Delta}$ holds as long as $\theta$ is sufficiently close to $\theta^*$.
\begin{assumption}
    \label{assumption_theta_basic}
    For $\alpha_i(X)$, $\beta_{\xi}^{\epsilon}(\Lambda)$ and $\epsilon(\theta)$,
    \begin{enumerate}
        \item \label{item_assumption_theta_basic_1} The vectors $\left\{ \mathbb{E}_{X \sim \Gamma} [\alpha_i(X)X] \mid 1 \leq i \leq d \right\}$ are linearly independent.
        \item \label{item_assumption_theta_basic_2} $\beta_{\xi}^{\epsilon}(\Lambda) \leq 0$ when $\xi^T \Lambda \geq 0$.
        \item \label{item_assumption_theta_basic_3} $\beta_{\xi}^{\epsilon} = \beta_{c\xi}^{\epsilon}$ for any $c > 0$ and $\beta_{\xi}^{\epsilon}(\Lambda) = -\beta_{-\xi}^{\epsilon}(\Lambda) = -\beta_{\xi}^{\epsilon}(-\Lambda)$.
        \item \label{item_assumption_theta_basic_4} For any $d_{\mathrm{uni}} > 0$, there exists $M > 0$ such that for any $\epsilon \geq d_{\mathrm{uni}}$, $\epsilon^\prime \in (0,\epsilon]$ and $\Lambda \in R_\xi^\epsilon$ with $\|\Lambda\| \geq M$, $\beta_{\xi}^{\epsilon^\prime}(\Lambda) = -1$, where $R_\xi^\epsilon = \left\{ \Lambda \neq 0 \mid \cossimilarity(\Lambda,\xi) > \epsilon \right\}$ is an open convex cone and $\cossimilarity(\Lambda,\xi) = \frac{\Lambda \xi}{\|\Lambda\| \|\xi\|}$.
        \item \label{item_assumption_theta_basic_5} $\epsilon(\theta = (\xi_1, \dots, \xi_d)) = \epsilon(\theta = (c_1 \xi_1, \dots, c_d \xi_d))$ for any $c_i > 0$.
        \item \label{item_assumption_theta_basic_6} If $\theta$ is nonsingular, then $\epsilon(\theta) > 0$ and $\bigcup_{i=1}^{d} \left[ R_{\xi_i}^{\epsilon(\theta)} \cup -R_{\xi_i}^{\epsilon(\theta)} \right] = \mathbb{R}^d \backslash \{0\}$. If $\theta$ is singular, $\epsilon(\theta) = 0$.
    \end{enumerate}
\end{assumption}

\begin{remark}

    We set the parameter $\xi_i$ to an estimate of $\mathbb{E}_{X \sim \Gamma} [\alpha_i(X)X]$.
The sign of $\beta$ indicates whether the allocation probability should be adjusted according to the scheme $\alpha_i$ or in the opposite direction.
    Specifically, we want $\beta$ to be positive when the angle between $\Lambda$ and $\mathbb{E}_{X \sim \Gamma} [\alpha_i(X)X]$ is obtuse, and to be negative when the angle is acute. Because the expectation $\mathbb{E}_{X \sim \Gamma} [\alpha_i(X)X]$ is unknown, we use the angle between $\Lambda$ and the estimate $\xi$ instead.
    Thus, we only care about the direction of $\xi$. A positive multiple of $\xi$ does not affect the value of $\beta$ and $\epsilon$.
    Since $\xi$ estimates $\mathbb{E}_{X \sim \Gamma} [\alpha_i(X)X]$, the function $\beta$ takes extreme values $\pm 1$ when $\|\Lambda\|$ is sufficiently large and the angle between $\Lambda$ and $\xi$ deviates from orthogonality.
    The parameter $\epsilon = \epsilon(\theta)$ represents the threshold angle at which the function $\beta$ switches to extreme values.
Moreover, the condition $\bigcup_{i=1}^{d} \left[ R_{\xi_i}^{\epsilon(\theta)} \cup -R_{\xi_i}^{\epsilon(\theta)} \right] = \mathbb{R}^d \backslash \{0\}$ imposes an upper bound on $\epsilon = \epsilon(\theta)$, ensuring that at least one of functions $\{\beta_i\}$ takes extreme values.
    This restriction prevents the functions $\{\beta_i\}$ from being so tolerant to errors that they take values close to zero.
\end{remark}

The assumption allows some error in the estimation of $\theta^*$. This consideration makes the assumption seem slightly more complicated.
We illustrate how Assumption \ref{assumption_theta_basic} operates by considering the case where the parameter is fixed at $\theta = \theta^*$.
\begin{example}[On the occasion of the oracle value]
    \label{example_oracle_value}

    When the parameter $\theta$ is exactly the oracle value $\theta^*$, We can concisely explain why the oracle randomization procedure can control the imbalance vector.
    In this case, the expectation in the definition of $K_{M,\Delta}$ is
    \begin{align}
        & \quad \mathbb{E}_{\theta^*} \left[ (\Lambda_1 - \Lambda_0)^T \frac{\Lambda_0}{\|\Lambda_0\|} \mid \Lambda_0 = \Lambda \right]
        = \frac{p}{d}\sum_{i=1}^d \beta_{\xi_i^*}^{\epsilon(\theta^*)}(\Lambda) \frac{\mathbb{E}_{X \sim \Gamma} [\alpha_i(X)X]^T\Lambda}{\|\Lambda\|} \label{eq_negative_feedback_condition_oracle_randomization} \\
        &= \frac{p}{d}\sum_{i=1}^d \beta_{\xi_i^*}^{\epsilon(\theta^*)}(\Lambda) \frac{{\xi_i^*}^T\Lambda}{\|\Lambda\|}
        \leq -\frac{p}{d}\epsilon(\theta^*)\min|\xi_i^*|
        < 0 , \notag
    \end{align}
    where $\|\Lambda\| \geq M$, $M$ is the constant in Assumption \ref{assumption_theta_basic}.
    The inequality (\ref{eq_negative_feedback_condition_oracle_randomization}) implies that $\theta^* \in K_{M,\Delta}$ with the constant $\epsilon=\frac{p}{d}\epsilon(\theta^*)\min|\xi_i^*|$.
\end{example}

Furthermore, when the parameter $\theta$ is sufficiently close to $\theta^*$, the randomization procedure under $\theta$ still shares the same properties as the randomization procedure under the fixed parameter sequence $\{\theta_n = \theta^*\}_{n \in \mathbb{N}}$.
To ensure this, it suffices to verify that $\overline{B(\theta^*, r_d)} \subset K_{M,\Delta}$ for some $r_d > 0$.
As a consequence of $\theta \in \overline{B(\theta^*,r_d)} \subset K_{M,\Delta}$, we have the following theorem:
\begin{theorem}
    \label{theorem_shift_simultaneous_geometric_finite_adjusted}

    Suppose Assumption \ref{assumption_theta_basic} for components of the allocation function holds. If Assumption \ref{assumption_x_sub_exponential_bound} holds, then there exists $r_d>0$ such that $\overline{B(\theta^*,r_d)} \subset K_{M,\Delta}$ for some $M>0$ and $\Delta>0$.
In addition, if Assumption \ref{assumption_x_spread_out} holds, then $P_\theta$ is positive recurrent with the invariant probability $\pi_\theta$, and the actual targeted allocation ratio $\tilde{\rho}_\theta(x) = \pi_\theta \left[ g_\theta(\cdot,x) \right] = \rho$.
Under the additional assumption that $Y$ has a finite expectation, we further have $\pi_\theta h_\theta = 0$ and
    \begin{equation*}
        \left| \mathbb{E}_\theta \left[ \sum_{i=1}^n (T_i-\rho)Y_i \mid \Lambda_0 = \Lambda \right] \right| \leq L^2 \left[\mathbb{E}|Y|\right] V(\Lambda) < \infty .
    \end{equation*}
\end{theorem}

A notable consequence of the corollary is that under the oracle randomization procedure, the average shift on the additional covariate converges to zero. Even if the parameter does not exactly equal the oracle value, the average shift may still converge to zero.

\subsection{Feasible Randomization Procedure}

\label{subsec_main_feasible}

In Subsection \ref{subsec_main_oracle}, we set the allocation function $g_\theta$ to
\begin{equation*}
    g_\theta(\Lambda, X) = \rho + \frac{p}{d} \sum_{i=1}^d { \alpha_i(X) \beta_{\xi_i}^{\epsilon(\theta)}(\Lambda) } ,
\end{equation*}
where $0 < p < \min(\rho,1-\rho)$, $\alpha_i(X) \in [-1,1]$, $\beta_{\xi}^{\epsilon}(\Lambda) \in [-1,1]$ and $\epsilon(\theta) \in (0,1)$.
These component functions should satisfy Assumption \ref{assumption_theta_basic} for the considerations of the asymptotic nonexistence of the average shift under the oracle randomization procedure.

In Subsection \ref{subsec_shift_properties_convergent}, we have mentioned that the allocation function $g_\theta$ and the parameter sequence $\{\theta_n\}$ should satisfy Assumptions \ref{assumption_g_lipschitz} and \ref{assumption_theta_convergence}, respectively. They guarantee parameter updates will not affect the center of the asymptotic distribution.
A randomization procedure is called feasible if Assumptions \ref{assumption_g_lipschitz} and \ref{assumption_theta_convergence} hold, and its allocation function takes the form \eqref{eq_oracle_allocation}.

In this subsection, we will give some examples of the components of the allocation function and the parameter sequence which satisfy the assumptions above. First, we give a sufficient condition for Assumption \ref{assumption_g_lipschitz}:
\begin{assumption}
    \label{assumption_theta_lipschitz}
    For $\alpha_i(X)$, $\beta_{\xi}^{\epsilon}(\Lambda)$ and $\epsilon(\theta)$,
    \begin{enumerate}
        \item \label{item_assumption_theta_lipschitz_1} $\epsilon(\theta)$ is Lipschitz continuous on $\left\{ \theta \mid \forall i, \|\xi_i\| \geq d_{\mathrm{uni}} \right\}$ for any $d_{\mathrm{uni}}>0$.
        \item \label{item_assumption_theta_lipschitz_2} $\beta_{\xi}^{\epsilon}(\Lambda)$ is Lipschitz continuous on $\left\{ (\xi,\epsilon,\Lambda) \mid \|\xi\| \geq d_{\mathrm{uni}}, \epsilon \geq d_{\mathrm{uni}} \right\}$ for any $d_{\mathrm{uni}}>0$.
    \end{enumerate}
\end{assumption}

Under the assumptions on $\alpha$, $\beta$ and $\epsilon$ functions, we can establish the following lemma to justify Assumption \ref{assumption_g_lipschitz}.
\begin{lemma}
    \label{lemma_g_lipschitz}
    If Assumptions \ref{assumption_theta_basic} and \ref{assumption_theta_lipschitz} hold, then Assumption \ref{assumption_g_lipschitz} holds.
\end{lemma}

Thus, for Lemma \ref{lemma_g_lipschitz}, we need to give an example of $\alpha_i(X)$, $\epsilon(\theta)$ and $\beta_{\xi}^{\epsilon}(\Lambda)$ satisfying Assumptions \ref{assumption_theta_basic} and \ref{assumption_theta_lipschitz}.

First, we can set
\begin{equation}
    \alpha_i(X) = \mathbb{I}(X_i > 0) - \mathbb{I}(X_i < 0) \label{eq_allocation_alpha} ,
\end{equation}
where $X_i$ is the $i$th component of $X$.
A conservative alternative choice of $\alpha_i$ is
\begin{equation*}
    \alpha_i(X) = X_i / \|X\|_\infty ,
\end{equation*}
which guarantees the linear independence of $\left\{ \mathbb{E}_{X \sim \Gamma} [\alpha_i(X)X] \mid 1 \leq i \leq d \right\}$ under Assumption \ref{assumption_x_spread_out}.

Second, for the choice of the function $\beta$, we set the function $\beta_{\xi}^{\epsilon}(\Lambda)$ by a composition of several functions
\begin{equation*}
    \cutsin(x) =
    \begin{cases}
        -\sin(x), & \text{if } x \in [-\frac{\pi}{2},\frac{\pi}{2}] \text{,} \\
        -\sgn(x), & \text{otherwise,}
    \end{cases}
    \quad \text{and} \quad
    \tau_{\xi}^{\epsilon}(\Lambda) = \frac{\sqrt{1+\epsilon^2}\frac{\xi^T}{\|\xi\|}\Lambda} {\sqrt{1+\epsilon^2\|\Lambda\|^2}} ,
\end{equation*}
such that
\begin{equation}
    \beta_{\xi}^{\epsilon}(\Lambda) = \cutsin\left( \frac{\pi}{2} \tau_{\xi}^{\epsilon}(\Lambda) \right) \label{eq_allocation_beta} .
\end{equation}

Third, we define the matrix $A(\theta) = \left( \frac{\xi_1}{\|\xi_1\|}, \dots, \frac{\xi_d}{\|\xi_d\|} \right)$ and set
\begin{equation}
    \epsilon(\theta) =
    \begin{cases}
        \frac{1}{\sqrt{d+1} \|A(\theta)^{-1}\|_2}, & \text{if the vectors } \{\xi_1, \dots, \xi_d\} \text{ are linearly independent,} \\
        0, & \text{otherwise,}
    \end{cases}
    \label{eq_allocation_epsilon}
\end{equation}
where $\|A(\theta)^{-1}\|_2$ is the operator norm of the matrix $A(\theta)^{-1}$.

For these choices of functions, we have the following conclusions:
\begin{lemma}
    \label{lemma_instance_established}
    When the components of the allocation function are (\ref{eq_allocation_alpha}), (\ref{eq_allocation_beta}) and (\ref{eq_allocation_epsilon}), if the distribution $\Gamma$ satisfies that $\left\{ \mathbb{E}_{X \sim \Gamma} [\alpha_i(X)X] \mid 1 \leq i \leq d \right\}$ are linearly independent, then Assumptions \ref{assumption_theta_basic} and \ref{assumption_theta_lipschitz} hold.
\end{lemma}
Combining with Lemma \ref{lemma_g_lipschitz}, we have that Assumption \ref{assumption_g_lipschitz} holds.
In the end, for the parameter sequence, we select
\begin{equation}
    \theta_n = (\xi_{n,1},\dots,\xi_{n,d}) = \left( \frac{1}{n}\sum_{j=1}^n \alpha_1(X_j)X_j, \dots, \frac{1}{n}\sum_{j=1}^n \alpha_d(X_j)X_j \right) \label{eq_allocation_parameter} .
\end{equation}
Then we have the lemma below to establish Assumption \ref{assumption_theta_convergence}:
\begin{lemma}
    \label{lemma_average_theta_check}
    Suppose the parameter $\theta_n$ is defined as in \eqref{eq_allocation_parameter}. If Assumption \ref{assumption_x_sub_exponential_bound} holds, then Assumption \ref{assumption_theta_convergence} is satisfied with $\theta^* = (\xi_1^*,\dots,\xi_d^*)$, where $\xi_i^*=\mathbb{E}_{X \sim \Gamma} [\alpha_i(X)X]$.
\end{lemma}

\begin{remark}
    Although Lemma \ref{lemma_average_theta_check} shows that the parameter sequence (\ref{eq_allocation_parameter}) satisfies Assumption \ref{assumption_theta_convergence}, the estimator $\theta_n$ may be a poor approximation to $\theta^*$ when $n$ is small.
    Thus, we suggest that in practice, when the number of allocated units is below a certain threshold, using $\rho$ instead of $g_{\theta_n}(\Lambda_n, X_{n+1})$ as the allocation probability is preferable.
    In the simulation (Sections A--B of the Supplementary Material), the threshold is set to 10.
\end{remark}

In summary, we can set the allocation function to (\ref{eq_oracle_allocation}) with the components (\ref{eq_allocation_alpha}), (\ref{eq_allocation_beta}) and (\ref{eq_allocation_epsilon}). Moreover, the parameter sequence can be chosen according to (\ref{eq_allocation_parameter}).

\subsection{Properties of the Feasible Randomization Procedure}

\label{subsec_randomization_procedure_convergent_parameter}

\begin{corollary}
    \label{corollary_boundedness_convergent}

    Under the allocation function (\ref{eq_oracle_allocation}) with the components (\ref{eq_allocation_alpha}), (\ref{eq_allocation_beta}) and (\ref{eq_allocation_epsilon}), suppose the parameter sequence is (\ref{eq_allocation_parameter}).

    Suppose Assumption \ref{assumption_x_sub_exponential_bound} holds and $\left\{ \mathbb{E}_{X \sim \Gamma} [\alpha_i(X)X] \mid 1 \leq i \leq d \right\}$ are linearly independent. Then the stochastic process $\{\Lambda_n\}$ is bounded in probability.
\end{corollary}

The corollary implies that the imbalance vector is $O_P(1)$ under the feasible randomization procedure, regardless of the covariate type.
Furthermore, under Assumption \ref{assumption_x_spread_out} on the spread-out property of the covariate distribution, we can show that under the feasible randomization procedure, the center of the asymptotic distribution of $\sum_{i=1}^n (T_i-\rho)Y_i$ is zero for any additional covariate $Y$.
\begin{corollary}
    \label{corollary_asymptotic_normality}

    Under the allocation function (\ref{eq_oracle_allocation}) with the components (\ref{eq_allocation_alpha}), (\ref{eq_allocation_beta}) and (\ref{eq_allocation_epsilon}), suppose the parameter sequence is (\ref{eq_allocation_parameter}).
    If $Y$ has a finite second moment, $\left\{ \mathbb{E}_{X \sim \Gamma} [\alpha_i(X)X] \mid 1 \leq i \leq d \right\}$ are linearly independent, and Assumptions \ref{assumption_x_sub_exponential_bound} and \ref{assumption_x_spread_out} hold, then there is a $\sigma^*_Y \geq 0$ such that
    \begin{equation*}
        \frac{1}{\sqrt{N}} \sum_{n=1}^{N} { \left[ (T_n-\rho) Y_n \right] } \xrightarrow{d} \mathcal{N}(0, {\sigma^*_Y}^2) ,
    \end{equation*}
    where ${\sigma^*_Y}^2$ is equal to the asymptotic variance under the fixed parameter $\theta^*$.
\end{corollary}

By Corollaries \ref{corollary_boundedness_convergent} and \ref{corollary_asymptotic_normality}, we can conclude that under the feasible randomization procedure, the imbalance vector $\Lambda_n = \sum_{i=1}^n (T_i-\rho)X_i = O_P(1)$, and the shift problem is addressed for any additional covariate $Y$.
Therefore, in contrast to existing procedures that may deviate from the desired allocation ratio, the feasible randomization procedure achieves covariate balance while preserving the targeted allocation ratio $\rho$ in the long run.

\section{Experiments}
\label{sec_experiments}

In this section, we present the simulation results of several randomization procedures: the complete randomization (CR), the randomization procedure minimizing the imbalance measure (RMM, implemented as MIX(0,0,1,0) in Liu, Hu and Ma (\cite{liuPropertiesCovariateadaptiveRandomization2025})), the feasible randomization procedure (FR) and the oracle randomization procedure (OR).
In this simulation, we do not incorporate the feature map. The randomization procedure is intended to balance the first moment of the covariate vector $X$.
The simulation is to evaluate the convergence rate of the imbalance vector and the average shift on some additional covariates.

Consider the case of the unequal targeted allocation ratio $\rho=2/3$. The covariate vector $X$ is a $3$-dimensional random vector $(A+B,B,C)$, where $A$, $B$, $C$ are independent, the first two are standard normal random variables and the last is the standard exponential random variable.
The procedure minimizing the imbalance measure we implement is the procedure in Subsection \ref{subsec_main_liu2024} with the biased coin probability $\rho_1 = 0.9$.
The components of the allocation function in the last two randomization procedure are the same as the setting in Subsection \ref{subsec_main_feasible} and the adjusted probability $p$ in the allocation function is set to $0.2$.
For the feasible randomization procedure, we set the allocation probability to $\rho$ exactly when the number of allocated units $n$ is less than the threshold $10$.
For the oracle randomization procedure, we set $\epsilon \equiv 0$ since the oracle value has already been assigned to the parameter.

The additional covariate vector we test is set to $\left( \sqrt{\sum_{i=1}^3 |X_i|}, \sum_{i=1}^3 X_i^2 \right)$, where $X_i$ denotes the $i$th component of $X$ here.
The sample sizes are $n=200, 400, 800, 1600, 3200$. This setting aims to illustrate the relationship between the shift and the sample size $n$. For each case, we conduct the simulation for $N=10000$ times.
The means and standard deviations are used to evaluate the balancing performance of imbalance vectors and shifts. The simulation results are given in Table \ref{table_balance_shift}.
Note that the shift refers to the expectation of the imbalance vector associated with the additional covariate, and the average shift is the shift divided by the sample size $n$.

\begin{table}
\centering
\caption{Means (standard deviations) of covariate imbalances under various randomization procedures. Each case is based on 10000 replicates}
\label{table_balance_shift}
\begin{adjustbox}{scale=0.7}
\begin{threeparttable}
    \begin{tabular}{rrrrrrrr}
        \toprule
        Randomization & n & $\sum_k (T_k-\rho)X_{k,1}$ & $\sum_k (T_k-\rho)X_{k,2}$ & $\sum_k (T_k-\rho)X_{k,3}$ & $\sum_k (T_k-\rho)\sqrt{\sum_i|X_{k,i}|}$ & $\sum_k (T_k-\rho)\sum_i X_{k,i}^2$ \\
        \hline
        CR & 200 & 0.02(9.39) & 0.05(6.67) & -0.13(9.50) & -0.27(11.36) & -0.98(51.36) \\
        & 400 & -0.21(13.15) & -0.18(9.36) & -0.07(13.27) & -0.10(15.98) & -0.14(72.03) \\
        & 800 & -0.06(18.78) & -0.05(13.31) & -0.25(18.96) & -0.28(22.73) & -1.66(101.91) \\
        & 1600 & 0.38(26.57) & 0.28(18.74) & -0.45(26.55) & -0.53(32.29) & -2.53(143.75) \\
        & 3200 & 0.72(37.65) & 0.31(26.49) & -0.21(37.76) & -0.73(45.82) & -1.48(204.82) \\
        \hline
        RMM & 200 & 0.01(0.98) & 0.02(0.90) & -0.22(1.17) & -8.51(7.41) & 15.00(30.01) \\
        & 400 & -0.00(0.98) & -0.00(0.90) & -0.24(1.20) & -17.06(10.49) & 30.77(42.36) \\
        & 800 & 0.01(0.99) & 0.01(0.90) & -0.21(1.18) & -33.88(14.60) & 60.69(59.30) \\
        & 1600 & -0.00(1.00) & 0.01(0.89) & -0.25(1.21) & -67.18(21.07) & 122.26(82.97) \\
        & 3200 & 0.00(1.01) & 0.00(0.88) & -0.23(1.18) & -134.47(29.88) & 243.73(118.75) \\
        \hline
        FR & 200 & -0.01(3.22) & -0.00(3.28) & -0.07(4.53) & -0.10(8.06) & -0.46(34.44) \\
        & 400 & -0.01(3.67) & -0.04(4.02) & -0.11(4.66) & -0.01(11.98) & 0.11(45.41) \\
        & 800 & 0.01(4.02) & -0.03(4.80) & -0.02(4.80) & -0.11(17.81) & -0.93(61.40) \\
        & 1600 & -0.03(4.20) & 0.05(5.15) & -0.09(4.82) & -0.28(25.76) & -0.93(83.55) \\
        & 3200 & 0.08(4.34) & -0.03(5.25) & -0.01(4.76) & 0.14(37.43) & 0.64(116.10) \\
        \hline
        OR & 200  & -0.01(3.12) & 0.03(3.20)  & -0.06(4.46) & -0.11(8.13) & -0.25(33.91) \\
        & 400  & -0.02(3.63) & 0.04(4.00)  & -0.08(4.70) & -0.03(12.09) & -0.05(45.77) \\
        & 800  & 0.02(3.91)  & -0.07(4.70) & -0.09(4.70) & 0.17(17.85)  & -0.78(61.98) \\
        & 1600 & 0.03(4.21)  & -0.00(5.02) & -0.05(4.80) & -0.17(26.02) & 0.27(83.44)  \\
        & 3200 & 0.03(4.30)  & 0.00(5.16)  & 0.02(4.77)  & 0.01(37.51)  & 1.12(116.69) \\
        \hline
    \end{tabular}
    \begin{tablenotes}
        \item
        CR, complete randomization; RMM, the randomization procedure minimizing the imbalance measure; FR, the feasible randomization procedure; OR, the oracle randomization procedure.
    \end{tablenotes}
\end{threeparttable}
\end{adjustbox}
\end{table}

First, our simulation results demonstrate that all randomization procedures except complete randomization balance the covariates well, and beyond a certain sample size, further increasing the sample size has little effect on the standard deviation of the balanced covariates.
The result indicates the convergence rate $O_P(1)$, which outperforms $O_P(\sqrt{n})$ for the complete randomization. It agrees with the conclusion in Corollary \ref{corollary_boundedness_convergent}, Theorems \ref{theorem_shift_liu_xiao} and \ref{theorem_shift_simultaneous_geometric_finite_adjusted}.

Second, the results show that the average shifts on the additional covariates vanish when using CR, FR, and OR procedures.
Numerical results show that the shift remains close to zero under OR and FR procedures.
Furthermore, the variances of the imbalances associated with the additional covariates, $\sum_k (T_k-\rho)\sqrt{\sum_i|X_{k,i}|}$ and $\sum_k (T_k-\rho)\sum_i X_{k,i}^2$, are lower under the OR and FR procedures than under the CR procedure.
We therefore conclude that these procedures are effective in balancing additional covariates in practice.
Moreover, for all procedures considered, the standard deviation of the imbalance vector associated with the additional covariate scales with $\sqrt{n}$, supporting the asymptotic normality of the imbalance vector associated with the additional covariate.

Third, we evaluate the performance of the FR and OR procedures proposed in this article.
FR and OR procedures have better performances in eliminating the average shift and weaker balance effects compared to the first two procedures.
This phenomenon arises because the allocation function is designed to avoid deviating much from $\rho$. The lower bound of the allocation probability is $2/3-0.2$, not $0.1$ compared to the RM procedure.
We find that the FR and OR procedures perform similarly in terms of standard deviation of imbalances, even in the small sample case.
This finding confirms that parameter updates have negligible impact on performance and estimation, validating the reliability of substituting the oracle value with the parameter estimate sequence.

\section{Discussion}
\label{sec_conclusion}

In this article, we propose a new CAR procedure that marginally balances continuous covariates without inducing the shift problem in the context of the unequal targeted allocation ratio.
This procedure is incorporated into a new general framework, within which we establish a series of theoretical results.
These results include sufficient conditions for achieving covariate balance.
They also include necessary and sufficient conditions under which the asymptotic distribution of the imbalance for additional covariates is centered at zero, either when the parameters are fixed or when they converge.

Our findings, especially the equivalent form of the asymptotic existence of the average shift, also offer insights into Pocock and Simon's minimization procedure (\cite{pocockSequentialTreatmentAssignment1975}), which aims to marginally balance discrete covariates.
Thus, the procedure is susceptible to the shift problem and potentially even more so in multi-arm trials.
This potential issue suggests that the shift problem may extend beyond the settings previously considered.

Lastly, the challenge of conducting valid and more efficient inference under marginal covariate balance, particularly in the absence of model-based assumptions, remains open.
We leave this as a future research project.

\section{Supplementary Material}

Supplementary Material includes additional simulation results, additional theoretical results and proofs.

\appendix

\section{Redesign of a Depression Clinical Trial}
\label{sec_empirical_illustration}

In this section, we present a clinical trial example to compare several randomization procedures.
The clinical trial involving 681 patients, reported in (\cite{kellerComparisonNefazodoneCognitive2000}), evaluated the effectiveness of nefazodone, cognitive behavioral-analysis system of psychotherapy, and their combination in treating chronic depression.
The trial has three treatment groups, which we combine into two: single treatment versus combination, giving an unequal targeted allocation ratio of $\rho = 2/3$.
The key covariate vector $X$ used for balance consists of three baseline variables: age (AGE, continuous), Global Assessment of Functioning Score (GlobalAssess, continuous), and the 24-item Hamilton Rating Scale for Depression at baseline (HAMD24, continuous).
After scaling the covariates, patients are re-allocated based on these scaled covariates using different randomization procedures, and the balancing properties are examined.
The procedures under comparison are the complete randomization (CR), the randomization procedure minimizing the imbalance measure (RMM, implemented as MIX(0,0,1,0) in Liu, Hu and Ma (\cite{liuPropertiesCovariateadaptiveRandomization2025})), and the feasible randomization procedure (FR).
The procedure minimizing the imbalance measure and the feasible randomization procedure are implemented with the same settings as in Section \ref{sec_experiments}, including $\rho_1 = 0.9$ and $p = 0.2$, with allocation probability set to $\rho$ when $n < 10$.

We first examine the balance of the key covariate vector $X$.
In addition, we assess the balance on three further covariates: the 17-item Hamilton Rating Scale for Depression (HAMD17) and the vector $\left( \sum_{i=1}^3 \sgn(X_i)\sqrt{|X_i|}, \sum_{i=1}^3 X_i^2 \right)$, where $X_i$ denotes the $i$th component of $X$ here.
The sample sizes are $n=85, 170, 340, 681$. This setting aims to illustrate the balancing properties at different stages of enrollment. For each case, we conduct the simulation for $N=10000$ times.
The means and standard deviations are used to evaluate the balancing performance of imbalance vectors.
The results are given in Table \ref{table_real_data_balance}.

\begin{table}
\centering
\caption{Means (standard deviations) of covariate imbalances under various randomization procedures. Each case is based on 10000 replicates}
\label{table_real_data_balance}
\begin{adjustbox}{scale=0.6}
\begin{threeparttable}
    \begin{tabular}{rrrrrrrr}
        \toprule
        Randomization & n & $\sum_k (T_k-\rho)X_{k,1}$ & $\sum_k (T_k-\rho)X_{k,2}$ & $\sum_k (T_k-\rho)X_{k,3}$ & $\sum_k (T_k-\rho)\sum_i \sgn(X_{k,i})\sqrt{|X_{k,i}|}$ & $\sum_k (T_k-\rho)\sum_i X_{k,i}^2$ & $\sum_k (T_k-\rho) \times HAMD17^2$ \\
        \hline
        CR & 85 & 0.01(4.36) & 0.00(4.36) & -0.02(4.31) & 0.00(5.96) & 0.05(16.64) & 1.33(88.31) \\
        & 170 & 0.07(6.16) & -0.01(6.18) & -0.06(6.12) & -0.00(8.49) & -0.24(23.33) & -1.23(125.79) \\
        & 340 & 0.00(8.71) & -0.02(8.75) & 0.10(8.59) & 0.03(11.81) & 0.39(33.09) & 1.15(178.35) \\
        & 681 & 0.05(12.35) & -0.03(12.26) & -0.01(12.35) & 0.05(16.87) & -0.66(47.11) & -6.07(251.08) \\
        \hline
        RMM & 85 & 0.01(0.73) & -0.01(0.76) & 0.04(0.74) & 0.10(2.12) & -3.13(16.94) & -71.05(93.48) \\
        & 170 & 0.02(0.74) & -0.01(0.76) & 0.03(0.76) & 0.15(2.81) & -8.10(24.06) & -154.88(133.62) \\
        & 340 & 0.01(0.75) & -0.03(0.77) & 0.04(0.75) & 0.26(3.92) & -17.55(34.41) & -320.18(190.98) \\
        & 681 & 0.02(0.73) & -0.01(0.78) & 0.04(0.75) & 0.67(5.41) & -37.37(49.08) & -657.63(269.93) \\
        \hline
        FR & 85 & 0.03(2.52) & -0.01(2.53) & -0.04(2.51) & -0.01(4.08) & 0.19(16.52) & 1.84(88.55) \\
        & 170 & 0.05(2.76) & 0.02(2.88) & -0.01(2.87) & 0.05(4.82) & -0.33(23.16) & -1.56(126.04) \\
        & 340 & -0.06(2.91) & 0.06(3.07) & -0.02(3.03) & -0.06(5.69) & 0.24(32.50) & 0.24(177.29) \\
        & 681 & 0.01(2.92) & -0.01(3.13) & 0.02(3.13) & 0.02(7.03) & -0.46(46.19) & -4.84(247.81) \\
        \hline
    \end{tabular}
    \begin{tablenotes}
        \item
        CR, complete randomization; RMM, the randomization procedure minimizing the imbalance measure; FR, the feasible randomization procedure.
    \end{tablenotes}
\end{threeparttable}
\end{adjustbox}
\end{table}

The result confirms that both adaptive procedures, RMM and FR, significantly outperform complete randomization (CR) in balancing the key covariates. As shown in Table \ref{table_real_data_balance}, their standard deviations for the imbalance remain small and stable as the sample size grows, demonstrating a clear advantage over CR.

The crucial distinction, however, lies in their handling of additional covariates.
The RMM procedure induces a severe and systematic imbalance, with the mean shift on the squared HAMD17 term reaching $-657.63$ at the full sample size, highlighting a significant practical risk.
Notably, for these additional covariates, the standard deviations of the imbalance scale at the order of $\sqrt{n}$ across all procedures, which is consistent with the expected asymptotic normality.
In contrast, the FR procedure successfully avoids this risk, maintaining balance on both the key and additional covariates.
This balance property makes the FR procedure a far more reliable and robust method, as it provides the benefits of adaptive randomization without the dangerous side effect of inducing bias in important additional covariates.

\section{Shift Problem for Pocock and Simon's Minimization Procedures under Unequal Targeted Allocation Ratio}
\label{sec_shift_problem_discrete}

In this section, we present the simulation results of Pocock and Simon's minimization procedures with two different imbalance measures. Consider $I$ covariates and $m_i$ levels for the
$i$th covariate, the imbalance measures are defined as
\begin{equation*}
    \mathrm{Imb}_{n,\text{square}}
    =\sum_{i=1}^I \sum_{k_i=1}^{m_i} w_i\left[D_n\left(i ; k_i\right)\right]^2 ,
\end{equation*}
and
\begin{equation*}
    \mathrm{Imb}_{n,\text{abs}}
    =\sum_{i=1}^I \sum_{k_i=1}^{m_i} w_i\left| D_n\left(i ; k_i\right) \right| ,
\end{equation*}
where $D_n\left(i ; k_i\right)$ is the differences between the numbers of units in the two treatment groups on the margin formed
by units whose $i$th covariate is at level $k_i$ and $w_i$ are nonnegative weights placed within a covariate margin.
Moreover, we set the allocation probability to
\begin{align*}
    \operatorname{pr}\left(T_n=1 \mid T_1, \ldots, T_{n-1}, X_1, \ldots, X_n\right)=
    \begin{cases}
        \rho_1 & \text { if } \operatorname{Imb}_n^{(1)}<\operatorname{Imb}_n^{(0)}, \\
        1-\rho_1 & \text { if } \operatorname{Imb}_n^{(1)}>\operatorname{Imb}_n^{(0)}, \\
        \rho & \text { if } \operatorname{Imb}_n^{(1)}=\operatorname{Imb}_n^{(0)}.
    \end{cases}
\end{align*}
Thus, the procedure is the same as the one in Subsection \ref{subsec_main_liu2024}, except that the imbalance measure is modified to accommodate the discrete covariate case.
Furthermore, the corresponding potential imbalance measures are obtained by replacing $D_n(i; k_i^*)$ with $D_{n-1}(i; k_i^*) + 1$ or $D_{n-1}(i; k_i^*) - 1$, depending on whether the current unit is hypothetically allocated to the treatment or the control, where $k_i^*$ denotes the level of  the $i$th covariate for the $n$th unit.
For further details, we refer the reader to (\cite{huTheoryCovariateAdaptiveDesigns2020}).

In this simulation, we set $\rho_1 = 0.99$ and $\rho = 2/3$.
The number of covariates is $I = 2$, with $m_1 = 2$ and $m_2 = 3$ levels for the two covariates, respectively.
We set the weights $w_1 = 1$ and $w_2 = 2$.
In generating the covariate data, we set the sampling probabilities to be proportional to $(1, 4, 1, 3, 1, 3)$ for the covariate level combinations $(1,1)$, $(1,2)$, $(1,3)$, $(2,1)$, $(2,2)$, and $(2,3)$, where each pair corresponds to a possible combination of levels for the two covariates.
The unequal sampling probabilities are intended to induce asymmetry in the covariate distribution.

In this simulation, we do not specify an additional covariate of interest.
Instead, we aim to study the average shift within each stratum, namely, within each unique covariate level combination.
The sample sizes are $n=200, 400, 800, 1600, 3200$. This setting aims to illustrate the relationship between the shift and the sample size $n$. For each case, we conduct the simulation for $N=10000$ times.
The means and standard deviations are used to evaluate the balancing performance of imbalance vectors and shifts. The simulation results are given in Table \ref{table_shift_discrete}.

\begin{table}
\centering
\caption{Means (standard deviations) of covariate imbalances under Pocock and Simon's minimization procedures with various imbalance types. Each case is based on 10000 replicates}
\label{table_shift_discrete}
\begin{adjustbox}{scale = 0.9}
\begin{threeparttable}
    \begin{tabular}{rrrrrrrr}
        \toprule
        $\mathrm{Imb}$ Type & $n$ & (1,1) & (1,2) & (1,3) & (2,1) & (2,2) & (2,3)
        \\
        \hline
        $\mathrm{Imb}_{\text{square}}$ & 200 & 0.09(1.33) & -0.17(1.35) & 0.08(1.33) & -0.09(1.34) & 0.17(1.33) & -0.08(1.34) \\
        & 400 & 0.18(1.86) & -0.36(1.88) & 0.19(1.87) & -0.17(1.87) & 0.36(1.88) & -0.18(1.88) \\
        & 800 & 0.39(2.61) & -0.72(2.64) & 0.34(2.63) & -0.38(2.62) & 0.71(2.64) & -0.34(2.64) \\
        & 1600 & 0.73(3.67) & -1.41(3.74) & 0.70(3.69) & -0.72(3.68) & 1.42(3.73) & -0.70(3.70) \\
        & 3200 & 1.50(5.25) & -2.90(5.34) & 1.42(5.23) & -1.50(5.26) & 2.90(5.33) & -1.41(5.23) \\
        \hline
        $\mathrm{Imb}_{\text{abs}}$ & 200 & 0.14(1.31) & -0.28(1.33) & 0.15(1.31) & -0.14(1.33) & 0.28(1.32) & -0.14(1.33) \\
        & 400 & 0.26(1.85) & -0.56(1.86) & 0.31(1.84) & -0.25(1.86) & 0.57(1.85) & -0.30(1.86) \\
        & 800 & 0.56(2.56) & -1.11(2.61) & 0.57(2.57) & -0.55(2.57) & 1.11(2.60) & -0.56(2.58) \\
        & 1600 & 1.10(3.70) & -2.19(3.69) & 1.12(3.66) & -1.09(3.70) & 2.21(3.68) & -1.12(3.67) \\
        & 3200 & 2.14(5.18) & -4.29(5.23) & 2.17(5.14) & -2.13(5.19) & 4.30(5.23) & -2.17(5.14) \\
        \hline
    \end{tabular}
    \begin{tablenotes}
        \item
        $\mathrm{Imb}_{\text{square}}=\sum_{i=1}^I \sum_{k_i=1}^{m_i} w_i\left[D_n\left(i ; k_i\right)\right]^2$; $\mathrm{Imb}_{\text{abs}}=\sum_{i=1}^I \sum_{k_i=1}^{m_i} w_i\left| D_n\left(i ; k_i\right) \right|$.
    \end{tablenotes}
\end{threeparttable}
\end{adjustbox}
\end{table}

The results show that under two different imbalance measures, the shift in each stratum increases linearly with $n$, resulting in a non-vanishing average shift.
This finding suggests that the shift problem may arise in the discrete covariate case when using Pocock and Simon's minimization procedures to marginally balance covariates.

\section{Additional Theorems on the Randomization Procedure}

\begin{theorem}
    \label{theorem_drift_condition_inequality}
    Suppose $\theta \in K_{M,\Delta}$. If Assumption \ref{assumption_x_sub_exponential_bound} holds, then there exists positive constants $\beta<1$, $b$ and $\lambda_1$ that only depend on $M$, $\Delta$, $\lambda$ and $C$ such that for any $\Lambda \in W_\Gamma$,
    \begin{equation*}
        \mathbb{E}_\theta \left[ e^{\lambda_1 \|\Lambda_1\|} \mid \Lambda_0 = \Lambda \right] \leq \beta e^{\lambda_1 \|\Lambda\|} + b .
    \end{equation*}

    Denote the Lyapunov function $V(\Lambda) = e^{\lambda_1 \|\Lambda\|}$. For any $\alpha \in (0,1]$, there exists positive constants $\beta_\alpha<1$ and $b_\alpha=b$ such that the inequality
    \begin{equation*}
        (P_\theta V^\alpha)(\Lambda) \leq \beta_\alpha V^\alpha(\Lambda) + b_\alpha
    \end{equation*}
    holds. In addition, $\beta_\alpha$ can be chosen as a continuous function of $\alpha$ which converges to $1$ as $\alpha \rightarrow 0$.
\end{theorem}

\begin{lemma}
    \label{lemma_spread_out_on_sphere}
    Under Assumption \ref{assumption_x_spread_out}, there exist three constants $d_{\mathrm{min}} \in \mathbb{N}^*$, $a_{d,P}>0$ and $b_{d,P} \in \mathbb{R}$ such that, for any $d \geq d_{\mathrm{min}}$, there exists a positive number $\delta_{d,P}$ satisfying, for any $\theta$,
    \begin{equation*}
        P_\theta^{d}(\Lambda, \cdot) \geq \delta_{d,P} \mu_{\mathrm{leb}}(\cdot \cap B(\Lambda,a_{d,P} d + b_{d,P})) ,
    \end{equation*}
    where $B(\Lambda,a_{d,P} d + b_{d,P})$ is the ball centered at $\Lambda$ with radius $a_{d,P} d + b_{d,P}$, and the set $\cdot \cap B(\Lambda,a_{d,P} d + b_{d,P})$ is the intersection of an arbitrary set $\cdot$ and $B(\Lambda,a_{d,P} d + b_{d,P})$.
    Consequently, the linear subspace $W_\Gamma$ spanned by the support of the distribution $\Gamma$ is equal to the entire space $\mathbb{R}^d$. In addition, the Markov chain corresponding to $P_\theta$ is aperiodic and irreducible.
    Moreover, the Lebesgue measure on $\mathbb{R}^d$, $\mu_{\mathrm{leb}}$, is the maximal irreducibility measure.
\end{lemma}

\begin{theorem}
    \label{theorem_simultaneous_geometric_ergodicity}
    Suppose $K_{M,\Delta}$ is not empty. If Assumptions \ref{assumption_x_sub_exponential_bound} and \ref{assumption_x_spread_out} hold, then for any $\theta \in K_{M,\Delta}$, $P_\theta$ is positive recurrent with the unique invariant probability $\pi_\theta$ and $\pi_\theta V \leq \frac{b}{1-\beta}$. Moreover, there exists a positive number $L>1$ such that for any $\theta \in K_{M,\Delta}$,
    \begin{equation*}
        \| P_\theta^{n}(\Lambda, \cdot) - \pi_\theta \|_V \leq L(1-L^{-1})^n V(\Lambda) .
    \end{equation*}
    
    If $V$ is replaced by $V^\alpha$ for some $\alpha \in (0,1]$, the corresponding constant is denoted by $L_\alpha$.
\end{theorem}

\section{Proofs for General Framework of Randomization Procedure}

\subsection{Proof of Theorem \ref{theorem_drift_condition_inequality}}

\label{subsec_proof_theorem_drift_condition_inequality}

In this subsection, the primary objective is to establish an upper bound for the expression
\begin{equation}
    \mathbb{E}_\theta \left[ e^{\lambda_1 (\|\Lambda_1\| - \|\Lambda_0\|) } \mid \Lambda_0 = \Lambda \right] \label{eq_target_to_bound}
\end{equation}
when $\Lambda \in W_\Gamma$ is sufficiently large and $\lambda_1 > 0$ is a constant determined by the constants $M$, $\Delta$, $\lambda$ and $C$ in Assumption \ref{assumption_x_sub_exponential_bound}.
Thus, the bound is universal for $\theta \in K_{M,\Delta}$.

For any parameter $\theta \in K_{M,\Delta}$ and the initial state $\Lambda_0$, the following random variable $\Lambda_1 \sim P_\theta(\Lambda_0, \cdot)$.
It holds that $\mathbb{E} \|X_1\| < \infty$.
Thus, there exists some sufficiently large number $M_1>0$,
\begin{align*}
    & \quad \mathbb{E}_\theta \left[ \| \Lambda_1 - \Lambda_0 \| \mathbb{I} \left( \|\Lambda_1 - \Lambda_0\| > M_1 \right) \mid \Lambda_0=\Lambda \right] \\
    & \leq \mathbb{E}_\theta \left[ \|X_1\| \mathbb{I} \left( \|X_1\| > M_1 \right) \mid \Lambda_0=\Lambda \right]
    \leq \frac{\Delta}{2},
\end{align*}
where $\Delta$ is the constant in the expression $K_{M,\Delta}$.

From here, we begin the main part of the proof. First, we split the expression (\ref{eq_target_to_bound}) into
\begin{equation*}
    \mathbb{E}_\theta \left[ e^{\lambda_1 (\|\Lambda_1\| - \|\Lambda_0\|) } \mathbb{I} \left( \|\Lambda_1 - \Lambda_0\| > M_2 \right) \mid \Lambda_0 = \Lambda \right],
\end{equation*}
and
\begin{equation*}
    \mathbb{E}_\theta \left[ e^{\lambda_1 (\|\Lambda_1\| - \|\Lambda_0\|) } \mathbb{I} \left( \|\Lambda_1 - \Lambda_0\| \leq M_2 \right) \mid \Lambda_0 = \Lambda \right],
\end{equation*}
for some $M_2 > 0$.

For the former expectation, denote the constant $C=\mathbb{E} e^{\lambda \|X_1\|} < \infty$ for some constant $\lambda>0$, then we can show that for any positive number $M_2 > 0$ and $\lambda_1 \in (0,\lambda)$,
\begin{align*}
    & \quad \mathbb{E}_\theta \left[ e^{\lambda_1 (\|\Lambda_1\| - \|\Lambda_0\|) } \mathbb{I} \left( \|\Lambda_1 - \Lambda_0\| > M_2 \right) \mid \Lambda_0 = \Lambda \right] \\
    & \leq \mathbb{E}_\theta \left[ e^{\lambda_1 (\|\Lambda_1 - \Lambda_0\|) } \mathbb{I} \left( \|\Lambda_1 - \Lambda_0\| > M_2 \right) \mid \Lambda_0 = \Lambda \right] \\
    & \leq \mathbb{E}_\theta \left[ e^{\lambda (\|\Lambda_1 - \Lambda_0\|) } e^{(\lambda_1-\lambda) M_2} \mathbb{I} \left( \|\Lambda_1 - \Lambda_0\| > M_2 \right) \mid \Lambda_0 = \Lambda \right] \\
    & \leq \mathbb{E}_\theta \left[ e^{\lambda (\|\Lambda_1 - \Lambda_0\|) } e^{(\lambda_1-\lambda) M_2} \mid \Lambda_0 = \Lambda \right] \\
    & \leq \mathbb{E}_\theta \left[ e^{\lambda \|X_1\| } e^{(\lambda_1-\lambda) M_2} \mid \Lambda_0 = \Lambda \right] \\
    & \leq C e^{(\lambda_1-\lambda) M_2} .
\end{align*}

For the latter expectation, we try to linearize the term in it. We first control the error between the exponent and the linear term of $\|\Lambda_1\| - \|\Lambda_0 \|$ under the condition $\|\Lambda_1 - \Lambda_0\| \leq M_2$:
\begin{align*}
    \left| e^{\lambda_1 (\|\Lambda_1\| - \|\Lambda_0\|)}
    - [ 1 + \lambda_1 (\|\Lambda_1\| - \|\Lambda_0\|) ] \right|
    & \leq \left[ \lambda_1 (\|\Lambda_1\| - \|\Lambda_0\|) \right]^2
    \leq \left( \lambda_1 \|\Lambda_1 - \Lambda_0\| \right)^2 \\
    & \leq \left( \lambda_1 M_2 \right)^2, 
\end{align*}
under the condition $\lambda_1 M_2 \leq \inf_{x \in \mathbb{R}} \{|x| \mid |e^x-1-x|>x^2 \}$.

By treating $\| \Lambda_1 \| - \| \Lambda_0 \| = \|\Lambda_1 - \Lambda_0 + \Lambda_0\| - \| \Lambda_0 \|$ as a function of the perturbation $\Lambda_1 - \Lambda_0$ and the initial state $\Lambda_0$, we approximate it to first order by $(\Lambda_1 - \Lambda_0)^T \frac{\Lambda_0}{\|\Lambda_0\|}$ and provide a bound on the resulting approximation error under the condition $\|\Lambda_1 - \Lambda_0\| \leq M_2$:
\begin{align*}
    & \quad \left| \|\Lambda_1\| - \|\Lambda_0\| - (\Lambda_1 - \Lambda_0)^T \frac{\Lambda_0}{\|\Lambda_0\|} \right| \\
    &= \left| \frac{2 (\Lambda_1-\Lambda_0)^T \Lambda_0 + \|\Lambda_1-\Lambda_0\|^2}{\|\Lambda_1\| + \|\Lambda_0\|} - (\Lambda_1 - \Lambda_0)^T \frac{\Lambda_0}{\|\Lambda_0\|} \right| \\
    &= \left| (\Lambda_1 - \Lambda_0)^T \Lambda_0 \left( \frac{2}{\|\Lambda_1\| + \|\Lambda_0\|} - \frac{1}{\|\Lambda_0\|} \right) + \frac{\|\Lambda_1-\Lambda_0\|^2}{\|\Lambda_1\| + \|\Lambda_0\|} \right| \\
    & \leq \left| \| \Lambda_1 - \Lambda_0 \| \| \Lambda_0 \| \frac{\|\Lambda_0\| - \|\Lambda_1\|}{(\|\Lambda_1\| + \|\Lambda_0\|) \|\Lambda_0\|} \right| + \left| \frac{\|\Lambda_1-\Lambda_0\|^2}{\|\Lambda_1\| + \|\Lambda_0\|} \right| \\
    & \leq \frac{\| \Lambda_1 - \Lambda_0 \|^2}{\|\Lambda_1\| + \|\Lambda_0\|} + \frac{\|\Lambda_1-\Lambda_0\|^2}{\|\Lambda_1\| + \|\Lambda_0\|} \\
    & \leq \frac{2 M_2^2}{\|\Lambda_0\|} .
\end{align*}

The above two inequalities are summarized as
\begin{align*}
    e^{\lambda_1 (\|\Lambda_1\| - \|\Lambda_0\|)}
    & \leq 1 + \lambda_1 (\|\Lambda_1\| - \|\Lambda_0\|) + \left( \lambda_1 M_2 \right)^2 \\
    & \leq 1 + \lambda_1
    \left[ (\Lambda_1 - \Lambda_0)^T \frac{\Lambda_0}{\|\Lambda_0\|} + \frac{2 M_2^2}{\|\Lambda_0\|} \right]
    + \left( \lambda_1 M_2 \right)^2 ,
\end{align*}
when $\|\Lambda_1 - \Lambda_0\| \leq M_2$ and $\lambda_1 M_2 \leq \inf_{x \in \mathbb{R}} \{|x| \mid |e^x-1-x|>x^2 \}$.

Thus, if $\lambda_1 M_2 \leq \inf_{x \in \mathbb{R}} \{|x| \mid |e^x-1-x|>x^2 \}$,
\begin{align*}
    & \quad \mathbb{E}_\theta \left[ e^{\lambda_1 (\|\Lambda_1\| - \|\Lambda_0\|) } \mid \Lambda_0 = \Lambda \right] \\
    & \leq C e^{(\lambda_1-\lambda) M_2} + \mathbb{E}_\theta \left[ e^{\lambda_1 (\|\Lambda_1\| - \|\Lambda_0\|) } \mathbb{I} \left( \|\Lambda_1 - \Lambda_0\| \leq M_2 \right) \mid \Lambda_0 = \Lambda \right] \\
    & \leq C e^{(\lambda_1-\lambda) M_2} \\
    & \quad + \mathbb{E}_\theta \left[
    \left[ 1 + \lambda_1(\Lambda_1 - \Lambda_0)^T \frac{\Lambda_0}{\|\Lambda_0\|} + \lambda_1^2 M_2^2 + \frac{2 \lambda_1 M_2^2}{\|\Lambda_0\|} \right]
    \mathbb{I} \left( \|\Lambda_1 - \Lambda_0\| \leq M_2 \right)
    \mid \Lambda_0 = \Lambda \right] \\
    & \leq \lambda_1 \mathbb{E}_\theta \left[ (\Lambda_1 - \Lambda_0)^T \frac{\Lambda_0}{\|\Lambda_0\|} \mid \Lambda_0 = \Lambda \right] + \lambda_1 \mathbb{E}_\theta \left[ \|\Lambda_1 - \Lambda_0\| \mathbb{I} \left( \|\Lambda_1 - \Lambda_0\| > M_2 \right) \mid \Lambda_0 = \Lambda \right] \\
    & \quad + C e^{(\lambda_1-\lambda) M_2} + 1 + \lambda_1^2 M_2^2 + \frac{2 \lambda_1 M_2^2}{\|\Lambda\|} .
\end{align*}
When $\theta \in K_{M,\Delta}$, $\Lambda \in W_\Gamma$, $\|\Lambda\| \geq M$ and $M_2 \geq M_1$, it can be further bounded by
\begin{align}
    & \quad \lambda_1 \cdot (-\Delta) + \lambda_1 \frac{\Delta}{2} + C e^{(\lambda_1-\lambda) M_2} + 1 + \lambda_1^2 M_2^2 + \frac{2 \lambda_1 M_2^2}{\|\Lambda\|} \label{eq_target_bound} \\
    &= -\frac{\lambda_1 \Delta}{2} + C e^{(\lambda_1-\lambda) M_2} + 1 + \lambda_1^2 M_2^2 + \frac{2 \lambda_1 M_2^2}{\|\Lambda\|} \notag .
\end{align}

In summary, we have supposed
\begin{equation}
    \label{eq_inequalities_drift_condition_1}
    \begin{cases}
        \lambda_1 M_2 & \leq \inf_{x \in \mathbb{R}} \{|x| \mid |e^x-1-x|>x^2 \} , \\
        M_2 & \geq M_1 , \\
        \|\Lambda\| & \geq M .
\end{cases}
\end{equation}

Furthermore, we suppose
\begin{equation}
    \label{eq_inequalities_drift_condition_2}
    \begin{cases}
        \lambda_1 &= \frac{1}{M_2^3} , \\
        \lambda_1 & \leq \frac{\lambda}{2} , \\
        \|\Lambda\| & \geq M_2^3 , \\
        M_2^3 & \geq M .
    \end{cases}
\end{equation}

Under conditions (\ref{eq_inequalities_drift_condition_1}) and (\ref{eq_inequalities_drift_condition_2}), the equation (\ref{eq_target_bound}) has an upper bound that
\begin{equation}
    -\frac{\Delta}{2M_2^3} + C e^{-\frac{\lambda M_2}{2}} + 1 + \frac{1}{M_2^4} + \frac{2}{M_2^4}
    = 1 + \frac{1}{M_2^3} \left[ -\frac{\Delta}{2} + C M_2^3 e^{-\frac{\lambda M_2}{2}} + \frac{3}{M_2} \right] \label{eq_upperbound_drift_condition} .
\end{equation}

Moreover, conditions (\ref{eq_inequalities_drift_condition_1}) and (\ref{eq_inequalities_drift_condition_2}) are equivalent to
\begin{equation}
    M_2 \geq \max \left\{ M^{\frac{1}{3}}, \left(\frac{2}{\lambda}\right)^{\frac{1}{3}}, M_1,
    \left[ \inf_{x \in \mathbb{R}} \{|x| \mid |e^x-1-x|>x^2 \} \right]^{-\frac{1}{2}} \right\} \label{eq_M2_lower_bound}
\end{equation}
and $\lambda_1 = \frac{1}{M_2^3}$, $\|\Lambda\| \geq M_2^3$. Since $C M_2^3 e^{-\frac{\lambda M_2}{2}} + \frac{3}{M_2} \rightarrow 0$ as $M_2 \rightarrow +\infty$, there exists a sufficiently large $M_2^*$ such that $M_2^*$ satisfies (\ref{eq_M2_lower_bound}) and for any $M_2 \geq M_2^*$, the inequality holds that
\begin{equation*}
    -\frac{\Delta}{2} + C M_2^3 e^{-\frac{\lambda M_2}{2}} + \frac{3}{M_2}
    < -\frac{\Delta}{4} .
\end{equation*}

In summary, we have proved that if $M_2 \geq M_2^*$, then for any $\Lambda \in W_\Gamma$ satisfying $\|\Lambda\| \geq M_2^3$, it holds that
\begin{equation*}
    \mathbb{E}_\theta \left[ e^{\frac{1}{M_2^3} (\|\Lambda_1\| - \|\Lambda_0\|) } \mid \Lambda_0 = \Lambda \right]
    \leq 1 + \frac{1}{M_2^3} \left[ -\frac{\Delta}{2} + C M_2^3 e^{-\frac{\lambda M_2}{2}} + \frac{3}{M_2} \right]
    < 1 .
\end{equation*}
On the other hand, for any $\Lambda$ satisfying $\|\Lambda\| < M_2^3$, it holds that
\begin{equation*}
    \mathbb{E}_\theta \left[ e^{\frac{1}{M_2^3} \|\Lambda_1\| } \mid \Lambda_0 = \Lambda \right]
    \leq \mathbb{E}_\theta \left[ e^{\frac{1}{M_2^3} \|\Lambda_1 - \Lambda_0\| } \mid \Lambda_0 = \Lambda \right]
    e^{\frac{1}{M_2^3} \|\Lambda\|}
    \leq \mathbb{E} e^{\frac{1}{M_2^3} \|X_1\| } \cdot e
    \leq C e ,
\end{equation*}
because $\frac{1}{M_2^3} < \lambda$. These two inequalities lead to
\begin{equation}
    \label{eq_origin_drift_condition}
    \mathbb{E}_\theta \left[ e^{\frac{1}{M_2^3} \|\Lambda_1\|} \mid \Lambda_0 = \Lambda \right]
    \leq \left\{ 1 + \frac{1}{M_2^3} \left[ -\frac{\Delta}{2} + C M_2^3 e^{-\frac{\lambda M_2}{2}} + \frac{3}{M_2} \right] \right\}
    e^{\frac{1}{M_2^3} \|\Lambda\|} + C e .
\end{equation}

Therefore, taking $\lambda_1=\frac{1}{(M_2^*)^3}$, $b = C e$ and
\begin{equation*}
    \beta
    = 1 + \frac{1}{(M_2^*)^3} \left[ -\frac{\Delta}{2} + C (M_2^*)^3 e^{-\frac{\lambda M_2^*}{2}} + \frac{3}{M_2^*} \right]
    < 1-\frac{\Delta}{4(M_2^*)^3}
    < 1 ,
\end{equation*}
the inequality
\begin{equation*}
    \mathbb{E}_\theta \left[ e^{\lambda_1 \|\Lambda_1\|} \mid \Lambda_0 = \Lambda \right] \leq \beta e^{\lambda_1 \|\Lambda\|} + b
\end{equation*}
holds for any $\theta \in K_{M,\Delta}$. Denote $\exp(\lambda_1\|\Lambda\|)$ in this situation by $V(\Lambda)$.

Finally, for any $\alpha \in (0,1]$ and $M_2 = \frac{M_2^*}{\alpha^{1/3}}$, the left-hand side of the inequality (\ref{eq_origin_drift_condition}) is
\begin{equation*}
    \mathbb{E}_\theta \left[ e^{\frac{1}{M_2^3} \|\Lambda_1\|} \mid \Lambda_0 = \Lambda \right]
    = \mathbb{E}_\theta \left[ e^{\lambda_1 \alpha \|\Lambda_1\|} \mid \Lambda_0 = \Lambda \right]
    = \mathbb{E}_\theta \left[ V^\alpha(\Lambda_1) \mid \Lambda_0 = \Lambda \right]
    = (P_\theta V^\alpha)(\Lambda) .
\end{equation*}
Thus, the inequality (\ref{eq_origin_drift_condition}) with $M_2 = \frac{M_2^*}{\alpha^{1/3}}$ implies that the inequality $P_\theta V^{\alpha} \leq \beta_{\alpha} V^{\alpha} + b_{\alpha}$ will hold for some $\beta_{\alpha} \in (0,1)$ and $b_{\alpha}>0$. The constants $\beta_{\alpha}$ and $b_{\alpha}$ can be obtained by
\begin{align*}
    \beta_{\alpha}
    &= 1 + \frac{1}{M_2^3} \left[ -\frac{\Delta}{2} + C M_2^3 e^{-\frac{\lambda M_2}{2}} + \frac{3}{M_2} \right] \\
    &= 1 + \frac{\alpha}{(M_2^*)^3} \left[ -\frac{\Delta}{2} + \frac{C (M_2^*)^3 e^{-\frac{\lambda M_2^*}{2\alpha^{1/3}}}}{\alpha} + \frac{3\alpha^{1/3}}{M_2^*} \right] , \\
    b_{\alpha} &= C e .
\end{align*}
Thus, the inequality $P_\theta V^\alpha \leq \beta_\alpha V^\alpha + b_\alpha$ holds with a constant parameter $b_\alpha=b$ and a continuous parameter $\beta_\alpha<1$ whose limit is $1$ as $\alpha \rightarrow 0$.

\subsection{Proof of Theorem \ref{theorem_boundedness}}

\label{subsec_proof_theorem_boundedness}

The proof of Theorem \ref{theorem_boundedness} is based on Lemma \ref{lemma_convergence_V_as}.
Theorem \ref{theorem_drift_condition_inequality} implies that
\begin{equation*}
    (P_\theta V^\alpha)(\Lambda) \leq \beta_\alpha V^\alpha(\Lambda) + b_\alpha
\end{equation*}
for any $\theta \in K_{M,\Delta}$ and $\alpha \in (0,1]$ with $V(\Lambda) = \exp(\lambda_1 \|\Lambda\|)$.
Suppose $\alpha = \min(1,\frac{\lambda}{\lambda_1})$, where the constant $\lambda$ is defined in Assumption \ref{assumption_x_sub_exponential_bound}.
By Assumption \ref{assumption_x_sub_exponential_bound}, the definition $V(\Lambda) = \exp(\lambda_1 \|\Lambda\|)$, and Jensen's inequality, the inequality
\begin{align*}
    \mathbb{E} [V^\alpha(\Lambda_n)]
    & \leq \mathbb{E} \left[ V^\alpha(\Lambda_0 + \sum_{i=1}^n (T_i - \rho) X_i) \right]
    = \mathbb{E} \left[ \exp(\alpha \lambda_1 \|\Lambda_0 + \sum_{i=1}^n (T_i - \rho) X_i\|) \right] \\
    & \leq \mathbb{E} \left[ \exp(\alpha \lambda_1 ( \|\Lambda_0\| + \sum_{i=1}^n \|(T_i - \rho) X_i\| ) ) \right] \\
    & \leq \mathbb{E} \left[ \exp(\alpha \lambda_1 ( \|\Lambda_0\| + \sum_{i=1}^n \| X_i\| ) ) \right]
    = \mathbb{E} \left[ \exp(\alpha \lambda_1 \|\Lambda_0\|) \prod_{i=1}^n \exp(\alpha \lambda_1 \|X_i\|) \right] \\
    & \leq \mathbb{E} \left[ \exp(\alpha \lambda_1 \|\Lambda_0\|) \right]
    \prod_{i=1}^n \mathbb{E} \left[ \exp(\lambda \|X_i\|) \right]
    = C^n \mathbb{E} [V^\alpha(\Lambda_0)]
    < \infty
\end{align*}
holds for any $n \geq 1$, where the constant $C$ is also defined in Assumption \ref{assumption_x_sub_exponential_bound}.
Therefore, $\mathbb{E} [V^\alpha(\Lambda_n)]$ is finite for any $n$.

Taking $S_\Theta = K_{M,\Delta}$, we have $P(\theta_n \notin S_\Theta \text{ i.o.}) = 0$ because $\theta_n \in K_{M,\Delta}$ for sufficiently large $n$ almost surely.
Then based on Lemma \ref{lemma_convergence_V_as} with $S_\Theta = K_{M,\Delta}$, we have $V^\alpha(\Lambda_n) = O_P(1)$. It implies that $\{\Lambda_n\}$ is bounded in probability.

\subsection{Proof of Lemma \ref{lemma_spread_out_on_sphere}}

In the proof, we write $\mu \geq \nu$ if the signed measure $\mu-\nu$ is a positive measure.
Recall that in the standard Markov chain setting, a transition kernel $P$ is a transition probability kernel, meaning that $P(\Lambda,\cdot)$ is a probability measure for each state $\Lambda$, so that $P(\Lambda,\mathrm{X}) = 1$, where $\mathrm{X}$ denotes the entire state space. 
In this article, however, we adopt a more general notion of transition kernel, allowing $P(\Lambda,\cdot)$ to be an arbitrary measure.
Consequently, it may not always hold that $P(\Lambda,\mathrm{X}) = 1$.
Furthermore, for two transition kernels $P$ and $Q$, we write $P \geq Q$ if, for every state $\Lambda$, the inequality $P(\Lambda,\cdot) \geq Q(\Lambda,\cdot)$ holds in the sense of measures.

We denote by $K_{\mathrm{rw}}$ the transition kernel representing complete randomization for allocating units according to an unequal targeted allocation ratio $\rho$. This kernel is equivalent to the Markov chain of the random walk with increment $\Gamma_\rho$.

We separate the proof into three parts.

\textbf{Spread-Out Condition}

\begin{lemma}
    \label{lemma_spread_out}
    If Assumption \ref{assumption_x_spread_out} holds, there exists some integer $d_{\mathrm{rw}} \geq 1$, positive numbers $\delta$ and $R$ such that
    \begin{equation*}
        K_{\mathrm{rw}}^{d_{\mathrm{rw}}} (\Lambda, \cdot) \geq \delta \mu_{\mathrm{leb}}(\cdot \cap B(\Lambda,R))
    \end{equation*}
    for any $\Lambda$.
\end{lemma}

\begin{proof}[Proof of Lemma \ref{lemma_spread_out}]
    Owing to the translation invariance of the random walk, it suffices to prove that
    \begin{equation*}
        K_{\mathrm{rw}}^{d_{\mathrm{rw}}} (0, \cdot) \geq \delta \mu_{\mathrm{leb}}(\cdot \cap B(0,R)) .
    \end{equation*}

    Under Assumption \ref{assumption_x_spread_out}, Proposition 8.8 in (\cite{follandRealAnalysisModern1999}) implies that the function $f = f_{s,c}^{2*}$, i.e., the convolution of $f_{s,c}$ with itself, is bounded and uniformly continuous, where $f_{s,c}$ is defined as $f_{s,c}(x) = f_s(x)\mathbb{I}(f_s(x) < c)$, for some $c>0$ satisfying $\mu_{\mathrm{leb}}(\{x \mid f_s(x) < c\}) > 0$.

    Thus, there exists some point $\Lambda$ such that $f(\Lambda) > 0$. By the continuity of $f$, there exists $r>0$ and $\delta>0$ such that for any $x \in B(\Lambda,r)$, $f(x) > \delta$.

    Let $d_{\mathrm{sp}} = 2d_s$. Then it holds that
    \begin{align*}
        K_{\mathrm{rw}}^{d_{\mathrm{sp}}}(0, A)
        &= (\Gamma_\rho^{d_{s}*})^{2*}(A)
        \geq \int_A f_s^{2*}(\Lambda) \mu_{\mathrm{leb}}(d\Lambda) \\
        & \geq \int_A f(\Lambda) \mu_{\mathrm{leb}}(d\Lambda)
        \geq \delta \mu_{\mathrm{leb}}(A \cap B(\Lambda,r)) ,
    \end{align*}
    which is a slight variation of the conclusion of the lemma.

    If there exists a proper linear subspace $W \subset \mathbb{R}^d$ such that $\Gamma(W) = 1$, then it follows that $\Gamma_\rho(W) = 1$.
    Since $\Gamma_\rho(\cdot) = K_{\mathrm{rw}}(0, \cdot)$, the equality $\Gamma_\rho(W) = 1$ is equivalent to $K_{\mathrm{rw}}(0, W) = 1$.
    By the definition of $K_{\mathrm{rw}}$, we have that for any $\Lambda \in W$, $K_{\mathrm{rw}}(\Lambda, W) = 1$, because the transition kernel is translation-invariant.
    Thus, by recursion, $K_{\mathrm{rw}}(\Lambda, W) = 1$ implies that $K_{\mathrm{rw}}^{d_{\mathrm{sp}}}(\Lambda, W) = 1$ for any $\Lambda \in W$.
    In particular, we have $K_{\mathrm{rw}}^{d_{\mathrm{sp}}}(0, W) = 1$, which contradicts the lower bound
    \begin{equation*}
        K_{\mathrm{rw}}^{d_{\mathrm{sp}}}(0, W)
        = 1 - K_{\mathrm{rw}}^{d_{\mathrm{sp}}}(0, W^c)
        \leq 1 - \delta \mu_{\mathrm{leb}}(W^c \cap B(\Lambda,r)) < 1
    \end{equation*}
    as the Lebesgue measure of any proper linear subspace is zero.
    Therefore, we conclude that for any proper linear subspace $W \subset \mathbb{R}^d$, $\Gamma(W) < 1$.

    Denote the support of the distribution of $\Gamma$ by $A_{\mathrm{Supp},\Gamma}$. By $\Gamma(W) < 1$ for any proper linear subspace $W \subset \mathbb{R}^d$, it can be shown that there are $d$ points $x_1$, $\dots$, $x_d$ which are linearly independent and belong to $A_{\mathrm{Supp},\Gamma}$.

    Let the point set $A_d$ and $A_c$ denote
    \begin{equation*}
        \left\{ \sum_{i=1}^d ( z_i\rho + z_i^\prime(\rho-1))x_i \mid z_i,z_i^\prime\in\mathbb{N} \right\}
    \end{equation*}
    and
    \begin{equation*}
        \left\{ \sum_{i=1}^d ( w_i\rho + w_i^\prime(\rho-1))x_i \mid w_i,w_i^\prime\in\mathbb{R}^+ \right\} ,
    \end{equation*}
    respectively.
    Because $x_1$, $\dots$, $x_d$ are linearly independent, the set $A_c$ is $\mathbb{R}^d$ exactly.
    By taking $z_i = \lfloor w_i \rfloor$ and $z_i^\prime = \lfloor w_i^\prime \rfloor$, we have that
    \begin{align*}
        & \quad d \left( \sum_{i=1}^d ( w_i\rho + w_i^\prime(\rho-1))x_i, \sum_{i=1}^d ( z_i\rho + z_i^\prime(\rho-1))x_i \right) \\
        & \leq \sum_{i=1}^d \left[ |w_i-z_i|\rho+|w_i^\prime-z_i^\prime|(1-\rho) \right] \|x_i\|
        \leq 2 \sum_{i=1}^d \|x_i\|
        < \infty .
    \end{align*}
    Thus,
    \begin{align*}
        \sup_{x \in \mathbb{R}^d} \inf_{y \in A_d} d(x,y)
        &= \sup_{w_i,w_i^\prime\in\mathbb{R}^+} \inf_{z_i,z_i^\prime\in\mathbb{N}} d
        \left( \sum_{i=1}^d ( w_i\rho + w_i^\prime(\rho-1))x_i,
        \sum_{i=1}^d ( z_i\rho + z_i^\prime(\rho-1))x_i \right) \\
        & \leq \sup_{w_i,w_i^\prime\in\mathbb{R}^+} d
        \left(\sum_{i=1}^d ( w_i\rho + w_i^\prime(\rho-1))x_i,
        \sum_{i=1}^d ( \lfloor w_i \rfloor \rho + \lfloor w_i^\prime \rfloor (\rho-1))x_i \right) \\
        & \leq 2 \sum_{i=1}^d \|x_i\|
        < \infty .
    \end{align*}
    It implies that when the integer $n$ is large enough, $nr-r/2>2 \sum_{i=1}^d \|x_i\|$. Thus, $B(n\Lambda,nr-r/2) \cap A_d$ is nonempty.
    Furthermore, it implies that the set $B(n\Lambda+\sum_{i=1}^d (-z_i\rho + z_i^\prime(1-\rho))x_i,nr-r/2)$ contains the origin point for some positive integer $n$ and some nonnegative integers $z_i$ and $z_i^\prime$.
    We fix the nonnegative integers $\left\{z_i, z_i^\prime \mid i \in \{1,\dots,d\} \right\}$ and the positive integer $n$.

    Because
    \begin{equation*}
        K_{\mathrm{rw}}^{d_{\mathrm{sp}}}(0, A)
        \geq \delta \mu_{\mathrm{leb}}(A \cap B(\Lambda,r)) ,
    \end{equation*}
    it can be shown by convolution that when $d_{\mathrm{sp}}^\prime=n d_{\mathrm{sp}}$,
    \begin{equation}
        K_{\mathrm{rw}}^{d_{\mathrm{sp}}^\prime}(0, A)
        \geq \delta^\prime \mu_{\mathrm{leb}}(A \cap B(n\Lambda,nr-r/4)) \label{eq_convolution_dsp_large} ,
    \end{equation}
    for some positive number $\delta^\prime$ depending on $d_{\mathrm{sp}}^\prime$.

    Since the points $x_1, \dots, x_d \in A_{\mathrm{Supp},\Gamma}$, we have that $\Gamma(B(x_i,r_o))>0$ for any $r_o>0$ and $i \in \{1,\dots,d\}$.
    Thus,
    \begin{align*}
        & \quad K_{\mathrm{rw}}^{\sum_{i=1}^d [|z_i|+|z_i^\prime|]}
        \left( 0, B \left( \sum_{i=1}^d (-z_i\rho + z_i^\prime(1-\rho))x_i,
        \sum_{i=1}^d [|z_i|+|z_i^\prime|] r_0 \right) \right) \\
        & \geq \prod_{i=1}^d \Gamma(B(x_i,r_o))^{|z_i|+|z_i^\prime|}
        \min(\rho,1-\rho)^{\sum_{i=1}^d [|z_i|+|z_i^\prime|]} .
    \end{align*}
    Denote the value of the right side by $\delta^{\prime\prime} > 0$, the ball
    \begin{equation*}
        B^{\prime} = B \left( \sum_{i=1}^d (-z_i\rho + z_i^\prime(1-\rho))x_i, \sum_{i=1}^d [|z_i|+|z_i^\prime|] r_0 \right)
    \end{equation*}
    and the ball centered at the origin point with the same radius
    \begin{equation*}
        B^{\prime\prime} = B \left( 0, \sum_{i=1}^d [|z_i|+|z_i^\prime|] r_0 \right) .
    \end{equation*}
    Let the offset of the center $a_{\mathrm{off}} = \sum_{i=1}^d (-z_i\rho + z_i^\prime(1-\rho))x_i$.

    By the inequality (\ref{eq_convolution_dsp_large}), it holds that
    \begin{align*}
        K_{\mathrm{rw}}^{d_{\mathrm{sp}}^\prime + \sum_{i=1}^d [|z_i|+|z_i^\prime|]}(0, A)
        &= \int
        K_{\mathrm{rw}}^{d_{\mathrm{sp}}^\prime}(x, A)
        K_{\mathrm{rw}}^{\sum_{i=1}^d [|z_i|+|z_i^\prime|]}(0,dx) \\
        & \geq \int
        \delta^\prime \mu_{\mathrm{leb}}(A \cap B(n\Lambda+x,nr-r/4))
        K_{\mathrm{rw}}^{\sum_{i=1}^d [|z_i|+|z_i^\prime|]}(0,dx) \\
        & \geq \delta^\prime \int_{x \in B^{\prime}}
        \mu_{\mathrm{leb}}(A \cap B(n\Lambda+x,nr-r/4))
        K_{\mathrm{rw}}^{\sum_{i=1}^d [|z_i|+|z_i^\prime|]}(0,dx) \\
        &= \delta^\prime \int_{y \in B^{\prime\prime}}
        \mu_{\mathrm{leb}} \left( A \cap B \left( n\Lambda+a_{\mathrm{off}}+y,nr-r/4 \right) \right) \\
        & \quad \times K_{\mathrm{rw}}^{\sum_{i=1}^d [|z_i|+|z_i^\prime|]} \left( 0, a_{\mathrm{off}} + dy \right) ,
    \end{align*}
    where $y = x - a_{\mathrm{off}}$ is a variable substitution.

    Because the ball $B(n\Lambda+\sum_{i=1}^d (-z_i\rho + z_i^\prime(1-\rho))x_i,nr-r/2)$ contains the origin point, we have
    \begin{equation*}
        B \left( y,r/4 \right)
        \subset
        B \left( n\Lambda+a_{\mathrm{off}}+y,nr-r/4 \right) .
    \end{equation*}
    Thus, if $r_0$ is sufficiently small so that $B^{\prime\prime} \subset B(0,r/8)$, it follows that
    \begin{equation*}
        B \left( 0,r/8 \right)
        \subset
        B \left( y,r/4 \right)
    \end{equation*}
    for any $y \in B^{\prime\prime}$, and
    \begin{align*}
        K_{\mathrm{rw}}^{d_{\mathrm{sp}}^\prime + \sum_{i=1}^d [|z_i|+|z_i^\prime|]}(0, A)
        & \geq \delta^\prime \int_{y \in B^{\prime\prime}}
        \mu_{\mathrm{leb}} \left( A \cap B \left( y,r/4 \right) \right)
        K_{\mathrm{rw}}^{\sum_{i=1}^d [|z_i|+|z_i^\prime|]} \left( 0, a_{\mathrm{off}} + dy \right) \\
        & \geq \delta^\prime \int_{y \in B^{\prime\prime}}
        \mu_{\mathrm{leb}} \left( A \cap B \left( 0,r/8 \right) \right)
        K_{\mathrm{rw}}^{\sum_{i=1}^d [|z_i|+|z_i^\prime|]} \left( 0, a_{\mathrm{off}} + dy \right) \\
        &= \delta^\prime
        \mu_{\mathrm{leb}} \left( A \cap B \left( 0,r/8 \right) \right)
        K_{\mathrm{rw}}^{\sum_{i=1}^d [|z_i|+|z_i^\prime|]} \left( 0, a_{\mathrm{off}} + B^{\prime\prime} \right) \\
        &= \delta^\prime
        \mu_{\mathrm{leb}} \left( A \cap B \left( 0,r/8 \right) \right)
        K_{\mathrm{rw}}^{\sum_{i=1}^d [|z_i|+|z_i^\prime|]} \left( 0, B^\prime \right) \\
        &= \delta^\prime \delta^{\prime\prime}
        \mu_{\mathrm{leb}} \left( A \cap B \left( 0,r/8 \right) \right) .
    \end{align*}

    In summary, when $d_{\mathrm{rw}} = d_{\mathrm{sp}}^\prime + \sum_{i=1}^d [|z_i|+|z_i^\prime|]$, $R = r/8 > 0$ and $\delta = \delta^\prime \delta^{\prime\prime} > 0$, the lemma holds.
\end{proof}

\textbf{Aperiodicity, Irreducibility and Small Set Condition}

Denote $\epsilon_{\mathrm{rw}} = \iota$, where $\iota$ is the constant defined in Assumption \ref{assumption_sampling}.
Hence, for any $\theta$, it holds that $P_\theta \geq \epsilon_{\mathrm{rw}} K_{\mathrm{rw}}$.
Based on the spread-out condition in Assumption \ref{assumption_x_spread_out} and Lemma \ref{lemma_spread_out}, we have shown that there exist a positive integer $d_{\mathrm{rw}}$, positive numbers $\delta$ and $R$ such that
\begin{equation*}
    K_{\mathrm{rw}}^{d_{\mathrm{rw}}} (\Lambda, \cdot) \geq \delta \mu_{\mathrm{leb}}(\cdot \cap B(\Lambda,R)).
\end{equation*}

Define the transition kernel $K_{\mathrm{leb},m}(\Lambda, \cdot) = \mu_{\mathrm{leb}}(\cdot \cap B(\Lambda,(m+1)R/2))$, which can be equivalently written as
\begin{equation*}
    K_{\mathrm{leb},m}(\Lambda, A) = \int_A \chi_m(x-\Lambda) \mu_{\mathrm{leb}}(dx) ,
\end{equation*}
where the indicator function $\chi_m(x) = \mathbb{I} \left( x \in B(0,(m+1)R/2) \right)$. Thus, $K_{\mathrm{rw}}^{d_{\mathrm{rw}}} \geq \delta K_{\mathrm{leb},1}$.

Then, it follows that
\begin{align*}
    (K_{\mathrm{leb},m_1}K_{\mathrm{leb},m_2})(\Lambda, A)
    &= \int \chi_{m_2}(y-\Lambda) \int_A \chi_{m_1}(x-y) \mu_{\mathrm{leb}}(dx) \mu_{\mathrm{leb}}(dy) \\
    &= \int_A \int \chi_{m_1}(x-y) \chi_{m_2}(y-\Lambda) \mu_{\mathrm{leb}}(dy) \mu_{\mathrm{leb}}(dx) \\
    &= \int_A (\chi_{m_1} * \chi_{m_2})(x-\Lambda) \mu_{\mathrm{leb}}(dx) .
\end{align*}
When $\|x-\Lambda\| \leq \frac{m_1+m_2+1}{2}R$, it has a lower bound that
\begin{align}
    & \quad (\chi_{m_1} * \chi_{m_2})(x-\Lambda) \label{eq_indicator_convolution_lower_bound} \\
    & \geq \int_{y \in B \left( \frac{2m_1+1}{2m_1+2m_2+2} \Lambda + \frac{2m_2+1}{2m_1+2m_2+2} x,
    \frac{R}{4} \right)}
    \chi_{m_1}(x-y) \chi_{m_2}(y-\Lambda) \mu_{\mathrm{leb}}(dy) \notag \\
    &= \int_{y \in B \left( \frac{2m_1+1}{2m_1+2m_2+2} \Lambda + \frac{2m_2+1}{2m_1+2m_2+2} x,
    \frac{R}{4} \right)} \mu_{\mathrm{leb}}(dy) \notag \\
    &= \mu_{\mathrm{leb}}(B(0,R/4)) , \notag
\end{align}
because when $y \in B \left( \frac{2m_1+1}{2m_1+2m_2+2} \Lambda + \frac{2m_2+1}{2m_1+2m_2+2} x, \frac{R}{4} \right)$, it holds that
\begin{equation*}
    x-y \in B \left( \frac{2m_1+1}{2m_1+2m_2+2} (x-\Lambda), \frac{R}{4} \right)
    \subset B \left( 0, \frac{m_1+1}{2}R \right)
\end{equation*}
and
\begin{equation*}
    y-\Lambda \in B \left( \frac{2m_2+1}{2m_1+2m_2+2} (x-\Lambda), \frac{R}{4} \right)
    \subset B \left( 0, \frac{m_2+1}{2}R \right) .
\end{equation*}
The inequality (\ref{eq_indicator_convolution_lower_bound}) implies that
\begin{align*}
    (K_{\mathrm{leb},m_1}K_{\mathrm{leb},m_2})(\Lambda, A)
    &= \int_A (\chi_{m_1} * \chi_{m_2})(x-\Lambda) \mu_{\mathrm{leb}}(dx) \\
    & \geq \int_{A \cap B \left( \Lambda,\frac{m_1+m_2+1}{2} \right) }
    \mu_{\mathrm{leb}}(B(0,R/4)) \mu_{\mathrm{leb}}(dx) \\
    & \geq \mu_{\mathrm{leb}}(B(0,R/4)) \mu_{\mathrm{leb}}( A \cap B \left( \Lambda,\frac{m_1+m_2+1}{2}R \right) ) \\
    & \geq \mu_{\mathrm{leb}}(B(0,R/4)) K_{\mathrm{leb},m_1+m_2}(\Lambda,A) .
\end{align*}

Thus, by recursion, we can prove that
\begin{align*}
    K_{\mathrm{leb},1}^m (\Lambda, \cdot)
    & \geq \mu_{\mathrm{leb}}(B(0,R/4))^{m-1} \mu_{\mathrm{leb}}(\cdot \cap B(\Lambda, (m+1)R/2)) \\
    &= \mu_{\mathrm{leb}}(B(0,R/4))^{m-1} K_{\mathrm{leb},m} (\Lambda, \cdot) .
\end{align*}
Furthermore,
\begin{align*}
    P_\theta^{m d_{\mathrm{rw}}}(\Lambda, \cdot)
    & \geq (\epsilon_{\mathrm{rw}} K_{\mathrm{rw}})^{m d_{\mathrm{rw}}}(\Lambda, \cdot) \geq \epsilon_{\mathrm{rw}}^{m d_{\mathrm{rw}}} (\delta K_{\mathrm{leb},1})^m (\Lambda, \cdot) \\
    & \geq \epsilon_{\mathrm{rw}}^{m d_{\mathrm{rw}}} \delta^m \mu_{\mathrm{leb}}(B(0,R/4))^{m-1} K_{\mathrm{leb},m} (\Lambda, \cdot) ,
\end{align*}
and
\begin{align*}
    P_\theta^{m d_{\mathrm{rw}} + l}
    & \geq \epsilon_{\mathrm{rw}}^{m d_{\mathrm{rw}}} \delta^m \mu_{\mathrm{leb}}(B(0,R/4))^{m-1} K_{\mathrm{leb},m}P_\theta^l \\
    & \geq \epsilon_{\mathrm{rw}}^{m d_{\mathrm{rw}} + l} \delta^m \mu_{\mathrm{leb}}(B(0,R/4))^{m-1} K_{\mathrm{leb},m}K_{\mathrm{rw}}^l .
\end{align*}
Because there exists a sufficiently large number $r_{\mathrm{rw}} > 0$ such that $\inf_{1 \leq l \leq d_{\mathrm{rw}}} K_{\mathrm{rw}}^l(0,B(0,r_{\mathrm{rw}}))>0$, we define
\begin{equation*}
    \delta_m = \epsilon_{\mathrm{rw}}^{m d_{\mathrm{rw}} + l} \delta^m \mu_{\mathrm{leb}}(B(0,R/4))^{m-1} .
\end{equation*}
Then we have that when $(m+1)R/2>r_{\mathrm{rw}}$, the inequality holds that
\begin{align*}
    P_\theta^{m d_{\mathrm{rw}} + l}(\Lambda, A)
    & \geq \delta_m \int K_{\mathrm{leb},m}(x, A)
    K_{\mathrm{rw}}^l(\Lambda, dx) \\
    & \geq \delta_m \int_{B(\Lambda,r_{\mathrm{rw}})}
    \left[ \int_A \chi_m(y-x) \mu_{\mathrm{leb}}(dy) \right] K_{\mathrm{rw}}^l(\Lambda, dx) \\
    &= \delta_m \int_{B(\Lambda,r_{\mathrm{rw}})}
    \mu_{\mathrm{leb}}((A-x) \cap B(0,(m+1)R/2))
    K_{\mathrm{rw}}^l(\Lambda, dx) \\
    &= \delta_m \int_{B(\Lambda,r_{\mathrm{rw}})}
    \mu_{\mathrm{leb}}((A-\Lambda) \cap B(x-\Lambda,(m+1)R/2))
    K_{\mathrm{rw}}^l(\Lambda, dx) \\
    & \geq \delta_m \int_{B(\Lambda,r_{\mathrm{rw}})} \mu_{\mathrm{leb}}((A-\Lambda) \cap B(0,(m+1)R/2-r_{\mathrm{rw}}))
    K_{\mathrm{rw}}^l(\Lambda, dx) \\
    &= \delta_m K_{\mathrm{rw}}^l(0,B(0,r_{\mathrm{rw}})) \mu_{\mathrm{leb}}(A \cap B(\Lambda,(m+1)R/2-r_{\mathrm{rw}})) ,
\end{align*}
where the set $A - \Lambda := \{ x_A - \Lambda \mid x_A \in A \}$.

Thus, we can conclude that
\begin{equation}
    P_\theta^{m d_{\mathrm{rw}} + l}(\Lambda, \cdot) \geq \delta_m \inf_{1 \leq l \leq d_{\mathrm{rw}}} K_{\mathrm{rw}}^l(0,B(0,r_{\mathrm{rw}})) \mu_{\mathrm{leb}}(\cdot \cap B(\Lambda,(m+1)R/2-r_{\mathrm{rw}}))
    \label{eq_irreducible_key_inequality} .
\end{equation}

The inequality implies that for any bounded measurable set $A$ with positive Lebesgue measure, if $n$ is sufficiently large, $A$ is a $\nu_n$-small set and therefore a petite set, where $\nu_n$ is a scalar multiple of the Lebesgue measure restricted to any preassigned fixed bounded set.
Therefore, we conclude that $\mu_{\mathrm{leb}}$ is an irreducible measure for $P_\theta$.

Moreover, due to the absolute continuity of the irreducible measure with respect to the maximal irreducible measure, the Markov chain $\{\Lambda_n\}$ is aperiodic according to the definition in (\cite{meynMarkovChainsStochastic2009}), even though the maximal irreducibility measure has not been explicitly specified.
In summary, $P_\theta$ is aperiodic and irreducible for any $\theta$.
In addition, the inequality in Lemma \ref{lemma_spread_out_on_sphere} is a variant of (\ref{eq_irreducible_key_inequality}), with $d = m d_{\mathrm{rw}} + l$.
Besides, the lower bound for $d = m d_{\mathrm{rw}} + l$ can be derived from the lower bound for $m$, which follows from the inequality $(m+1)R/2>r_{\mathrm{rw}}$.

\textbf{Maximal Irreducibility Measure: Lebesgue Measure}

At the beginning, we note the equivalence between $P_\theta$ and $K_{\mathrm{rw}}$:
\begin{equation*}
    \frac{1}{\epsilon_{\mathrm{rw}}}K_{\mathrm{rw}}(\Lambda, \cdot) \geq P_\theta \geq \epsilon_{\mathrm{rw}} K_{\mathrm{rw}}(\Lambda, \cdot) ,
\end{equation*}

Denote $\Gamma_n = \Gamma_\rho^{n*}$. Then $K_{\mathrm{rw}}^n(\Lambda, \cdot) = \Gamma_n(\cdot-\Lambda)$. So for any Lebesgue null set $A$, we can show that for any positive integer $n$,
\begin{align*}
    \int \mu_{\mathrm{leb}}(d\Lambda) K_{\mathrm{rw}}^n(\Lambda, A)
    &= \int \mu_{\mathrm{leb}}(d\Lambda) \Gamma_n(A-\Lambda)
    = \int \int \mathbb{I}(\Lambda+y \in A) \mu_{\mathrm{leb}}(d\Lambda) \Gamma_n(dy) \\
    &= \int \mu_{\mathrm{leb}}(A-y) \Gamma_n(dy)
    = 0 .
\end{align*}
This equality implies that
\begin{align*}
    \int \mu_{\mathrm{leb}}(d\Lambda) K_{a_{\frac{1}{2}}}(\Lambda, A)
    &= \int \mu_{\mathrm{leb}}(d\Lambda) \sum_{n=0}^{\infty} 2^{-(n+1)} P_\theta^n(\Lambda, A) \\
    & \leq \int \mu_{\mathrm{leb}}(d\Lambda) \sum_{n=0}^{\infty} 2^{-(n+1)}\epsilon_{\mathrm{rw}}^{-n} K_{\mathrm{rw}}^n(\Lambda, A) \\
    &= \sum_{n=0}^{\infty} 2^{-(n+1)}\epsilon_{\mathrm{rw}}^{-n} \int \mu_{\mathrm{leb}}(d\Lambda) K_{\mathrm{rw}}^n(\Lambda, A) \\
    &= 0 .
\end{align*}
This shows that $\int \mu_{\mathrm{leb}}(d\Lambda) K_{a_{\frac{1}{2}}}(\Lambda, \cdot)$ is absolutely continuous with respect to $\mu_{\mathrm{leb}}$.
Because $\mu_{\mathrm{leb}}$ is a $\sigma$-finite irreducible measure for $P_\theta$, from Proposition 4.2.2 and the definition of the maximal irreducibility measure in (\cite{meynMarkovChainsStochastic2009}), we can show that $\int \mu_{\mathrm{leb}}(d\Lambda) K_{a_{\frac{1}{2}}}(\Lambda, \cdot)$ is the maximal irreducibility measure.
Moreover, Proposition 4.2.2 also implies that $\mu_{\mathrm{leb}}$ is absolutely continuous with respect to $\int \mu_{\mathrm{leb}}(d\Lambda) K_{a_{\frac{1}{2}}}(\Lambda, \cdot)$.
In summary, $\mu_{\mathrm{leb}}$ is equivalent to $\int \mu_{\mathrm{leb}}(d\Lambda) K_{a_{\frac{1}{2}}}(\Lambda, \cdot)$, hence it is the maximal irreducibility measure.

\subsection{Proof of Theorem \ref{theorem_simultaneous_geometric_ergodicity}}

\textbf{Positive Recurrence}

We have assumed the parameter $\theta \in K_{M,\Delta}$ for some positive numbers $M$ and $\Delta$ to bound the variation of $\Lambda$ in the Markov chain. With Assumption \ref{assumption_x_sub_exponential_bound}, we can derive the following inequality from Theorem \ref{theorem_drift_condition_inequality}:
\begin{equation*}
    \mathbb{E}_\theta \left[ e^{\lambda_1 \|\Lambda_1\|} \mid \Lambda_0 = \Lambda \right] \leq \beta e^{\lambda_1 \|\Lambda_0\|} + b ,
\end{equation*}
for some positive numbers $\beta<1$, $b$ and $\lambda_1$.

For positive recurrence, we need Assumption \ref{assumption_x_spread_out} to ensure the validity of Lemma \ref{lemma_spread_out_on_sphere}.
Based on Lemma \ref{lemma_spread_out_on_sphere} and the inequality above, we can prove $\{\Lambda_n\}$ is positive recurrent by Theorem 11.3.4 in (\cite{meynMarkovChainsStochastic2009}).
Theorem 10.4.9 in (\cite{meynMarkovChainsStochastic2009}) implies that the invariant probability measure $\pi_\theta$ for $P_\theta$ is unique and equivalent to the maximal irreducibility measure $\mu_{\mathrm{leb}}$.
This result implies that $P_\theta$ is $\pi_\theta$-irreducible.

\textbf{Simultaneous Geometric Ergodicity for $\theta$}

Theorem \ref{theorem_drift_condition_inequality} implies that there exist positive constants $\lambda_1$, $\beta<1$ and $b$ such that for any $\theta \in K_{M,\Delta}$,
\begin{equation*}
    P_\theta V \leq \beta V + b ,
\end{equation*}
where the Lyapunov function $V(\Lambda) = \exp(\lambda_1 \|\Lambda\|)$.

Since $V(\Lambda) \leq 2b(1-\beta)^{-1}$ when $\|\Lambda\| \leq \ln\left[2b(1-\beta)^{-1}\right]/\lambda_1$, denote the constant
\begin{equation*}
    c_\Lambda = \ln\left[2b(1-\beta)^{-1}\right]/\lambda_1 > 0 .
\end{equation*}
Let $d_{\mathrm{small}}$ denote the smallest positive integer greater than $\max(\frac{2c_\Lambda-b_{d,P}}{a_{d,P}}, d_{\mathrm{min}})$, where $d_{\mathrm{min}}$ is the constant in Lemma \ref{lemma_spread_out_on_sphere}.
By Lemma \ref{lemma_spread_out_on_sphere} with $d = d_{\mathrm{small}}$, it holds that for any $\theta \in K_{M,\Delta}$,
\begin{equation}
    P_\theta^{d}(\Lambda, \cdot) \geq \delta_{d,P} \mu_{\mathrm{leb}}(\cdot \cap B(\Lambda,a_{d,P} d + b_{d,P})) \geq \delta_{d, P} \mu_{\mathrm{leb}}(\cdot \cap B(\Lambda, 2c_\Lambda))
    \label{eq_simultaneous_ergodicity_small_set_condition} ,
\end{equation}
and
\begin{equation}
    P_\theta^d V \leq \beta^d V + b\frac{1-\beta^d}{1-\beta}
    \label{eq_simultaneous_ergodicity_inequality_condition} .
\end{equation}

Combining the identity
\begin{equation*}
    \ln\left[2b\frac{1-\beta^d}{1-\beta}(1-\beta^d)^{-1}\right]/\lambda_1 = \ln\left[2b(1-\beta)^{-1}\right]/\lambda_1 = c_\Lambda ,
\end{equation*}
with the two inequalities (\ref{eq_simultaneous_ergodicity_small_set_condition}) and (\ref{eq_simultaneous_ergodicity_inequality_condition}), the condition in Lemma \ref{lemma_simultaneous_geometric_ergodicity} is satisfied with the $d$-step transition probability kernel $P_\theta^d$ for $\theta \in S_\Theta = K_{M,\Delta}$.
Therefore, we can conclude that there exists some constant $L_{\mathrm{small}}>1$, as defined in Lemma \ref{lemma_simultaneous_geometric_ergodicity}, such that for any $\theta \in K_{M,\Delta}$,
\begin{equation*}
    \| P_\theta^{d n}(\Lambda, \cdot) - \pi_\theta \|_V \leq L_{\mathrm{small}}(1-L_{\mathrm{small}}^{-1})^n V(\Lambda) .
\end{equation*}

Furthermore, we extend the inequality for $P_\theta^{d n}$ to arbitrary powers $P_\theta^n$.
For any nonnegative integer $l$, the following inequality holds:
\begin{align*}
\| P_\theta^{d n + l}(\Lambda, \cdot) - \pi_\theta \|_V
&= \| (P_\theta^{d n}(\Lambda, \cdot) - \pi_\theta ) P_\theta^{l} \|_V \\
& \leq \| P_\theta^{d n}(\Lambda, \cdot) - \pi_\theta \|_{\beta^l V + b\frac{1-\beta^l}{1-\beta}} \\
& \leq \| P_\theta^{d n}(\Lambda, \cdot) - \pi_\theta \|_{ \left( \beta^l + b\frac{1-\beta^l}{1-\beta} \right) V } \\
& \leq \left( \beta^l + b\frac{1-\beta^l}{1-\beta} \right) \| P_\theta^{d n}(\Lambda, \cdot) - \pi_\theta \|_V \\
& \leq (1+\frac{b}{1-\beta}) \| P_\theta^{d n}(\Lambda, \cdot) - \pi_\theta \|_V ,
\end{align*}
where the first inequality can be derived by
\begin{equation*}
    P_\theta^l V
    \leq P_\theta^{l-1} (\beta V + b)
    \leq \dots
    \leq \beta^l V + b\frac{1-\beta^l}{1-\beta} .
\end{equation*}

Let $n=n^\prime d+l$ with $l \in \{0,\dots,d-1\}$. Using the inequality $n^\prime \geq \frac{n}{d}-1$, it holds that
\begin{equation*}
    \| P_\theta^{n}(\Lambda, \cdot) - \pi_\theta \|_V \leq (1+\frac{b}{1-\beta})L_{\mathrm{small}}(1-L_{\mathrm{small}}^{-1})^{\frac{n}{d}-1} V(\Lambda) .
\end{equation*}

Thus, for the constant $L$ in Theorem \ref{theorem_simultaneous_geometric_ergodicity}, we set $L = \max\{C,(1-\rho)^{-1}\}$, where $C = (1+\frac{b}{1-\beta})L_{\mathrm{small}}(1-L_{\mathrm{small}}^{-1})^{-1}$ and $\rho = (1-L_{\mathrm{small}}^{-1})^{\frac{1}{d}}$.

\textbf{Boundedness of $\pi_\theta V$}

For any $\theta \in K_{M,\Delta}$,
\begin{equation*}
    P_\theta V \leq \beta V + b
\end{equation*}
implies
\begin{equation*}
    \mathbb{E} V(\Lambda_n)
    \leq \beta \mathbb{E} V(\Lambda_{n-1}) + b
    \leq \dots
    \leq \beta^n \mathbb{E} V(\Lambda_0) + b\frac{1-\beta^l}{1-\beta}
\end{equation*}
and
\begin{equation*}
    \pi_\theta V
    = (\pi_\theta P_\theta) V
    = \pi_\theta (P_\theta V)
    \leq \beta \pi_\theta V + b .
\end{equation*}
Thus, $\mathbb{E} V(\Lambda_n) \leq \max\{\frac{b}{1-\beta}, \mathbb{E} V(\Lambda_0)\}$ and $\pi_\theta V \leq \frac{b}{1-\beta}$.
Denote the constant $C_V = \frac{b}{1-\beta}$.

\subsection{Proof of Theorem \ref{theorem_boundedness_shift}}

\label{subsec_proof_theorem_boundedness_shift}

Consider the shift on the additional covariate $Y$ defined by $\mathbb{E}_\theta \left[ \sum_{i=1}^n (T_i-\rho)Y_i \right]$.
This satisfies the inequality
\begin{align*}
    \left| \mathbb{E}_\theta \left[ \sum_{i=1}^n (T_i-\rho)Y_i \mid \Lambda_0 = \Lambda \right] - n \pi_\theta h_\theta \right|
    &= \left| \sum_{i=1}^n P_\theta^{i-1}(\Lambda, h_\theta) - n \pi_\theta h_\theta \right| \\
    & \leq \sum_{i=1}^n \left|P_\theta^{i-1}(\Lambda, \cdot) - \pi_\theta \right|(|h_\theta|) ,
\end{align*}
where $h_\theta(\Lambda) = \mathbb{E}_\theta \left[ (T_1-\rho)Y_1 \mid \Lambda_0 = \Lambda \right] = \mathbb{E}_{X \sim \Gamma} \left[ [g_\theta(\Lambda, X) - \rho] f(X) \right]$.

Moreover, $|h_\theta(\Lambda)| \leq \mathbb{E}_\theta \left[ |Y_1| \mid \Lambda_0 = \Lambda \right] \leq \mathbb{E}|Y|$, so $h_\theta(\Lambda)$ is bounded.

Under Assumptions \ref{assumption_x_sub_exponential_bound} and \ref{assumption_x_spread_out}, Theorem \ref{theorem_simultaneous_geometric_ergodicity} implies that, for any $\theta \in K_{M,\Delta}$, we have the geometric ergodicity bound
\begin{equation*}
    \| P_\theta^{n}(\Lambda, \cdot) - \pi_\theta \|_V \leq L(1-L^{-1})^n V(\Lambda) .
\end{equation*}
It follows that
\begin{equation*}
    |P^{i-1}(\Lambda, \cdot) - \pi_\theta |(|h_\theta|) \leq \left[\mathbb{E}|Y|\right] \| P_\theta^{i-1}(\Lambda, \cdot) - \pi_\theta \|_V \leq \left[\mathbb{E}|Y|\right] L(1-L^{-1})^{i-1} V(\Lambda) .
\end{equation*}

Thus, for any $n$ and any $\Lambda$,
\begin{equation*}
    \sum_{i=1}^n \left|P_\theta^{i-1}(\Lambda, \cdot) - \pi_\theta \right|(|h_\theta|)
    \leq \sum_{i=1}^n \left[\mathbb{E}|Y|\right] L(1-L^{-1})^{i-1} V(\Lambda)
    \leq L^2 \left[\mathbb{E}|Y|\right] V(\Lambda) < \infty .
\end{equation*}

It implies that $\mathbb{E}_\theta \left[ \sum_{i=1}^n (T_i-\rho)Y_i \mid \Lambda_0 = \Lambda \right] - n \pi_\theta h_\theta$ is $O(1)$, and it also converges to a finite limit due to the bound of each term.

\section{Proofs for Existing Covariate Adaptive Randomization Procedures}

\subsection{Proof of Theorem \ref{theorem_shift_liu_xiao}}

\label{subsec_proof_theorem_shift_liu_xiao}

We view the procedure as starting from the allocation of the second unit, with the initial imbalance vector $\Lambda = (T_1 - \rho) X_1$.
It is straightforward to verify that this shift does not affect the asymptotic properties.

We first seek constants $M > 0$ and $\Delta > 0$ such that for any $\Lambda \in W_\Gamma$ with $\|\Lambda\| \geq M$, the inequality holds that
\begin{equation*}
    \mathbb{E}_\theta \left[ (\Lambda_1 - \Lambda_0)^T \frac{\Lambda_0}{\|\Lambda_0\|} \mid \Lambda_0 = \Lambda \right]
    \leq -\Delta .
\end{equation*}
For the allocation function of the randomization procedure minimizing the imbalance measure, the left-hand side of the inequality can be written as
\begin{align}
    & \quad \mathbb{E}_{X \sim \Gamma} \left[ (g_\theta(\Lambda, X)-\rho)\frac{X^T \Lambda}{\|\Lambda\|} \right]
    = \mathbb{E}_{X \sim \Gamma} \left[ (g_\theta(\Lambda, X)-\rho)\frac{X^T \Lambda}{\|\Lambda\|} \right] \label{eq_liu_negative_feedback_condition} \\
    &= \mathbb{E}_{X \sim \Gamma} \left[ \left[ \mathbb{I}\left( D(\Lambda, X)<0 \right)(\rho_1-\rho)
    + \mathbb{I}\left( D(\Lambda, X)>0 \right)(1-\rho_1-\rho) \right]
    \frac{X^T \Lambda}{\|\Lambda\|} \right] , \notag
\end{align}
where the function $D(\Lambda, X)$ is defined as $D(\Lambda, X) = 2X^T\Lambda+(1-2\rho)X^T X$.
The difference between the right-hand side of (\ref{eq_liu_simple}) and a simpler expression
\begin{equation}
    \mathbb{E}_{X \sim \Gamma} \left[ \left[ \mathbb{I}\left(2X^T\Lambda<0\right)(\rho_1-\rho)
    + \mathbb{I}\left(2X^T\Lambda>0\right)(1-\rho_1-\rho) \right]
    \frac{X^T \Lambda}{\|\Lambda\|} \right] \label{eq_liu_simple}
\end{equation}
can be bounded by
\begin{align*}
    & \quad \mathbb{E}_{X \sim \Gamma} \left[ \mathbb{I}\left(2|X^T\Lambda| \leq |(1-2\rho)X^T X|\right)
    \frac{|X^T\Lambda|}{\|\Lambda\|} \right] \\
    & \leq \mathbb{E}_{X \sim \Gamma} \left[ \mathbb{I}\left(2|X^T\Lambda| \leq |(1-2\rho)X^T X|\right)
    \frac{|(1-2\rho)X^T X|}{2\|\Lambda\|} \right] \\
    & \leq \mathbb{E}_{X \sim \Gamma} \left[
    \frac{|(1/2-\rho)X^T X|}{\|\Lambda\|} \right]
    = \frac{|\rho-1/2|\mathbb{E}\|X\|^2}{\|\Lambda\|} .
\end{align*}
This holds because, whenever $2|X^T\Lambda| > |(1-2\rho)X^T X|$, the two quantities $2X^T\Lambda$ and $D(\Lambda, X) = 2X^T\Lambda+(1-2\rho)X^T X$ share the same sign.

Because $\rho_1>\rho$ and $1-\rho_1<\rho$, (\ref{eq_liu_simple}) can be bounded by
\begin{equation}
    \mathbb{E}_{X \sim \Gamma} \left[ -\min\{|\rho_1-\rho|,|1-\rho_1-\rho|\}
    \frac{|X^T\Lambda|}{\|\Lambda\|} \right]
    \label{eq_inner_product_X_Lambda} ,
\end{equation}
which can be further bounded above by a negative constant.
The proof is provided at the end of this subsection.
Denote this negative constant by $-\Delta^\prime$ with $\Delta^\prime > 0$.
\begin{align*}
    \mathbb{E}_\theta \left[ (\Lambda_1 - \Lambda_0)^T \frac{\Lambda_0}{\|\Lambda_0\|} \mid \Lambda_0 = \Lambda \right]
    & \leq \frac{|\rho-1/2|\mathbb{E}\|X\|^2}{\|\Lambda\|} - \Delta^\prime
    \leq \frac{|\rho-1/2|\mathbb{E}\|X\|^2}{M} - \Delta^\prime \\
    & < \frac{\Delta^\prime}{2} - \Delta^\prime = -\Delta .
\end{align*}
Thus, $\theta \in K_{M,\Delta}$ with $M = \frac{|2\rho-1|\mathbb{E}\|X\|^2}{\Delta^\prime}+1>0$ and $\Delta=\Delta^\prime/2>0$.

Since $\theta \in K_{M,\Delta}$, by Theorem \ref{theorem_boundedness_shift}, the average shift on $Y$ converges to $\pi_\theta h_\theta$.
Therefore, the average shift on the additional covariate $Y_n = f(X_n) + \epsilon_n$ converges to $\pi_\theta h_\theta = \mathbb{E}_{X \sim \Gamma} \left[ [\pi_\theta \left[ g_\theta(\cdot,X) \right] - \rho] f(X) \right]$.
If $\rho>\frac{1}{2}$, then
\begin{equation*}
    \pi_\theta(\{ \Lambda \mid 2x^T\Lambda+(1-2\rho)x^Tx < 0 \})
    \geq \pi_\theta(\{ \Lambda \mid \|\Lambda\| < (\rho-1/2)\|x\| \})
    \rightarrow 1 ,
\end{equation*}
as $\|x\| \rightarrow \infty$.

Moreover, when $x \neq 0$, $\pi_\theta(\{ \Lambda \mid 2x^T\Lambda+(1-2\rho)x^Tx = 0 \}) = 0$.
A technical justification is provided below for completeness.
Lemma \ref{lemma_spread_out_on_sphere} implies that $\mu_{\mathrm{leb}}$ is the maximal irreducibility measure.
Theorem 10.4.9 in (\cite{meynMarkovChainsStochastic2009}) then ensures that the invariant probability measure $\pi_\theta$ for $P_\theta$ is unique and equivalent to the maximal irreducibility measure $\mu_{\mathrm{leb}}$.
Consequently, when $x \neq 0$, the set $\{ \Lambda \mid 2x^T\Lambda+(1-2\rho)x^Tx = 0 \}$ is a Lebesgue null set, and hence $\pi_\theta(\{ \Lambda \mid 2x^T\Lambda+(1-2\rho)x^Tx = 0 \}) = 0$.

Thus, by $\rho_1>\rho$ and $1-\rho_1<\rho$, there exists $M_L>0$ such that the infimum below is strictly positive:
\begin{align*}
    & \quad \inf_{\|x\| \geq M_L} \left\{ \pi_\theta \left[ g_\theta(\cdot, x) \right] - \rho \right\} \\
    &= \inf_{\|x\| \geq M_L} \left[ \rho_1 \pi_\theta(\{ \Lambda \mid 2x^T\Lambda+(1-2\rho)x^Tx < 0 \}) \right. \\
    & \quad \left. + (1-\rho_1) \pi_\theta(\{ \Lambda \mid 2x^T\Lambda+(1-2\rho)x^Tx > 0 \}) -\rho \right] \\
    &= (\rho_1-\rho) \inf_{\|x\| \geq M_L}\pi_\theta(\{ \Lambda \mid 2x^T\Lambda+(1-2\rho)x^Tx < 0 \}) \\
    & \quad + (1-\rho_1-\rho) [1-\inf_{\|x\| \geq M_L}\pi_\theta(\{ \Lambda \mid 2x^T\Lambda+(1-2\rho)x^Tx < 0 \})] .
\end{align*}
Let the function $f(x) = \mathbb{I}(\|x\| \geq M_L)$.
Then the limit of the average shift, $\pi_\theta h_\theta$, satisfies the following inequality:
\begin{align*}
    \pi_\theta h_\theta
    &= \mathbb{E}_{X \sim \Gamma} \left[ [\pi_\theta \left[ g_\theta(\cdot,X) \right] - \rho] f(X) \right] \\
    &= \mathbb{E}_{X \sim \Gamma} \left[ \left\{ \rho_1 \pi_\theta(\{ \Lambda \mid 2X^T\Lambda+(1-2\rho)X^TX < 0 \}) \right. \right. \\
    & \quad \left. \left. + (1-\rho_1) \pi_\theta(\{ \Lambda \mid 2X^T\Lambda+(1-2\rho)X^TX > 0 \})
    - \rho \right\} \mathbb{I}(\|X\| \geq M_L)
    \right] \\
    &= \mathbb{E}_{X \sim \Gamma} \left[ (\rho_1-\rho) \mathbb{I}(\|X\| \geq M_L) \pi_\theta(\{ \Lambda \mid 2X^T\Lambda+(1-2\rho)X^TX < 0 \}) \right. \\
    & \quad + \left.
        (1-\rho_1-\rho) \mathbb{I}(\|X\| \geq M_L) [1-\pi_\theta(\{ \Lambda \mid 2X^T\Lambda+(1-2\rho)X^TX < 0 \})]
    \right] \\
    &= \mathbb{E}_{X \sim \Gamma} \left[ \mathbb{I}(\|X\| \geq M_L)\left[
        (\rho_1-\rho) \pi_\theta(\{ \Lambda \mid 2X^T\Lambda+(1-2\rho)X^TX < 0 \}) \right. \right. \\
    & \quad + \left. \left.
        (1-\rho_1-\rho) [1-\pi_\theta(\{ \Lambda \mid 2X^T\Lambda+(1-2\rho)X^TX < 0 \})]
    \right] \right] \\
    &> 0 .
\end{align*}

The case $\rho<\frac{1}{2}$ follows analogously.

Finally, we prove \eqref{eq_inner_product_X_Lambda}.
Denote the support of the distribution of $\Gamma$ by $A_{\mathrm{Supp},\Gamma}$ and the linear subspace spanned by $A_{\mathrm{Supp},\Gamma}$ by $W_\Gamma$.
Suppose the dimension of $W_\Gamma$ is $d_W \leq d$.
Then there exist $x_1, \dots, x_{d_W} \in A_{\mathrm{Supp},\Gamma}$ that are linearly independent.
Thus, as in the proof of Condition \ref{condition_2_epsilon} for $\epsilon(\theta)$ (see Subsection \ref{subsec_proof_lemma_instance_established}), we can show that there exists a positive number $\epsilon$ such that
\begin{equation*}
    \bigcup_{i=1}^{d_W} \left[ R_{x_i}^{\epsilon} \cup -R_{x_i}^{\epsilon} \right] = W_\Gamma \backslash \{0\} .
\end{equation*}
This equality is equivalent to the statement that for any $\Lambda \in W_\Gamma$, the inequality
\begin{equation*}
    |x_{i_\Lambda}^T\Lambda| \geq \epsilon \|x_{i_\Lambda}\|\|\Lambda\|
\end{equation*}
holds for some $i_\Lambda \in \{1,\dots,d_W\}$.

Let $C_i = \{X \mid \|X-x_i\| < \epsilon\min_i\{\|x_i\|\}/2 \}$ be the open ball around $x_i$.
Then $x_i \in A_{\mathrm{Supp},\Gamma}$ implies $\Gamma(C_i)>0$. For any $\Lambda$, for any $y \in C_{i_\Lambda}$, we have
\begin{equation*}
    |y^T \Lambda|
    \geq |x_{i_\Lambda}^T\Lambda| - |(y-x_{i_\Lambda})^T\Lambda|
    \geq \epsilon \|x_{i_\Lambda}\|\|\Lambda\| - \epsilon\min_i\{\|x_i\|\}\|\Lambda\|/2
    \geq \epsilon\min_i\{\|x_i\|\}\|\Lambda\|/2
\end{equation*}

Thus,
\begin{align*}
    \mathbb{E}_\theta \left[ (\Lambda_1 - \Lambda_0)^T \frac{\Lambda_0}{\|\Lambda_0\|} \mid \Lambda_0 = \Lambda \right]
    &= \mathbb{E}_{X \sim \Gamma} \left[ -p\frac{|X^T \Lambda|}{\|\Lambda\|} \right]
    \leq \mathbb{E}_{X \sim \Gamma} \left[ -p\frac{|X^T \Lambda|}{\|\Lambda\|} \mathbb{I}(X \in C_{i_\Lambda}) \right] \\
    & \leq \mathbb{E}_{X \sim \Gamma} \left[ -\frac{p\epsilon\min_i\{\|x_i\|\}}{2} \mathbb{I}(X \in C_{i_\Lambda}) \right] \\
    & \leq -p\epsilon\min_i\{\Gamma(C_i)\}\min_i\{\|x_i\|\}/2
    < 0 .
\end{align*}
It implies that $\theta \in K_{M,p\epsilon\min_i\{\Gamma(C_i)\}\min_i\{\|x_i\|\}/2}$ for any $M>0$.

\section{Proofs for Oracle Randomization Procedure}

\subsection{Negative Feedback Condition}

The following lemma establishes the negative feedback condition, which corresponds to the first part of the conclusion in Theorem \ref{theorem_shift_simultaneous_geometric_finite_adjusted}.
\begin{lemma}
    \label{lemma_drift_condition}
    Suppose Assumption \ref{assumption_theta_basic} for components of the allocation function holds. If Assumption \ref{assumption_x_sub_exponential_bound} holds, then there exist positive numbers $M$, $\Delta$ and $r_d$ such that for any $\theta \in \overline{B(\theta^*,r_d)}$,
    \begin{equation*}
        \mathbb{E}_\theta \left[(\Lambda_1 - \Lambda_0)^T \frac{\Lambda_0}{\|\Lambda_0\|} \mid \Lambda_0 = \Lambda \right] \leq - \Delta,
    \end{equation*}
    for any $\Lambda \in W_\Gamma = \mathbb{R}^d$ with $\|\Lambda\| \geq M$.
    Consequently, $\overline{B(\theta^*,r_d)} \subset K_{M,\Delta}$.
\end{lemma}

\begin{proof}[Proof of Lemma \ref{lemma_drift_condition}]
    Even without Lemma \ref{lemma_spread_out_on_sphere}, using Item \ref{item_assumption_theta_basic_1} of Assumption \ref{assumption_theta_basic}, one can prove by contradiction that $W_\Gamma = \mathbb{R}^d$. We omit the details, as they are not essential to the main argument.

    We begin the main argument by noting the following inequality:
    \begin{align}
        & \quad \mathbb{E}_\theta \left[(\Lambda_1 - \Lambda_0)^T \frac{\Lambda_0}{\|\Lambda_0\|} \mid \Lambda_0 = \Lambda \right] \label{eq_oracle_negative_feedback_1} \\
        &= \mathbb{E}_\theta \left[(T_1-\rho)X_1^T \frac{\Lambda_0}{\|\Lambda_0\|} \mid \Lambda_0 = \Lambda \right]
        = \int {(g_\theta(\Lambda, X)-\rho)X^T \frac{\Lambda}{\|\Lambda\|} \Gamma(dX)} \notag \\
        &= \frac{p}{d} \sum_{i=1}^d \beta_{\xi_i}^{\epsilon(\theta)}(\Lambda) \left[ \int {\alpha_i(X) X \Gamma(dX)} \right]^T \frac{\Lambda}{\|\Lambda\|}
        = \frac{p}{d} \sum_{i=1}^d \beta_{\xi_i}^{\epsilon(\theta)}(\Lambda) \frac{{\xi_i^*}^T \Lambda}{\|\Lambda\|} \notag \\
        &= \frac{p}{d} \sum_{i=1}^d \beta_{\xi_i}^{\epsilon(\theta)}(\Lambda) \frac{{\xi_i}^T \Lambda}{\|\Lambda\|}
        + \frac{p}{d} \sum_{i=1}^d \beta_{\xi_i}^{\epsilon(\theta)}(\Lambda) \frac{(\xi_i^* - \xi_i)^T \Lambda}{\|\Lambda\|} \notag \\
        & \leq \frac{p}{d} \sum_{i=1}^d \beta_{\xi_i}^{\epsilon(\theta)}(\Lambda) \frac{\xi_i^T \Lambda}{\|\Lambda\|}
        + \frac{p}{d} \sum_{i=1}^d \left| \beta_{\xi_i}^{\epsilon(\theta)}(\Lambda) \right| \|\xi_i^* - \xi_i\| \notag \\
        & \leq \frac{p}{d} \sum_{i=1}^d \beta_{\xi_i}^{\epsilon(\theta)}(\Lambda) \frac{\xi_i^T \Lambda}{\|\Lambda\|} + p \|\theta-\theta^*\| \notag .
    \end{align}
    The last inequality holds because $\left| \beta_{\xi_i}^{\epsilon(\theta)}(\Lambda) \right| \leq 1$ and $\|\xi_i^* - \xi_i\| \leq \|\theta-\theta^*\|$.

    To bound the inner product of $\xi_i$ and $\Lambda$, we denote by $\epsilon^*(\theta)$ the function $\epsilon(\theta)$ defined in Subsection \ref{subsec_proof_lemma_instance_established}.
    Its existence and properties, including satisfaction of Assumptions \ref{assumption_theta_basic} and \ref{assumption_theta_lipschitz}, are established in Subsection \ref{subsec_proof_lemma_instance_established}.
    Item \ref{item_assumption_theta_basic_2} of Assumption \ref{assumption_theta_basic} implies that $\beta_{\xi}^{\epsilon}(\Lambda) \leq 0$ when $\xi^T \Lambda \geq 0$, and Item \ref{item_assumption_theta_basic_3} implies that $\beta_{\xi}^{\epsilon}(\Lambda) = -\beta_{-\xi}^{\epsilon}(\Lambda) = -\beta_{\xi}^{\epsilon}(-\Lambda)$.
    Hence, $\beta_{\xi}^{\epsilon}(\Lambda) \geq 0$ when $\xi^T \Lambda \leq 0$.
    Therefore,
    \begin{equation*}
        \beta_{\xi_i}^{\epsilon(\theta)}(\Lambda) \frac{{\xi_i}^T \Lambda}{\|\Lambda\|} \leq 0 .
    \end{equation*}

    Suppose $d_{\mathrm{uni}} = \epsilon^*(\theta^*)/2$.
    Then there exists $M > 0$ such that for any $\epsilon \geq d_{\mathrm{uni}}$, $\epsilon^\prime \in (0,\epsilon]$ and $\Lambda \in R_\xi^\epsilon$ with $\|\Lambda\| \geq M$, $\beta_{\xi}^{\epsilon^\prime}(\Lambda) = -1$.
    Due to the continuity of the function $\epsilon^*$, there exists a radius $r$ such that for any $\theta \in B(\theta^*,r)$, we have $\epsilon^*(\theta) \geq d_{\mathrm{uni}}$.

    Moreover, $R_{\xi}^{\epsilon} \subset R_{\xi}^{\epsilon^\prime}$ for any $\epsilon > \epsilon^\prime$ and $\xi \neq 0$.
    Consider any $\Lambda \in W_\Gamma$ and $\theta \in B(\theta^*,r)$.
    Because $\bigcup_{i=1}^{d} \left[ R_{\xi_i}^{\epsilon} \cup -R_{\xi_i}^{\epsilon} \right]  = \mathbb{R}^d \backslash \{0\}$ when $\epsilon = \epsilon(\theta)$ or $\epsilon^*(\theta)$, there exists $i_0$ such that
    \begin{equation*}
        \Lambda \in \pm R_{\xi_{i_0}}^{\max\{\epsilon(\theta), \epsilon^*(\theta)\}} = \pm R_{\xi_{i_0}}^{\epsilon(\theta)} \cap R_{\xi_{i_0}}^{\epsilon^*(\theta)} .
    \end{equation*}
    Then for some $M>0$, for any $\Lambda \in W_\Gamma$ with $\|\Lambda\| \geq M$, there exists some $i_0$ such that the inequality
    \begin{equation}
        \beta_{\xi_{i_0}}^{\epsilon(\theta)}(\Lambda) \frac{{\xi_{i_0}}^T \Lambda}{\|\Lambda\|}
        = -\frac{\left| {\xi_{i_0}}^T \Lambda \right|}{\|\Lambda\|}
        \leq -\frac{\epsilon^*(\theta) \|\xi_{i_0}\|\|\Lambda\|}{\|\Lambda\|}
        = -\epsilon^*(\theta) \|\xi_{i_0}\| \label{eq_oracle_negative_feedback_2}
    \end{equation}
    holds.
    This is because $\Lambda \in \pm R_{\xi_{i_0}}^{\max\{\epsilon(\theta), \epsilon^*(\theta)\}} \subset \pm R_{\xi_{i_0}}^{\epsilon^*(\theta)}$ implies $\left| {\xi_{i_0}}^T \Lambda \right| \geq \epsilon^*(\theta) \|\xi_{i_0}\|\|\Lambda\|$, and $\Lambda \in \pm R_{\xi_{i_0}}^{\max\{\epsilon(\theta), \epsilon^*(\theta)\}}$ implies $|\beta_{\xi_{i_0}}^{\epsilon(\theta)}(\Lambda)|=1$.

    Combining (\ref{eq_oracle_negative_feedback_1}) and (\ref{eq_oracle_negative_feedback_2}), it follows that
    \begin{align*}
        \mathbb{E}_\theta \left[(\Lambda_1 - \Lambda_0)^T \frac{\Lambda_0}{\|\Lambda_0\|} \mid \Lambda_0 = \Lambda \right]
        & \leq \frac{p}{d} \sum_{i=1}^d \beta_{\xi_i}^{\epsilon(\theta)}(\Lambda) \frac{\xi_i^T \Lambda}{\|\Lambda\|} + p \|\theta-\theta^*\| \\
        & \leq -\frac{\epsilon^*(\theta)p}{d} \min_i \|\xi_i\| + p \|\theta-\theta^*\| ,
    \end{align*}
    for $\Lambda$ with $\|\Lambda\| \geq M$.

    The continuity of $\epsilon^*$ implies that the right side of the inequality is continuous at $\theta^*$. It is negative when $\theta=\theta^*$. Thus, there exists a radius $r_d < r$ such that the right side of the inequality is bounded by a negative number for any $\theta \in \overline{B(\theta^*, r_d)}$.
\end{proof}

\subsection{Proof of Theorem \ref{theorem_shift_simultaneous_geometric_finite_adjusted}}

\label{subsec_proof_theorem_shift_simultaneous_geometric_finite_adjusted}

Based on Lemma \ref{lemma_drift_condition}, we have $\overline{B(\theta^*,r_d)} \subset K_{M,\Delta}$ for some $r_d>0$, $M>0$ and $\Delta>0$.
Therefore, the main part of the theorem's conclusion follows as a corollary from Theorems \ref{theorem_simultaneous_geometric_ergodicity} and \ref{theorem_boundedness_shift}.
The only thing we need to prove is $\pi_\theta h_\theta = 0$.

If $\tilde{\rho}_\theta(x) = \pi_\theta \left[ g_\theta(\cdot, x) \right] = \rho$ for all $x$, then
\begin{align*}
    \pi_\theta h_\theta
    &= \mathbb{E}_{\Lambda \sim \pi_\theta, X \sim \Gamma} \left[ [g_\theta(\Lambda, X) - \rho] f(X) \right]
    = \mathbb{E}_{X \sim \Gamma} \left[ [\pi_\theta \left[ g_\theta(\cdot,X) \right] - \rho] f(X) \right] \\
    &= \mathbb{E}_{X \sim \Gamma} \left[ (\tilde{\rho}_\theta(X) - \rho) f(X) \right]
    = 0 .
\end{align*}
Therefore, it suffices to prove that $\tilde{\rho}_\theta(x) = \rho$ for all $x$.
According to Lemma \ref{lemma_Phi}, we only need to verify:
    \begin{enumerate}[label=(\Alph*)]
        \item $\Phi$ is injective,
        \item the function $\tilde{\rho}$, defined by $\tilde{\rho}(x) \equiv \rho$, belongs to the set $\mathcal{C}$,
    \end{enumerate}
where the mapping $\Phi: \mathcal{C} \to \mathbb{R}^d$ is defined by
\begin{equation*}
    \Phi\left(\tilde{\rho}(\cdot)\right)
    = \mathbb{E}_{X \sim \Gamma} \left[ (\tilde{\rho}(X) - \rho) X \right] ,
\end{equation*}
and $\mathcal{C} = \{ \mu \left[ g_\theta(\cdot, x) \right] \mid \mu \text{ is a probability measure on } \mathbb{R}^d \}$.
Because
\begin{equation*}
    \mu \left[ g_\theta(\cdot, x) \right]
    = \rho + \frac{p}{d} \sum_{i=1}^d { \alpha_i(x) \left[ \mu \beta_{\xi_i}^{\epsilon(\theta)} \right] } ,
\end{equation*}
it follows that
\begin{equation*}
    \mathcal{C} \subset
    \left\{ \rho + \sum_{i=1}^d c_i \alpha_i(\cdot) \mid c_i \in \mathbb{R} \right\} .
\end{equation*}
Denote the latter set by $\tilde{\mathcal{C}}$. We then naturally extend the mapping $\Phi$ to $\tilde{\mathcal{C}}$ by
\begin{equation*}
    \Phi\left(\rho + \sum_{i=1}^d c_i \alpha_i(\cdot)\right)
    = \sum_{i=1}^d c_i \mathbb{E}_{X \sim \Gamma} [\alpha_i(X)X] .
\end{equation*}
Since the vectors $\left\{ \mathbb{E}_{X \sim \Gamma} [\alpha_i(X)X] \mid 1 \leq i \leq d \right\}$ are linearly independent, it follows that $\Phi$ is injective on $\tilde{\mathcal{C}}$.
This confirms the condition (A) in Lemma \ref{lemma_Phi}.
For the condition (B), because Item \ref{item_assumption_theta_basic_3} of Assumption \ref{assumption_theta_basic} implies that $\beta_{\xi}^{\epsilon}(0) = -\beta_{\xi}^{\epsilon}(0)$, we have $\beta_{\xi}^{\epsilon}(0) = 0$.
Thus, for any $x$, it holds that
\begin{equation*}
    \delta_0 \left[ g_\theta(\cdot, x) \right] = g_\theta(0,x) = \rho + \frac{p}{d} \sum_{i=1}^d { \alpha_i(x) \beta_{\xi_i}^{\epsilon(\theta)}(0) } = \rho ,
\end{equation*}
where $\delta_0$ is the Dirac measure at $0$.
Thus, the function $\tilde{\rho}$, defined by $\tilde{\rho}(x) \equiv \rho$, belongs to the set $\mathcal{C}$.

\section{Proofs for Feasible Randomization Procedure}

\subsection{Proof of Lemma \ref{lemma_g_lipschitz}}

\label{subsec_proof_lemma_g_lipschitz}

For some $r_g > 0$, denote the Lipschitz constant of $\epsilon$ and $\beta$ for $\theta \in \overline{B(\theta^*, r_g)}$ by $L_f$.
We have
\begin{align*}
    |g_\theta(\Lambda, X) - g_{\theta^{\prime}}(\Lambda, X)|
    & \leq \frac{p}{d} \sum_{i=1}^d { |\beta_{\xi_i}^{\epsilon(\theta)}(\Lambda) - \beta_{\xi_i^{\prime}}^{\epsilon(\theta^{\prime})}(\Lambda)| } \\
    & \leq \frac{p}{d} \sum_{i=1}^d \left[ L_f\|\xi_i-\xi_i^{\prime}\| + L_f|\epsilon(\theta)-\epsilon(\theta^{\prime})| \right] \\
    & \leq \frac{p}{d} \sum_{i=1}^d \left[ L_f\|\theta - \theta^{\prime}\| + L_f^2\|\theta - \theta^{\prime}\| \right] \\
    &= p(L_f+L_f^2) \|\theta - \theta^{\prime}\|
\end{align*}
Let $L_g = p(L_f+L_f^2)$ be the Lipschitz constant for $g_\theta(\Lambda, X)$ with respect to $\theta$.

\subsection{Proof of Lemma \ref{lemma_instance_established}}

\label{subsec_proof_lemma_instance_established}

\textbf{Proof of the Properties of $\epsilon(\theta)$}

We have set the function $\epsilon(\theta)$ as follows:
\begin{equation*}
    \epsilon(\theta) =
    \begin{cases}
        \frac{1}{\sqrt{d+1} \|A(\theta)^{-1}\|_2}, & \text{if the vectors } \{\xi_1, \dots, \xi_d\} \text{ are linearly independent,} \\
        0, & \text{otherwise,}
    \end{cases}
\end{equation*}
where $\|A(\theta)^{-1}\|_2$ is the operator norm of the matrix $A(\theta)^{-1}$ and the matrix $A(\theta) = \left( \frac{\xi_1}{\|\xi_1\|}, \dots, \frac{\xi_d}{\|\xi_d\|} \right)$.
We will prove the properties of the $\epsilon(\theta)$ function in the following two parts. The remaining properties are easy to check.
\begin{condition}
    \label{condition_1_epsilon}
    The function $\epsilon(\theta)$ is Lipschitz continuous on $\left\{ \theta \mid \forall i, \|\xi_i\| \geq d_{\mathrm{uni}} \right\}$ for any $d_{\mathrm{uni}}>0$.
\end{condition}
\begin{condition}
    \label{condition_2_epsilon}
    If the vectors $\{\xi_1, \dots, \xi_d\}$ in $\theta = (\xi_1, \dots, \xi_d)$ are linearly independent, then for $\epsilon = \epsilon(\theta)$, the following holds:
    \begin{equation*}
        \bigcup_{i=1}^{d} \left[ R_{\xi_i}^{\epsilon} \cup -R_{\xi_i}^{\epsilon} \right] = \mathbb{R}^d \backslash \{0\}.
    \end{equation*}
\end{condition}

\begin{proof}[Proof of Condition \ref{condition_1_epsilon}]
    We first transform the function $\epsilon(\theta)$ into the function related to the singular value of $A(\theta)$:
    \begin{align*}
        \frac{1}{\sqrt{d+1} \|A(\theta)^{-1}\|_2} &= \frac{1}{\sqrt{d+1} \max\{\sigma \mid \sigma \text{ is the singular value of } A(\theta)^{-1}\}} \\
        &= \frac{1}{\sqrt{d+1} \max\{\sigma^{-1} \mid \sigma \text{ is the singular value of } A(\theta) \}} \\
        &= \frac{\min\{\sigma \mid \sigma \text{ is the singular value of } A(\theta) \}}{\sqrt{d+1}} .
    \end{align*}

    Then, denote $\min\{\sigma \mid \sigma \text{ is the singular value of } A \}$ as $\sigma_d(A)$.
    The norm $|\cdot|$ of the matrix $A$ coincides with the norm of the matrix $\theta$, which is the Frobenius norm.
    From Problem III.6.13 in (\cite{bhatiaMatrixAnalysis1997}), we have
    \begin{equation*}
        |\sigma_d(A) - \sigma_d(A^{\prime})| \leq \|A - A^{\prime}\|_2 \leq \|A - A^{\prime}\| ,
    \end{equation*}
    which implies that $\sigma_d(A)$ is Lipschitz continuous.

    On the set $\left\{ \theta \mid \forall i, \|\xi_i\| \geq d_{\mathrm{uni}} \right\}$, The function mapping $\theta$ to $A(\theta)$ is Lipschitz continuous.
    Since $\epsilon(\theta)$ is proportional to $\sigma_d(A(\theta))$, which is a composition of two Lipschitz continuous functions $\sigma_d(A)$ and $A(\theta)$, $\epsilon(\theta)$ is also Lipschitz continuous on this set.
    Therefore, we can finish the proof of Condition \ref{condition_1_epsilon}.
\end{proof}

\begin{proof}[Proof of Condition \ref{condition_2_epsilon}]
    The expression
    \begin{equation*}
        \bigcup_{i=1}^{d} \left[ R_{\xi_i}^{\epsilon} \cup -R_{\xi_i}^{\epsilon} \right] \neq \mathbb{R}^d \backslash \{0\}
    \end{equation*}
    is equivalent to that there exists a vector $\xi \neq 0$ such that, for any index $i \in \{1, \dots, d\}$, $\cossimilarity(\xi,\xi_i) \leq \epsilon$, which is equivalent to that there exists a vector $\xi_i^{\prime}$ with $\|\xi_i^{\prime}\| \leq \epsilon$ such that
    \begin{equation*}
        \cossimilarity(\xi,\frac{\xi_i}{\|\xi_i\|}+\xi_i^{\prime}) = 0
    \end{equation*}
    From this observation, we can derive the equivalence between $\bigcup_{i=1}^{d} \left[ R_{\xi_i}^{\epsilon} \cup -R_{\xi_i}^{\epsilon} \right] \neq \mathbb{R}^d \backslash \{0\}$ and the proposition that there exists $d$ vectors $\{\xi_i^{\prime}\}$ with $\|\xi_i^{\prime}\| \leq \epsilon$ such that
    \begin{equation*}
        \left( \frac{\xi_1}{\|\xi_1\|}+\xi_1^{\prime} , \dots, \frac{\xi_d}{\|\xi_d\|}+\xi_d^{\prime} \right)
    \end{equation*}
    is a singular matrix.

    Moreover, when the vectors $\{\xi_1, \dots, \xi_d\}$ in $\theta = (\xi_1, \dots, \xi_d)$ are linearly independent, the matrix
    \begin{equation*}
        A(\theta) = \left( \frac{\xi_1}{\|\xi_1\|}, \dots, \frac{\xi_d}{\|\xi_d\|} \right)
    \end{equation*}
    is nonsingular.
    Because $\epsilon = \frac{1}{\sqrt{d+1} \|A(\theta)^{-1}\|_2}$, we have $\|\xi_i^{\prime}\| \leq \frac{1}{\sqrt{d+1} \|A(\theta)^{-1}\|_2} < \frac{1}{\sqrt{d} \|A(\theta)^{-1}\|_2}$ and for every index $i$, $\|A(\theta)^{-1} \xi_i^{\prime}\|$ is less than $\frac{1}{\sqrt{d}}$.
    Thus, for every index $i$, $\|A(\theta)^{-1} \xi_i^{\prime}\|_1 \leq \sqrt{d} \|A(\theta)^{-1} \xi_i^{\prime}\| < 1$.
    Therefore, $I + A(\theta)^{-1} (\xi_1^{\prime} , \dots, \xi_d^{\prime})$ is a strictly diagonally dominant matrix and thus its determinant is nonzero.

    In summary, for any $\{\xi_i^{\prime}\}$ with $\|\xi_i^{\prime}\| \leq \epsilon$,
    \begin{align*}
        \left| \left( \frac{\xi_1}{\|\xi_1\|}+\xi_1^{\prime} , \dots, \frac{\xi_d}{\|\xi_d\|}+\xi_d^{\prime} \right) \right|
        &= \left| A(\theta) + (\xi_1^{\prime} , \dots, \xi_d^{\prime}) \right| \\
        &= |A(\theta)| \left| I + A(\theta)^{-1} (\xi_1^{\prime} , \dots, \xi_d^{\prime}) \right|
        \neq 0 .
    \end{align*}
    This statement is equivalent to $\bigcup_{i=1}^{d} \left[ R_{\xi_i}^{\epsilon} \cup -R_{\xi_i}^{\epsilon} \right] = \mathbb{R}^d \backslash \{0\}$.
\end{proof}

\textbf{Proof of the Properties of $\tau_{\xi}^{\epsilon}(\Lambda)$}

The construction of $\tau_{\xi}^{\epsilon}(\Lambda)$ is a little complicated to analyze:
\begin{equation*}
    \tau_{\xi}^{\epsilon}(\Lambda) = \frac{\sqrt{1+\epsilon^2}\frac{\xi^T}{\|\xi\|}\Lambda} {\sqrt{1+\epsilon^2\|\Lambda\|^2}} .
\end{equation*}

For ease of analyzing Lipschitz continuity, we split the domain
\begin{equation*}
    \left\{ (\xi,\epsilon,\Lambda) \mid \|\xi\| \geq d_{\mathrm{uni}}, \epsilon \geq d_{\mathrm{uni}} \right\}
\end{equation*}
into two domains and define two corresponding functions.

First, we set $B(\epsilon, \zeta, r)$ on $\left\{ (\epsilon, \zeta, r) \mid \epsilon \in [d_{\mathrm{uni}},1], \zeta \in [-1,1], r \geq 0 \right\}$ for any $d_{\mathrm{uni}}>0$:
\begin{equation*}
    B(\epsilon, \zeta, r)
    = \frac{\sqrt{1+\epsilon^2} r \zeta}{\sqrt{1+\epsilon^2 r^2}} .
\end{equation*}
This function is continuously differentiable and has bounded derivatives on this domain because its partial derivatives are bounded and continuous on this domain.
The connection between $B$ and $\tau$ is the equation $\tau_{\xi}^{\epsilon}(\Lambda) = B(\epsilon, \cossimilarity(\Lambda,\xi), \|\Lambda\|)$.

Thus, the function $\tau_{\xi}^{\epsilon}(\Lambda) = B(\epsilon, \cossimilarity(\Lambda,\xi), \|\Lambda\|)$ is continuously differentiable and has bounded derivatives on the domain $\left\{ (\xi,\epsilon,\Lambda) \mid \|\xi\| \geq d_{\mathrm{uni}}, \epsilon \geq d_{\mathrm{uni}}, \|\Lambda\| \geq d_{\mathrm{uni}} \right\}$.
This is because its component functions $\cossimilarity(\Lambda,\xi)$ and $\|\Lambda\|$ have bounded and continuous partial derivatives on this domain.

Second, we set $C(\epsilon, x, r)$ on $\left\{ (\epsilon, x, r) \mid \epsilon \in [d_{\mathrm{uni}},1], |x| \leq r, r \geq 0 \right\}$ for any $d_{\mathrm{uni}}>0$:
\begin{equation*}
    C(\epsilon, x, r)
    = \frac{\sqrt{1+\epsilon^2} x}{\sqrt{1+\epsilon^2 r^2}} .
\end{equation*}
This function is continuously differentiable and has bounded derivatives on this domain because its partial derivatives are bounded and continuous on this domain.
The connection between $C$ and $\tau$ is the equation $\tau_{\xi}^{\epsilon}(\Lambda) = C(\epsilon, \frac{\xi^T}{\|\xi\|}\Lambda, \|\Lambda\|)$.

Thus, $\tau_{\xi}^{\epsilon}(\Lambda) = C(\epsilon, \frac{\xi^T}{\|\xi\|}\Lambda, \|\Lambda\|)$ is continuously differentiable and has bounded derivatives on the domain $\left\{ (\xi,\epsilon,\Lambda) \mid \|\xi\| \geq d_{\mathrm{uni}}, \epsilon \geq d_{\mathrm{uni}}, \|\Lambda\| \leq 2d_{\mathrm{uni}} \right\}$.
This is because its component functions $\frac{\xi^T}{\|\xi\|}\Lambda$ and $\|\Lambda\|$ have bounded and continuous partial derivatives on this domain.

In summary, the function $\tau_{\xi}^{\epsilon}(\Lambda)$ is continuously differentiable and has bounded derivatives on the domain
\begin{equation*}
    \left\{ (\xi,\epsilon,\Lambda) \mid \|\xi\| \geq d_{\mathrm{uni}}, \epsilon \geq d_{\mathrm{uni}} \right\}
\end{equation*}
for any $d_{\mathrm{uni}}>0$.

\textbf{Proof of the Properties of $\beta_{\xi}^{\epsilon}(\Lambda)$}

From the proof above and the similar property for the function $\cutsin$, we can show that $\beta_{\xi}^{\epsilon}(\Lambda)$ is continuously differentiable and has bounded derivatives on $\left\{ (\xi,\epsilon,\Lambda) \mid \|\xi\| \geq d_{\mathrm{uni}}, \epsilon \geq d_{\mathrm{uni}} \right\}$ for any $d_{\mathrm{uni}}>0$.

Furthermore, the expression
\begin{equation*}
    \tau_{\xi}^{\epsilon}(\Lambda) = \frac{\sqrt{1+\epsilon^2}\frac{\xi^T}{\|\xi\|}\Lambda} {\sqrt{1+\epsilon^2\|\Lambda\|^2}}
\end{equation*}
implies that for any $d_{\mathrm{uni}} > 0$, if we let $M = \max\{\frac{1}{d_{\mathrm{uni}}^2},1\}$, then for any $\Lambda \in R_\xi^\epsilon$ with $\|\Lambda\| \geq M$, $\epsilon \geq d_{\mathrm{uni}}$ and $\epsilon^\prime \in (0,\epsilon]$, we have
\begin{equation*}
    \tau_{\xi}^{\epsilon^\prime}(\Lambda)
    \geq \frac{\sqrt{1+{\epsilon^\prime}^2} \epsilon\|\Lambda\|} {\sqrt{1+{\epsilon^\prime}^2\|\Lambda\|^2}}
    \geq \frac{\sqrt{1+{\epsilon}^2} \epsilon\|\Lambda\|} {\sqrt{1+{\epsilon}^2\|\Lambda\|^2}}
    \geq \frac{\sqrt{1+{\epsilon}^2}} {\sqrt{\frac{1}{{\epsilon}^2\|\Lambda\|^2} + 1}}
    \geq \frac{\sqrt{1+{\epsilon}^2}} {\sqrt{\frac{1}{{\epsilon}^2M^2} + 1}}
    \geq 1 ,
\end{equation*}
since $\frac{1}{\epsilon M} \leq \epsilon$. Then $\beta_{\xi}^{\epsilon^\prime}(\Lambda) = -1$ holds.

Consequently, we have verified Item \ref{item_assumption_theta_basic_4} of Assumption \ref{assumption_theta_basic}, as well as Item \ref{item_assumption_theta_lipschitz_2} of Assumption \ref{assumption_theta_lipschitz}.
The remaining Items \ref{item_assumption_theta_basic_2} and \ref{item_assumption_theta_basic_3} of Assumption \ref{assumption_theta_basic} are straightforward to check and thus omitted.

\subsection{Proof of Lemma \ref{lemma_average_theta_check}}

\label{subsec_proof_lemma_average_theta_check}

In this subsection, it suffices to analyze the properties of $\xi_i$. For notational simplicity, let $U_j := \alpha_i(X_j)X_j$. The following analysis will be based on the properties of $U_j$.

The almost sure convergence can be easily verified. For the difference between consecutive estimates, we have
\begin{equation*}
    \|\xi_{n+1,i} - \xi_{n,i}\|
    = \left\|\frac{U_{n+1}}{n+1} - \frac{\overline{U}_n}{n+1}\right\|
    \leq \frac{\|U_{n+1}\|}{n+1} + \frac{\|\overline{U}_n\|}{n+1} .
\end{equation*}
From Assumption \ref{assumption_x_sub_exponential_bound} on $X_j$, we have that for any $\epsilon>0$, $\mathbb{E} \|U_j\|^{1/\epsilon}$ is finite. This implies that
\begin{equation*}
    \limsup_{n \rightarrow +\infty} n^{-2\epsilon} \|U_{n+1}\| = 0 \quad a.s. ,
\end{equation*}
by the Borel-Cantelli lemma and
\begin{equation*}
    \sum_n P(\|U_{n+1}\| > a n^{2\epsilon})
    \leq \sum_n \frac{\mathbb{E} \|U_j\|^{1/\epsilon}}{a^{1/\epsilon} n^2}
    < \infty ,
\end{equation*}
for any $a>0$.
Moreover, for $\epsilon < 1$, we have
\begin{equation*}
    \sum_n P(\|\overline{U}_n\| > a n^{2\epsilon})
    \leq \sum_n \frac{\mathbb{E} \|\overline{U}_n\|^{1/\epsilon}}{a^{1/\epsilon} n^2}
    \leq \sum_n \frac{\frac{1}{n} \sum_{j=1}^n \mathbb{E} \|U_j\|^{1/\epsilon}}{a^{1/\epsilon} n^2}
    = \sum_n \frac{\mathbb{E} \|U_j\|^{1/\epsilon}}{a^{1/\epsilon} n^2}
    < \infty ,
\end{equation*}
where the second inequality follows from Jensen's inequality, noting that the mapping $x \mapsto \|x\|^{1/\epsilon}$ is convex on $\mathbb{R}^d$ with respect to the $\ell_2$ norm $\|\cdot\|$.
Therefore,
\begin{equation*}
    \limsup_{n \rightarrow +\infty} n^{-2\epsilon} \|\overline{U}_n\| = 0 \quad a.s.
\end{equation*}
Consequently, we have that for any $\epsilon>0$,
\begin{equation*}
    \sup_n n^{1-\epsilon} \|\xi_{n+1,i} - \xi_{n,i}\| < \infty ,
\end{equation*}
almost surely. This further implies that $\sup_n n^\epsilon \|\theta_n-\theta_{n+1}\| < \infty$ almost surely for any $\epsilon \in (0,1)$, which implies Assumption \ref{assumption_theta_convergence}.

\subsection{Proof of Corollary \ref{corollary_boundedness_convergent}}

\label{subsec_proof_corollary_boundedness_convergent}

The proof of Corollary \ref{corollary_boundedness_convergent} is based on Theorem \ref{theorem_boundedness}. 
Lemma \ref{lemma_instance_established} leads to Assumption \ref{assumption_theta_basic}.
Combining with Assumption \ref{assumption_x_sub_exponential_bound}, Lemma \ref{lemma_drift_condition} implies $\overline{B(\theta^*,r_d)} \subset K_{M,\Delta}$ for some positive numbers $M$, $\Delta$ and $r_d$.
Lemma \ref{lemma_average_theta_check} leads to Assumption \ref{assumption_theta_convergence}.
Thus, the conditions of Theorem \ref{theorem_boundedness} are satisfied, and the result follows.

\subsection{Proof of Corollary \ref{corollary_asymptotic_normality}}

\label{subsec_proof_corollary_asymptotic_normality}

Under the assumptions and settings in the corollary, Lemmas \ref{lemma_g_lipschitz} and \ref{lemma_instance_established} lead to Assumption \ref{assumption_g_lipschitz}.
Lemma \ref{lemma_average_theta_check} leads to Assumption \ref{assumption_theta_convergence}.
Theorem \ref{theorem_shift_simultaneous_geometric_finite_adjusted} leads to a condition that there exists $r_d>0$ such that $\overline{B(\theta^*,r_d)} \subset K_{M,\Delta}$ for some $M>0$ and $\Delta>0$.
Therefore, the assumptions of Theorem \ref{theorem_CLT_complex_center} are satisfied.

By Theorem \ref{theorem_shift_simultaneous_geometric_finite_adjusted}, the centering term, $\pi_\theta h_\theta$, in the asymptotic distribution of $\sum_{i=1}^n (T_i-\rho)Y_i$ is zero.

\section{Central Limit Theorem}
\label{sec_universal_CLT}

\subsection{Proof of Theorem \ref{theorem_CLT_complex_center}}
\label{subsec_proof_theorem_CLT_complex_center}

The assumption below follows from Theorems \ref{theorem_drift_condition_inequality} and \ref{theorem_simultaneous_geometric_ergodicity}.
It is key properties in our proof.
\begin{assumption}
    \label{assumption_simultaneous_geometric_ergodicity}
    There exists a constant $r_d>0$, for any $\theta \in \overline{B(\theta^*,r_d)}$, $P_\theta$ is positive recurrent with invariant probability $\pi_\theta$. Moreover, there exists a Lyapunov function $V(\Lambda) = \exp(\lambda_1 \|\Lambda\|) \geq 1$ with $\lambda_1 > 0$ such that for any $\alpha \in (0,1]$, there exist some $L_\alpha>1$, $\beta_\alpha \in (0,1)$, $b_\alpha>0$, the inequalities $\pi_\theta V^\alpha \leq \frac{b_\alpha}{1-\beta_\alpha}$,
    \begin{equation*}
        \| P_\theta^{n}(\Lambda, \cdot) - \pi_\theta \|_{V^\alpha} \leq L_\alpha(1-L_\alpha^{-1})^n V^\alpha(\Lambda)
    \end{equation*}
    and
    \begin{equation*}
        P_\theta V^\alpha \leq \beta_\alpha V^\alpha + b_\alpha
    \end{equation*}
    hold for any $\theta \in \overline{B(\theta^*,r_d)}$.
\end{assumption}

Let $S_\Theta=\overline{B(\theta^*,\min\{r_d,r_g\})}$.

Before the main part of the proof, we first provide some basic notation and present the decomposition of the left-hand side of (\ref{eq_CLT}).
In addition, the proof relies on prerequisite properties established in Lemmas \ref{lemma_lipschitz_continuity_P_pi}--\ref{lemma_lipschitz_boundedness_FGH}.
These include the Lipschitz continuity and uniform bounds for the transition kernel $P_\theta$ and the functions $F_\theta$, $G_\theta$, and $H_\theta$.

Recall that
\begin{equation*}
    h_\theta(\Lambda) = \mathbb{E} \left[ [g_\theta(\Lambda, X) - \rho] f(X) \right],
\end{equation*}
where $f(x) = E[Y \mid X=x]$.

Given the function $h_\theta$, the transition kernel $P_\theta$ and the invariant probability $\pi_\theta$, consider the following Poisson equation, which is common in the theory of Markov chain (\cite{meynMarkovChainsStochastic2009}):
\begin{equation*}
    \hat{h}_\theta - P_\theta \hat{h}_\theta = h_\theta - \pi_\theta h_\theta ,
\end{equation*}
where $\hat{h}_\theta$ denotes the solution to the Poisson equation.
If $P_\theta$ is geometrically ergodic and $h_\theta$ is bounded, the Poisson equation admits the following solution:
\begin{equation*}
    \hat{h}_\theta = \sum_{n=0}^\infty (P_\theta^n - \pi_\theta) (h_\theta) .
\end{equation*}

To establish the asymptotic normality result stated in (\ref{eq_CLT}), we utilize the Poisson equation and decompose the left-hand side of (\ref{eq_CLT}) into the following seven components:
\begin{enumerate}
    \item \label{item_proof_CLT_1} $\frac{1}{\sqrt{N}} \sum_{n=0}^{N-1} { \left[ (T_{n+1}-\rho) Y_{n+1} - h_{\theta_n}(\Lambda_n) + Z_{n+1}-\mathbb{E}Z_{n+1} \right] \mathbb{I}(\theta_n \in S_\Theta) }$.
    \item \label{item_proof_CLT_2} $\frac{1}{\sqrt{N}} \sum_{n=0}^{N-1} { \left[ \hat{h}_{\theta_n}(\Lambda_{n+1}) - (P_{\theta_n}\hat{h}_{\theta_n})(\Lambda_n) \right] \mathbb{I}(\theta_n \in S_\Theta) }$.
    \item \label{item_proof_CLT_3} $\frac{1}{\sqrt{N}} \sum_{n=0}^{N-1} { \left[ \hat{h}_{\theta_{n+1}}(\Lambda_{n+1}) - \hat{h}_{\theta_n}(\Lambda_{n+1}) \right] \mathbb{I}(\theta_n \in S_\Theta, \theta_{n+1} \in S_\Theta) }$.
    \item \label{item_proof_CLT_4} $\frac{1}{\sqrt{N}} \sum_{n=0}^{N-1} -\hat{h}_{\theta_n}(\Lambda_{n+1})
    \mathbb{I}(\theta_n \in S_\Theta, \theta_{n+1} \notin S_\Theta)$.
    \item \label{item_proof_CLT_5} $\frac{1}{\sqrt{N}} \sum_{n=1}^N \hat{h}_{\theta_n}(\Lambda_n)
    \mathbb{I}(\theta_n \in S_\Theta, \theta_{n-1} \notin S_\Theta)$.
    \item \label{item_proof_CLT_6} $\frac{1}{\sqrt{N}} \left[ \hat{h}_{\theta_0}(\Lambda_0)\mathbb{I}(\theta_0 \in S_\Theta) - \hat{h}_{\theta_N}(\Lambda_N)\mathbb{I}(\theta_N \in S_\Theta) \right]$.
    \item \label{item_proof_CLT_7} $\frac{1}{\sqrt{N}} \sum_{n=0}^{N-1} { \left[ (T_{n+1}-\rho) Y_{n+1} - \pi_{\theta^*}(h_{\theta^*}) + Z_{n+1}-\mathbb{E}Z_{n+1} \right] \mathbb{I}(\theta_n \notin S_\Theta) }$.
\end{enumerate}

We will prove that the sum of Items \ref{item_proof_CLT_1} and \ref{item_proof_CLT_2} is asymptotically normal, while the other five terms are $o_P(1)$.

\textbf{Item \ref{item_proof_CLT_1} and Item \ref{item_proof_CLT_2}}

The sum of Item \ref{item_proof_CLT_1} and Item \ref{item_proof_CLT_2} forms a martingale sequence. The corresponding martingale difference sequence, denoted by $\{\Delta M_n\}$, is
\begin{align*}
    \Delta M_n &= \left[ (T_{n+1}-\rho) Y_{n+1} - h_{\theta_n}(\Lambda_n) + Z_{n+1}-\mathbb{E}Z_{n+1}
    + \hat{h}_{\theta_n}(\Lambda_{n+1}) - (P_{\theta_n}\hat{h}_{\theta_n})(\Lambda_n) \right] \\
    & \quad \mathbb{I}(\theta_n \in S_\Theta)
\end{align*}

Moreover, the conditional variance is
\begin{align*}
    & \quad \mathbb{E} \left[ \left\{ \Delta M_n \right\}^2
    \mid \mathcal{F}_n \right] \\
&= \mathbb{I}(\theta_n \in S_\Theta) \left\{
        \mathbb{E} \left[ \left[ (T_{n+1}-\rho) Y_{n+1} - h_{\theta_n}(\Lambda_n) + Z_{n+1}-\mathbb{E}Z_{n+1} \right]^2 \mid \mathcal{F}_n \right]
    \right. \\
    & \quad + \mathbb{E} \left[ P_{\theta_n}(\hat{h}_{\theta_n}^2)(\Lambda_n) - (P_{\theta_n}\hat{h}_{\theta_n})^2(\Lambda_n) \mid \mathcal{F}_n \right] \\
    & \quad + 2 \mathbb{E} \left[ \left[ (T_{n+1}-\rho) Y_{n+1} - h_{\theta_n}(\Lambda_n) + Z_{n+1}-\mathbb{E}Z_{n+1} \right] \right. \\
    & \quad\quad \left. \left. \left[ \hat{h}_{\theta_n}(\Lambda_{n+1}) - (P_{\theta_n}\hat{h}_{\theta_n})(\Lambda_n) \right] \mid \mathcal{F}_n \right] \right\} \\
    &= G_{\theta_n}(\Lambda_n)\mathbb{I}(\theta_n \in S_\Theta) + F_{\theta_n}(\Lambda_n)\mathbb{I}(\theta_n \in S_\Theta) + 2 H_{\theta_n}(\Lambda_n)\mathbb{I}(\theta_n \in S_\Theta) ,
\end{align*}
where the functions $G_\theta$, $F_\theta$, and $H_\theta$ correspond to the first, second, and third terms in the decomposition of the conditional variance, respectively.

From the subsequent proof, we can show that
\begin{equation}
    \frac{1}{N} \sum_{n=0}^{N-1} \left[ G_{\theta_n}(\Lambda_n) + F_{\theta_n}(\Lambda_n) + 2H_{\theta_n}(\Lambda_n) \right]\mathbb{I}(\theta_n \in S_\Theta)
    \xrightarrow{\mathbb{P}} \pi_{\theta^*}[G_{\theta^*}+F_{\theta^*}+2H_{\theta^*}]
    \label{eq_CLT_average_convergence}
\end{equation}

Next, we need to check the Lindeberg condition. For any $\epsilon>0$,
\begin{align*}
    & \quad \frac{1}{N} \sum_{n=0}^{N-1} \mathbb{E} \left[ \Delta M_n^2 \mathbb{I}(|\Delta M_n| \geq \epsilon \sqrt{N}) \right] \\
    & \leq \frac{25}{N} \sum_{n=0}^{N-1} \mathbb{E} \left[ ((T_{n+1}-\rho)Y_{n+1})^2 \mathbb{I}(|(T_{n+1}-\rho)Y_{n+1}| \geq \epsilon \sqrt{N}/5) \right] \\
    & \quad + \frac{25}{N} \sum_{n=0}^{N-1} \mathbb{E} \left[ (Z_{n+1}-\mathbb{E}Z_{n+1})^2 \mathbb{I}(|Z_{n+1}-\mathbb{E}Z_{n+1}| \geq \epsilon \sqrt{N}/5) \right] \\
    & \quad + \frac{25}{N} \sum_{n=0}^{N-1} \mathbb{E} \left[ (h_{\theta_n}(\Lambda_n))^2 \mathbb{I}(|h_{\theta_n}(\Lambda_n)| \geq \epsilon \sqrt{N}/5) \right] \\
    & \quad + \frac{25}{N} \sum_{n=0}^{N-1} \mathbb{E} \left[ (\hat{h}_{\theta_n}(\Lambda_{n+1}))^2 \mathbb{I}(\theta_n \in S_\Theta) \mathbb{I}(|\hat{h}_{\theta_n}(\Lambda_{n+1})| \geq \epsilon \sqrt{N}/5) \right] \\
    & \quad + \frac{25}{N} \sum_{n=0}^{N-1} \mathbb{E} \left[ ((P_{\theta_n}\hat{h}_{\theta_n})(\Lambda_n))^2 \mathbb{I}(\theta_n \in S_\Theta) \mathbb{I}(|(P_{\theta_n}\hat{h}_{\theta_n})(\Lambda_n)| \geq \epsilon \sqrt{N}/5) \right] \\
    & \leq \frac{25}{N} \sum_{n=0}^{N-1} \mathbb{E} \left[ Y_{n+1}^2 \mathbb{I}(|Y_{n+1}| \geq \epsilon \sqrt{N}/5) \right]
    + \frac{25}{N} \sum_{n=0}^{N-1} \mathbb{E} \left[ \mathbb{E}|Y| \mathbb{I}((\mathbb{E}|Y|)^2 \geq \epsilon \sqrt{N}/5) \right] \\
    & \quad + \frac{25}{N} \sum_{n=0}^{N-1} \mathbb{E} \left[ (Z_{n+1}-\mathbb{E}Z_{n+1})^2 \mathbb{I}(|Z_{n+1}-\mathbb{E}Z_{n+1}| \geq \epsilon \sqrt{N}/5) \right] \\
    & \quad + \frac{25}{N} \sum_{n=0}^{N-1} \mathbb{E} \left[ (C_{h,\alpha} V^{\alpha}(\Lambda_{n+1}))^2 \mathbb{I}(|C_{h,\alpha} V^{\alpha}(\Lambda_{n+1})| \geq \epsilon \sqrt{N}/5) \right] \\
    & \quad + \frac{25}{N} \sum_{n=0}^{N-1} \mathbb{E} \left[ (C_{Ph,\alpha} V^{\alpha}(\Lambda_n))^2 \mathbb{I}(|C_{Ph,\alpha} V^{\alpha}(\Lambda_n)| \geq \epsilon \sqrt{N}/5) \right] .
\end{align*}

Since the sequence $\{(Y_{n+1},Z_{n+1})\}$ is i.i.d. with finite second moments, the first three terms on the right-hand side converge to zero.
The fourth and fifth terms can be derived from the inequality that
\begin{align*}
    & \quad \frac{1}{N} \sum_{n=0}^{N-1} \mathbb{E} \left[ V^{2\alpha}(\Lambda_n) \mathbb{I}(|V^{\alpha}(\Lambda_n)| \geq \epsilon \sqrt{N}/5) \right]
    \leq \frac{1}{N} \sum_{n=0}^{N-1} \mathbb{E} \left[ \frac{5}{\epsilon \sqrt{N}} V^{3\alpha}(\Lambda_n) \right] \\
    & \leq \frac{1}{N} \sum_{n=0}^{N-1} \frac{5}{\epsilon \sqrt{N}} \max\{\frac{b_{3\alpha}}{1-\beta_{3\alpha}},\mathbb{E} V(\Lambda_0)\}
    \leq \frac{5}{\epsilon \sqrt{N}} \max\{\frac{b_{3\alpha}}{1-\beta_{3\alpha}},\mathbb{E} V(\Lambda_0)\}
    \rightarrow 0 ,
\end{align*}
as $N \rightarrow \infty$, where $\alpha \in (0,1/3]$.

Thus, we have successively proved the average convergence of variances and the Lindeberg condition. Therefore, we can conclude from Corollary 3.1 in (\cite{hallMartingaleLimitTheory1980}) that
\begin{equation*}
    \frac{1}{\sqrt{N}} \sum_{n=0}^{N-1} \Delta M_n \xrightarrow{d} \mathcal{N}(0, {\sigma^*_{Y,Z}}^2) .
\end{equation*}

\textbf{Proof of Average Convergence}

Before the proof of \eqref{eq_CLT_average_convergence}, we show that the transition kernels and several functions satisfy boundedness in the sense of $|\cdot|_V$ and Lipschitz continuity, as stated in the following lemmas.
\begin{lemma}
    \label{lemma_lipschitz_continuity_P_pi}

    Under Assumptions \ref{assumption_g_lipschitz} and \ref{assumption_simultaneous_geometric_ergodicity}, there exists a constant $L_P>0$ such that for any $\theta$, $\theta^\prime \in \overline{B(\theta^*, r_g)}$,
    \begin{equation*}
        D_V\left(\theta, \theta^{\prime}\right) \leq L_P \|\theta - \theta^{\prime}\| .
    \end{equation*}
    Moreover, the constant $L_P$ depends only on $L_g$, $\beta$, $b$ and $\iota$.

    Furthermore, the family of invariant probability measures $\{\pi_\theta\}$ satisfies, for any $\theta$, $\theta^\prime \in \overline{B(\theta^*, \min\{r_d,r_g\})}$,
    \begin{equation*}
        \|\pi_\theta-\pi_{\theta^{\prime}}\|_V \leq L_{d\pi} \|\theta - \theta^{\prime}\| ,
    \end{equation*}
    where the constant $L_{d\pi}$ depends only on $L$, $L_P$ and $C_V = \frac{b}{1-\beta}$.

    If $V$ is replaced by $V^\alpha$ for some $\alpha \in (0,1]$, the corresponding Lipschitz constants are denoted by $L_{P,\alpha}$ and $L_{d\pi,\alpha}$, respectively.
\end{lemma}

\begin{lemma}
    \label{lemma_lipschitz_continuity_poisson}

    Let $\hat{h}_\theta$ be the solution to the Poisson equation $\hat{h}_\theta - P_\theta \hat{h}_\theta = h_\theta - \pi_\theta h_\theta$.
    Under the same assumptions as in Lemma \ref{lemma_lipschitz_continuity_P_pi}, the function $\hat{h}_\theta$ is Lipschitz continuous on the set $\overline{B(\theta^*,\min\{r_d,r_g\})}$ in the sense that for any $\theta$, $\theta^\prime \in \overline{B(\theta^*,\min\{r_d,r_g\})}$,
    \begin{equation*}
        |\hat{h}_\theta(\Lambda) - \hat{h}_{\theta^{\prime}}(\Lambda)|
        \leq L_{dh} \|\theta - \theta^{\prime}\| V(\Lambda) ,
    \end{equation*}
    where the constant $L_{dh}$ depends only on $\mathbb{E}|Y|$, $L$, $L_P$, $L_{d\pi}$ and $L_h$.
    If $V$ is replaced by $V^\alpha$ for some $\alpha \in (0,1]$, the corresponding Lipschitz constant is denoted by $L_{dh,\alpha}$.

    Moreover, for any $\alpha \in (0,1]$, the function $\hat{h}_\theta$ satisfies the bound,
    \begin{equation*}
        |\hat{h}_\theta(\Lambda)|
        \leq C_{h,\alpha} V^{\alpha}(\Lambda) ,
    \end{equation*}
    where the constant $C_{h,\alpha} = \mathbb{E}|Y| L_\alpha^2$.
\end{lemma}

\begin{lemma}
    \label{lemma_lipschitz_boundedness_FGH}

    Under Assumptions \ref{assumption_g_lipschitz} and \ref{assumption_simultaneous_geometric_ergodicity}, for any $\alpha \in (0,1]$, there exist positive constants $L_{dF,\alpha}$ and $C_{F,\alpha}$ such that for any $\theta$, $\theta^\prime \in S_\Theta$,
    \begin{align*}
        |F_\theta(\Lambda) - F_{\theta^{\prime}}(\Lambda)|
        & \leq L_{dF,\alpha} V^{\alpha}(\Lambda) \|\theta - \theta^{\prime}\| , \\
        |F_\theta(\Lambda)|
        & \leq C_{F,\alpha} V^{\alpha}(\Lambda) .
    \end{align*}

    Analogous bounds also hold for $G_\theta$ and $H_\theta$ with corresponding constants.
\end{lemma}

The verification that the average of $F$ converges can be easily adapted to $G$ and $H$, as they share similar bounds and Lipschitz continuity properties.
Thus, we only need to focus on the function $F$ here.

We check the condition in Lemma \ref{lemma_convergence_average_enter} for any $a \in [1/2,1)$. The condition in the beginning of the lemma can be derived from Assumption \ref{assumption_simultaneous_geometric_ergodicity}.
The other conditions hold as follows:
\begin{enumerate}[label=\roman*.]
    \item The condition holds due to Assumption \ref{assumption_theta_convergence}.
    \item The condition holds due to Assumption \ref{assumption_simultaneous_geometric_ergodicity}.
    \item The condition holds due to the boundedness of $F_\theta$ in the sense of $|\cdot|_{V^a}$.
    Specifically, Lemma \ref{lemma_lipschitz_boundedness_FGH} establishes the bound $|F_\theta|_{V^a} \leq C_{F,\alpha}$.
    \item The condition holds due to Lemma \ref{lemma_convergence_V_as}. The requirement of Lemma \ref{lemma_convergence_V_as} can be satisfied with $S_\Theta = B(\theta^*,\min\{r_d,r_g\})$, and the proof of finiteness is similar to that of Corollary \ref{corollary_boundedness_convergent}.
\item The condition holds because
    \begin{align*}
        & \quad \frac{1}{N} \sum_{n=1}^{N-1} D_V\left(\theta_{n+1}, \theta_n\right) V(\Lambda_{n+1}) \mathbb{I}(\theta_n \in S_\Theta, \theta_{n+1} \in S_\Theta) \\
        & \leq \frac{L_P}{N} \sum_{n=1}^{N-1} \|\theta_{n+1}-\theta_n\| V(\Lambda_{n+1})
        \xrightarrow{\mathbb{P}} 0
    \end{align*}
    by Lemma \ref{lemma_lipschitz_continuity_P_pi}, Assumption \ref{assumption_theta_convergence} and Lemma \ref{lemma_convergence_V_as}.
    \item The condition holds because
    \begin{align*}
        & \quad \frac{1}{N} \sum_{n=1}^{N-1} \left|F_{\theta_{n+1}}-F_{\theta_n}\right|_V V(\Lambda_{n+1}) \mathbb{I}(\theta_n \in S_\Theta, \theta_{n+1} \in S_\Theta) \\
        & \leq \frac{L_{dF}}{N} \sum_{n=1}^{N-1} \|\theta_{n+1}-\theta_n\| V(\Lambda_{n+1}) \xrightarrow{\mathbb{P}} 0
    \end{align*}
    by Assumption \ref{assumption_theta_convergence}, Lemmas \ref{lemma_lipschitz_boundedness_FGH} and \ref{lemma_convergence_V_as}.
\end{enumerate}
The inequalities in the last two items follow from Lemma \ref{lemma_lipschitz_continuity_P_pi}.
The limit is because
\begin{equation}
    \sum_{n=1}^{N-1} \frac{1}{n}\|\theta_{n+1}-\theta_n\| V(\Lambda_{n+1}) \label{eq_convergence_parameter_diff_V_n_sum}
\end{equation}
can be bounded by a multiple of
\begin{equation*}
    \sum_{n=1}^{N-1} \frac{1}{n^{1+\epsilon}} V(\Lambda_{n+1})
\end{equation*}
due to Assumption \ref{assumption_theta_convergence}.
The latter expression converges almost surely as $N \rightarrow \infty$, which is a result of Lemma \ref{lemma_convergence_V_as}.
From the Kronecker Lemma, (\ref{eq_convergence_parameter_diff_V_n_sum}) implies that
\begin{equation*}
    \frac{1}{N} \sum_{n=1}^{N-1} \|\theta_{n+1}-\theta_n\| V(\Lambda_{n+1}) \rightarrow 0
\end{equation*}
almost surely.

The result of Lemma \ref{lemma_convergence_average_enter} is
\begin{equation*}
    \frac{1}{N} \sum_{n=0}^{N-1} F_{\theta_n}\left(\Lambda_n\right)\mathbb{I}(\theta_n \in S_\Theta)
    -\frac{1}{N} \sum_{n=0}^{N-1} \pi_{\theta_n}F_{\theta_n}\mathbb{I}(\theta_n \in S_\Theta)
    \xrightarrow{\mathbb{P}} 0 .
\end{equation*}

For the expectation term:
\begin{align*}
    & \quad \left| \frac{1}{N} \sum_{n=0}^{N-1} \pi_{\theta_n}F_{\theta_n}\mathbb{I}(\theta_n \in S_\Theta)
    - \frac{1}{N} \sum_{n=0}^{N-1} \pi_{\theta^*}F_{\theta^*}\mathbb{I}(\theta_n \in S_\Theta) \right| \\
    & \leq \frac{1}{N} \sum_{n=0}^{N-1} \left[ \pi_{\theta_n}|F_{\theta_n} - F_{\theta^*}| + |\pi_{\theta_n} - \pi_{\theta^*}|F_{\theta^*} \right]\mathbb{I}(\theta_n \in S_\Theta) \\
    & \leq \frac{1}{N} \sum_{n=0}^{N-1} \left[ \pi_{\theta_n} L_{dF,\alpha} V^{\alpha} \|\theta_n - \theta^*\| + |\pi_{\theta_n} - \pi_{\theta^*}|(C_{F,\alpha} V^{\alpha}) \right]\mathbb{I}(\theta_n \in S_\Theta) \\
    & \leq \frac{1}{N} \sum_{n=0}^{N-1} \left[ L_{dF,\alpha} C_{V,\alpha} \|\theta_n - \theta^*\| + C_{F,\alpha} L_{d\pi,\alpha} \|\theta_n - \theta^*\| \right]\mathbb{I}(\theta_n \in S_\Theta) \\
    & \leq (L_{dF,\alpha} C_{V,\alpha} + C_{F,\alpha} L_{d\pi,\alpha}) \frac{1}{N} \sum_{n=0}^{N-1} \|\theta_n - \theta^*\|
\end{align*}
for $\alpha \in (0,1]$.
The last term converges to $0$ by Assumption \ref{assumption_theta_convergence}.

Since $\frac{1}{N} \sum_{n=0}^{N-1} \pi_{\theta^*}F_{\theta^*}\mathbb{I}(\theta_n \in S_\Theta)$ converges to $\pi_{\theta^*}F_{\theta^*}$ almost surely, combining these results shows that $\frac{1}{N} \sum_{n=0}^{N-1} F_{\theta_n}\left(\Lambda_n\right)\mathbb{I}(\theta_n \in S_\Theta)$ converges to $\pi_{\theta^*}F_{\theta^*}$ in probability.

Similarly, $\frac{1}{N} \sum_{n=0}^{N-1} G_{\theta_n}\left(\Lambda_n\right)\mathbb{I}(\theta_n \in S_\Theta)$ and $\frac{1}{N} \sum_{n=0}^{N-1} H_{\theta_n}\left(\Lambda_n\right)\mathbb{I}(\theta_n \in S_\Theta)$ converge to $\pi_{\theta^*}G_{\theta^*}$ and $\pi_{\theta^*}H_{\theta^*}$ in probability, respectively.

\textbf{Other Items}

Item \ref{item_proof_CLT_3} in the decomposition (\ref{eq_CLT}) is bounded by
\begin{equation*}
    L_{dh} \frac{1}{\sqrt{N}} \sum_{n=0}^{N-1} {\|\theta_{n+1} - \theta_n\|}V(\Lambda_{n+1}) .
\end{equation*}
From Assumption \ref{assumption_theta_convergence}, we have the limit that
\begin{equation*}
    L_{dh} \frac{1}{\sqrt{N}} \sum_{n=0}^{N-1} {\|\theta_{n+1} - \theta_n\|}V(\Lambda_{n+1}) \rightarrow 0 ,
\end{equation*}
almost surely. This limit can be derived from
\begin{equation*}
    \frac{1}{\sqrt{N}} \sum_{n=0}^{N-1} \frac{1}{n^\epsilon} V(\Lambda_{n+1})
    \leq \sum_{n=0}^{N-1} \frac{1}{n^{1/2+\epsilon}} V(\Lambda_{n+1})
    \rightarrow 0 ,
\end{equation*}
almost surely, where the limit is derived from Assumption \ref{assumption_theta_convergence} and Lemma \ref{lemma_convergence_V_as}.

For Item \ref{item_proof_CLT_6}, we first observe that
\begin{align*}
    & \quad \frac{1}{\sqrt{N}} \left| \hat{h}_{\theta_0}(\Lambda_0)\mathbb{I}(\theta_0 \in S_\Theta) - \hat{h}_{\theta_N}(\Lambda_N)\mathbb{I}(\theta_N \in S_\Theta) \right| \\
    & \leq \frac{\left| \hat{h}_{\theta_0}(\Lambda_0) \right|}{\sqrt{N}}
    + \frac{\left| \hat{h}_{\theta_N}(\Lambda_N) \right|}{\sqrt{N}}
    \leq \frac{\left| \hat{h}_{\theta_0}(\Lambda_0) \right|}{\sqrt{N}}
    + \frac{C_h V(\Lambda_N)}{\sqrt{N}} .
\end{align*}
Moreover, Theorem \ref{theorem_boundedness} implies that $\Lambda_N = O_P(1)$, and thus $V(\Lambda_N) = O_P(1)$.
Combining these two facts, each term on the right-hand side converges to zero in probability.
Consequently, Item \ref{item_proof_CLT_6} also converges to zero in probability.

For the remaining Items \ref{item_proof_CLT_4}, \ref{item_proof_CLT_5}, and \ref{item_proof_CLT_7}, the number of non-zero terms in each summation is finite almost surely.
This is because Assumption \ref{assumption_theta_convergence} implies that the event $\theta_n \notin S_\Theta$ occurs only a finite number of times.
It follows that a finite sum of bounded terms, when divided by $\sqrt{N}$, converges to zero almost surely.

\subsection{Proof of Lemma \ref{lemma_lipschitz_continuity_P_pi}}

Viewing $P_\theta$ as a mapping on the function space, we can show that
\begin{equation*}
    P_\theta(\Lambda, h) = \int \left[ g_\theta(\Lambda, X) h(\Lambda + (1-\rho) X) \right] \Gamma(dX) + \int \left[ (1-g_\theta(\Lambda, X)) h(\Lambda - \rho X) \right] \Gamma(dX).
\end{equation*}
Then $\| P_\theta(\Lambda, \cdot) - P_{\theta^{\prime}}(\Lambda, \cdot) \|_V$ is equal to the supremum of
\begin{equation*}
    V(\Lambda)^{-1} \left| P_\theta(\Lambda, h) - P_{\theta^{\prime}}(\Lambda, h) \right|,
\end{equation*}
where the supremum is taken over all measurable functions $h$ with $|h|_V \le 1$.
In addition, $\left| P_\theta(\Lambda, h) - P_{\theta^{\prime}}(\Lambda, h) \right|$ is bounded by
\begin{equation*}
    \int \left[ |g_\theta(\Lambda, X) - g_{\theta^{\prime}}(\Lambda, X)| (|h|(\Lambda + (1-\rho) X) + |h|(\Lambda - \rho X)) \right] \Gamma(dX) .
\end{equation*}
Therefore,
\begin{align*}
    \left| P_\theta(\Lambda, h) - P_{\theta^{\prime}}(\Lambda, h) \right|
    & \leq L_g \|\theta - \theta^{\prime}\| \int \left[ (|h|(\Lambda + (1-\rho) X) + |h|(\Lambda - \rho X)) \right] \Gamma(dX) \\
    & \leq L_g \|\theta - \theta^{\prime}\| \int \left[ (V(\Lambda + (1-\rho) X) + V(\Lambda - \rho X)) \right] \Gamma(dX) \\
    & \leq L_g \|\theta - \theta^{\prime}\| \frac{1}{\iota} (P_\theta V)(\Lambda) \\
    & \leq L_g \|\theta - \theta^{\prime}\| \frac{\beta+b}{\iota} V(\Lambda)
\end{align*}

Thus, by setting $L_P = L_g \frac{\beta+b}{\iota}$, we have $\| P_\theta - P_{\theta^{\prime}} \|_V \leq L_P \|\theta - \theta^{\prime}\|$, which in turn implies that $D_V\left(\theta, \theta^{\prime}\right) \leq L_P \|\theta - \theta^{\prime}\|$.
If $V$ is replaced by $V^\alpha$ for some $\alpha \in (0,1]$, the corresponding Lipschitz constant is denoted by $L_{P,\alpha} = L_g \frac{\beta_\alpha+b_\alpha}{\iota}$.

By Assumptions \ref{assumption_g_lipschitz} and \ref{assumption_simultaneous_geometric_ergodicity}, the requirements of Lemma \ref{lemma_solution_continuity} are satisfied on the set
\begin{equation*}
    S_\Theta = \overline{B(\theta^*, \min\{r_d,r_g\})} .
\end{equation*}
Because $D_V$ can be bounded by a constant multiple of $\|\theta - \theta^{\prime}\|$, $\pi_\theta$ is Lipschitz continuous with respect to $\theta \in \overline{B(\theta^*, \min\{r_d,r_g\})}$ in the sense of $\|\cdot\|_V$ norm.
Specifically,
\begin{equation*}
    \left\|\pi_\theta-\pi_{\theta^{\prime}}\right\|_V
    \leq L^2\left\{\pi_\theta(V)+L V(x)\right\} D_V\left(\theta, \theta^{\prime}\right)
    \leq 2 L^3 L_P C_V \|\theta - \theta^{\prime}\| .
\end{equation*}
We denote the Lipschitz constant by $L_{d\pi} = 2 L^3 L_P C_V$ from Lemma \ref{lemma_solution_continuity}.
If $V$ is replaced by $V^\alpha$ for some $\alpha \in (0,1]$, the corresponding Lipschitz constant is denoted by $L_{d\pi,\alpha} = 2 L_\alpha^3 L_{P,\alpha} C_V$.

\subsection{Proof of Lemma \ref{lemma_lipschitz_continuity_poisson}}

The function
\begin{equation*}
    h_\theta(\Lambda) = \mathbb{E}_{X \sim \Gamma} \left[ [g_\theta(\Lambda, X) - \rho] f(X) \right],
\end{equation*}
where $f(x) = \mathbb{E} [Y \mid X=x]$.
Since $g_\theta(\Lambda, X)$ is Lipschitz continuous with respect to $\theta$ with a Lipschitz constant that is uniform in $\Lambda$ and $X$, $h_\theta(\Lambda)$ is also Lipschitz continuous.
The Lipschitz constant for $h_\theta$, denoted by $L_h$, can be derived as follows.
For any $\theta$, $\theta^\prime \in \overline{B(\theta^*, r_g)}$,
\begin{align*}
    |h_\theta(\Lambda) - h_{\theta'}(\Lambda)|
    &\leq \mathbb{E}_{X \sim \Gamma} [|g_\theta(\Lambda, X) - g_{\theta'}(\Lambda, X)| |f(X)|] \\
    & \leq L_g \|\theta-\theta'\| \mathbb{E}|f(X)|
    \leq L_g \|\theta-\theta'\| \mathbb{E}|Y| .
\end{align*}
Thus, $L_h = L_g \mathbb{E}|Y|$.

Under Assumption \ref{assumption_simultaneous_geometric_ergodicity}, for any $\theta \in \overline{B(\theta^*,r_d)}$, the solution to the Poisson equation, $\hat{h}_\theta$, equals $\sum_{n=0}^\infty (P_\theta^n - \pi_\theta) (h_\theta)$.
With $C_\alpha = L_\alpha$ and $\rho_\alpha = 1-L_\alpha^{-1}$, the inequality
\begin{equation*}
    \| P^n_\theta(\Lambda, \cdot) - \pi_\theta \|_{V^{\alpha}}
    \leq C_{\alpha} \rho_{\alpha}^n V^{\alpha}(\Lambda)
\end{equation*}
implies that for any $\alpha \in (0,1]$, there exists a positive number $C_{h,\alpha}$ such that
\begin{align*}
    |\hat{h}_\theta(\Lambda)|
    &= |\sum_{n=0}^{\infty} (P^n_\theta(\Lambda, \cdot) - \pi_\theta )(h_\theta)|
    \leq \sum_{n=0}^{\infty} |(P^n_\theta(\Lambda, \cdot) - \pi_\theta )(h_\theta)| \\
    & \leq \sum_{n=0}^{\infty}
    \| P^n_\theta(\Lambda, \cdot) - \pi_\theta \|_{V^{\alpha}} 
    |h|_{V^{\alpha}}
    \leq \sum_{n=0}^{\infty}
    \left[ \mathbb{E}|Y| C_{\alpha} \rho_{\alpha}^n V^{\alpha}(\Lambda) \right] \\
    &= C_{h,\alpha} V^{\alpha}(\Lambda)
\end{align*}
for any $\theta \in \overline{B(\theta^*,r_d)}$, where $C_{h,\alpha} = \mathbb{E}|Y| \frac{C_\alpha}{1-\rho_\alpha} = \mathbb{E}|Y| L_\alpha^2$.

From Lemma \ref{lemma_solution_continuity} with $S_\Theta = \overline{B(\theta^*, \min\{r_d,r_g\})}$, we can show that for any $\theta$, $\theta^\prime \in S_\Theta$, $\hat{h}_\theta$ is Lipschitz continuous as follows:
\begin{align*}
    & \quad |\hat{h}_\theta(\Lambda) - \hat{h}_{\theta^{\prime}}(\Lambda)| \\
    & \leq |P_\theta\hat{h}_\theta - P_{\theta^{\prime}}\hat{h}_{\theta^{\prime}}|_{V^{\alpha}} V^{\alpha}(\Lambda)
    + |h_\theta(\Lambda) - h_{\theta^{\prime}}(\Lambda)|
    + |\pi_\theta h_\theta - \pi_{\theta^{\prime}} h_{\theta^{\prime}}| \\
    & \leq |P_\theta\hat{h}_\theta - P_{\theta^{\prime}}\hat{h}_{\theta^{\prime}}|_{V^{\alpha}} V^{\alpha}(\Lambda)
    + |h_\theta(\Lambda) - h_{\theta^{\prime}}(\Lambda)|
    + \pi_\theta |h_\theta - h_{\theta^{\prime}}|
    + |\pi_\theta - \pi_{\theta^{\prime}}| (h_{\theta^{\prime}}) \\
    & \leq \sup_{\theta \in S_\Theta} |h_\theta|_{V^{\alpha}} L_{\alpha}^2(L_{\alpha}^2 D_{V^{\alpha}}(\theta, \theta^{\prime}) + \|\pi_\theta - \pi_{\theta^{\prime}}\|_{V^{\alpha}}) V^{\alpha}(\Lambda) \\
    & \quad + L_{\alpha}^2 |h_\theta-h_{\theta^{\prime}}|_{V^{\alpha}} V^{\alpha}(\Lambda)
    + 2L_h \|\theta - \theta^{\prime}\|
    + \sup_{\theta \in S_\Theta} |h_\theta|_{V^{\alpha}} \|\pi_\theta-\pi_{\theta^{\prime}}\|_{V^{\alpha}} V^{\alpha}(\Lambda) \\
    & \leq \mathbb{E}|Y| L_{\alpha}^2(L_{\alpha}^2 L_{P,\alpha} \|\theta - \theta^{\prime}\| + L_{d\pi,\alpha} \|\theta - \theta^{\prime}\|) V^{\alpha}(\Lambda) \\
    & \quad + L_{\alpha}^2 L_h \|\theta - \theta^{\prime}\| V^{\alpha}(\Lambda)
    + 2L_h \|\theta - \theta^{\prime}\|
    + \mathbb{E}|Y| L_{d\pi,\alpha} \|\theta - \theta^{\prime}\| V^{\alpha}(\Lambda) \\
    & \leq \left[ \mathbb{E}|Y| L_{\alpha}^2(L_{\alpha}^2 L_{P,\alpha} + 2 L_{d\pi,\alpha}) + 3 L_{\alpha}^2 L_h \right]
    \|\theta - \theta^{\prime}\| V^{\alpha}(\Lambda)
\end{align*}
for any $\alpha \in (0,1]$.
Denote the constant $\mathbb{E}|Y| L_{\alpha}^2(L_{\alpha}^2 L_{P,\alpha} + 2 L_{d\pi,\alpha}) + 3 L_{\alpha}^2 L_h$ by $L_{dh,\alpha}$.

\subsection{Proof of Lemma \ref{lemma_lipschitz_boundedness_FGH}}

\label{subsec_properties_LLN}

In this subsection, we will frequently use the inequality $P_\theta V^\alpha \leq \beta_\alpha V^\alpha + b_\alpha \leq (\beta_\alpha + b_\alpha) V^\alpha$, which is a direct consequence of Assumption \ref{assumption_simultaneous_geometric_ergodicity}.
Moreover, it follows from Assumption \ref{assumption_sampling} that
\begin{equation*}
    \int \left[ V^{\alpha}(\Lambda+(1-\rho)X)+V^{\alpha}(\Lambda-\rho X) \right] \Gamma(dX)
    \leq \frac{(P_\theta V^{\alpha})(\Lambda)}{\iota} .
\end{equation*}
Finally, we will repeatedly use the results from Lemmas \ref{lemma_lipschitz_continuity_P_pi} and \ref{lemma_lipschitz_continuity_poisson}, namely the Lipschitz continuity of the transition kernel $P_\theta$, the invariant measure $\pi_\theta$ and the solution to the Poisson equation $\hat{h}_\theta$.

\textbf{Properties of $F_\theta$}

Suppose the function
\begin{equation*}
    F_\theta = P_\theta(\hat{h}_\theta^2) - (P_\theta \hat{h}_\theta)^2 .
\end{equation*}

By Lemmas \ref{lemma_lipschitz_continuity_P_pi} and \ref{lemma_lipschitz_continuity_poisson}, the Lipschitz continuity of $P_\theta(\hat{h}_\theta^2)$ with respect to $\theta$ follows from the inequality below:
\begin{align*}
    & \quad |[ P_\theta (\hat{h}_\theta^2) ](\Lambda) - [ P_{\theta^{\prime}} (\hat{h}_{\theta^{\prime}}^2) ](\Lambda)| \\
    & \leq |[ P_\theta (\hat{h}_{\theta^{\prime}}^2) ](\Lambda) - [ P_{\theta^{\prime}} (\hat{h}_{\theta^{\prime}}^2) ](\Lambda)|
    + | [ P_\theta (\hat{h}_\theta^2) ](\Lambda) - [ P_\theta (\hat{h}_{\theta^{\prime}}^2) ](\Lambda) | \\
    & \leq C_{h,\alpha}^2 V^{2\alpha}(\Lambda) \|P_\theta(\Lambda,\cdot) - P_{\theta^{\prime}}(\Lambda,\cdot)\|_{V^{2\alpha}}
    + | [ P_\theta(\hat{h}_\theta^2-\hat{h}_{\theta^{\prime}}^2) ] (\Lambda) | \\
    & \leq C_{h,\alpha}^2 V^{2\alpha}(\Lambda) \|P_\theta - P_{\theta^{\prime}}\|_{V^{2\alpha}}
    + | [ P_\theta(2C_{h, \alpha} V^{\alpha}|\hat{h}_\theta-\hat{h}_{\theta^{\prime}}| ) ] (\Lambda) | \\
    & \leq C_{h,\alpha}^2 L_{P,2\alpha} V^{2\alpha}(\Lambda) \|\theta - \theta^{\prime}\|
    + | [ P_\theta(2C_{h, \alpha} L_{dh,\alpha} \|\theta-\theta^{\prime}\| V^{2\alpha}) ] (\Lambda) | \\
    &= C_{h,\alpha}^2 L_{P,2\alpha} V^{2\alpha}(\Lambda) \|\theta - \theta^{\prime}\|
    + 2C_{h, \alpha} L_{dh,\alpha} \|\theta-\theta^{\prime}\|
    (P_\theta V^{2\alpha})(\Lambda) \\
    & \leq C_{h,\alpha}^2 L_{P,2\alpha} V^{2\alpha}(\Lambda) \|\theta - \theta^{\prime}\|
    + 2C_{h, \alpha} L_{dh,\alpha}
    (\beta_{2\alpha} + b_{2\alpha}) V^{2\alpha}(\Lambda)
    \|\theta-\theta^{\prime}\| \\
    &= [ C_{h,\alpha}^2 L_{P,2\alpha} + 2C_{h, \alpha} L_{dh,\alpha} (\beta_{2\alpha} + b_{2\alpha}) ] V^{2\alpha}(\Lambda) \|\theta - \theta^{\prime}\|
\end{align*}
implies that $[ P_\theta (\hat{h}_\theta^2) ](\Lambda)$ is Lipschitz continuous with respect to $\theta$.

Assumption \ref{assumption_simultaneous_geometric_ergodicity} provides the inequality $P_\theta V^\alpha \leq \beta_\alpha V^\alpha + b_\alpha \leq (\beta_\alpha + b_\alpha) V^\alpha$.
Thus, given the bound on $|\hat{h}_\theta|_{V^{\alpha}}$, we can similarly establish a bound for $P_\theta \hat{h}_\theta$:
\begin{equation*}
    \left| \left( P_\theta \hat{h}_\theta \right) (\Lambda) \right|
    \leq P_\theta \left( C_{h,\alpha} V^{\alpha} \right) (\Lambda)
    \leq C_{Ph,\alpha} V^{\alpha}(\Lambda)
\end{equation*}
for $C_{Ph,\alpha} = (\beta_\alpha+b_\alpha)C_{h,\alpha} > 0$.

Consequently, for $(P_\theta \hat{h}_\theta)^2$, due to $|P_\theta \hat{h}_\theta|(\Lambda) \leq C_{h,\alpha}(\beta_{\alpha} + b_{\alpha}) V^{\alpha}(\Lambda)$, we have
\begin{align*}
    |(P_\theta \hat{h}_\theta)^2(\Lambda) - (P_{\theta^{\prime}} \hat{h}_{\theta^{\prime}})^2(\Lambda)|
    \leq 2C_{h,\alpha}(\beta_{\alpha} + b_{\alpha}) V^{\alpha}(\Lambda) |(P_\theta \hat{h}_\theta)(\Lambda) - (P_{\theta^{\prime}} \hat{h}_{\theta^{\prime}})(\Lambda)| .
\end{align*}
From Lemmas \ref{lemma_lipschitz_continuity_P_pi} and \ref{lemma_solution_continuity}, we can show that
\begin{align*}
    & \quad |(P_\theta \hat{h}_\theta)(\Lambda) - (P_{\theta^{\prime}} \hat{h}_{\theta^{\prime}})(\Lambda)| \\
    & \leq \left[ \sup_{\theta \in S_\Theta} \left|h_\theta\right|_{V^\alpha}
    L_\alpha^2 \left( L_\alpha^2 D_{V^\alpha}\left(\theta, \theta^{\prime}\right)
    + \left\| \pi_\theta-\pi_{\theta^{\prime}} \right\|_{V^\alpha} \right)
    + L_\alpha^2 \left| h_\theta-h_{\theta^{\prime}} \right|_{V^\alpha} \right] V^{\alpha}(\Lambda) \\
    & \leq \left[ \mathbb{E}|Y|
    L_\alpha^2 \left( L_\alpha^2 L_{P,\alpha} \|\theta-\theta^{\prime}\|
    + L_{d\pi,\alpha} \|\theta-\theta^{\prime}\| \right)
    + L_\alpha^2 L_h \right] V^{\alpha}(\Lambda) \\
    &= L_{dPh,\alpha} V^{\alpha}(\Lambda) \|\theta-\theta^{\prime}\| ,
\end{align*}
where $L_{dPh,\alpha} = \mathbb{E}|Y| L_{\alpha}^2(L_{\alpha}^2 L_{P,\alpha} + L_{d\pi,\alpha}) + L_{\alpha}^2 L_h$.

Thus, for any $\alpha \in (0,1/2]$, we define the constant
\begin{equation*}
    L_{dF,2\alpha} = 2C_{h,\alpha}(\beta_{\alpha} + b_{\alpha}) L_{dPh,\alpha}
    + \left[ C_{h,\alpha}^2 L_{P,2\alpha} + 2C_{h, \alpha} L_{dh,\alpha} (\beta_{2\alpha} + b_{2\alpha}) \right] .
\end{equation*}
With this, we have
\begin{align*}
    & \quad |F_\theta(\Lambda) - F_{\theta^{\prime}}(\Lambda)| \\
    & \leq |(P_\theta \hat{h}_\theta)(\Lambda) - (P_{\theta^{\prime}} \hat{h}_{\theta^{\prime}})(\Lambda)|
    \left[ |(P_\theta \hat{h}_\theta)(\Lambda)| + |(P_{\theta^{\prime}} \hat{h}_{\theta^{\prime}})(\Lambda)| \right]
    + |[ P_\theta (\hat{h}_\theta^2) ](\Lambda) - [ P_{\theta^{\prime}} (\hat{h}_{\theta^{\prime}}^2) ](\Lambda)| \\
    & \leq \left\{ 2C_{h,\alpha}(\beta_{\alpha} + b_{\alpha})
    L_{dPh,\alpha}
    + \left[ C_{h,\alpha}^2 L_{P,2\alpha} + 2C_{h, \alpha} L_{dh,\alpha} (\beta_{2\alpha} + b_{2\alpha}) \right] \right\}
    V^{2\alpha}(\Lambda) \|\theta - \theta^{\prime}\| \\
    & \leq L_{dF,2\alpha} V^{2\alpha}(\Lambda) \|\theta - \theta^{\prime}\| .
\end{align*}

We can now establish a bound on $|F_\theta(\Lambda)|$. From Lemma \ref{lemma_lipschitz_continuity_poisson}, we have the bound on $\hat{h}_\theta$.
Furthermore, Assumption \ref{assumption_simultaneous_geometric_ergodicity} provides the inequality $P_\theta V^\alpha \leq \beta_\alpha V^\alpha + b_\alpha \leq (\beta_\alpha + b_\alpha) V^\alpha$.
Combining these results yields
\begin{align*}
    |F_\theta(\Lambda)|
    & \leq \left| [ P_\theta (\hat{h}_\theta^2) ](\Lambda) \right|
    + \left| (P_\theta \hat{h}_\theta)^2(\Lambda) \right|
    \leq P_\theta \left( C_{h,\alpha}^2 V^{2\alpha} \right) (\Lambda)
    + \left[ P_\theta \left( C_{h,\alpha} V^\alpha \right) (\Lambda) \right]^2 \\
    & \leq C_{h,\alpha}^2 (\beta_{2\alpha} + b_{2\alpha}) V^{2\alpha}(\Lambda)
    + C_{h,\alpha}^2 (\beta_{\alpha} + b_{\alpha})^2 V^{2\alpha}(\Lambda) .
\end{align*}
Thus, for any $\alpha \in (0,1/2]$, we define the constant
\begin{equation*}
    C_{F,2\alpha} = C_{h,\alpha}^2 (\beta_{2\alpha} + b_{2\alpha})
    + C_{h,\alpha}^2 (\beta_{\alpha} + b_{\alpha})^2 .
\end{equation*}
With this, we have $|F_\theta| \leq C_{F,2\alpha} V^{2\alpha}$.
In summary, a simple change of variable shows that the constants $L_{dF,\alpha}$ and $C_{F,\alpha}$ are established for all $\alpha \in (0,1]$.

\textbf{Properties of $G_\theta$}

Let $f_2(x) = \mathbb{E} [Y^2 \mid X=x]$. Since $(Y_1,Z_1)$ and $T_1$ are conditionally independent given $X_1$, the function $G_\theta(\Lambda)$ can be expressed as:
\begin{align*}
    G_\theta(\Lambda) &= \mathbb{E}_\theta \left[ (T_1-\rho)^2Y_1^2 - h_\theta^2(\Lambda_0) \mid \Lambda_0=\Lambda \right]
    + \var(Z) \\
    & \quad + 2 \mathbb{E}_\theta \left[ [(T_1-\rho)Y_1-h_\theta(\Lambda)][Z_1-\mathbb{E}Z_1] \mid \Lambda_0=\Lambda \right] \\
    &= \int \left[ \rho^2 + (1-2\rho)g_\theta(\Lambda, X) \right] f_2(X) \Gamma(dX) - h_\theta^2(\Lambda)
    + \var(Z) \\
    & \quad + 2\mathbb{E}_{X \sim \Gamma} \left[ (g_\theta(\Lambda, X)-\rho) \mathbb{E}\left[ (Z-\mathbb{E}Z)Y \mid X \right] \right] .
\end{align*}

Since the function $h_\theta$ is Lipschitz continuous with respect to $\theta$ and bounded by $\mathbb{E}|Y|$, the function $h_\theta^2$ is also Lipschitz continuous.
Given that $g_\theta$ is Lipschitz continuous with a uniform constant $L_g$ for all $X$ and $\Lambda$, it follows that $\int g_\theta(\Lambda, X) f_2(X) \Gamma(dX)$ is Lipschitz continuous with respect to $\theta$, and
\begin{align*}
    |G_\theta(\Lambda) - G_{\theta^{\prime}}(\Lambda)|
    & \leq |1-2\rho| \int |g_\theta(\Lambda, X)-g_{\theta^{\prime}}(\Lambda, X)| f_2(X) \Gamma(dX)
    + |h_\theta^2(\Lambda)-h_{\theta^{\prime}}^2(\Lambda)| \\
    & \quad + 2\mathbb{E}_{X \sim \Gamma} \left[ |g_\theta(\Lambda, X)-g_{\theta^\prime}(\Lambda, X)| \mathbb{E}\left[ |Z-\mathbb{E}Z||Y| \mid X \right] \right] \\
    & \leq \left\{ [ |1-2\rho|\mathbb{E}Y^2 + 2\mathbb{E}\left[ |Z-\mathbb{E}Z||Y| \right] ] L_g + 2 L_h \mathbb{E}|Y| \right\} \|\theta - \theta^{\prime}\| .
\end{align*}
Thus, $G_\theta$ is Lipschitz continuous with respect to $\theta$, and we denote its Lipschitz constant by
\begin{equation*}
    L_{dG}
    = [ |1-2\rho|\mathbb{E}Y^2 + 2\mathbb{E}\left[ |Z-\mathbb{E}Z||Y| \right] ] L_g + 2 L_h \mathbb{E}|Y| .
\end{equation*}
Furthermore, we show that $G_\theta(\Lambda)$ is uniformly bounded. The finiteness of the second moments of $Y$ and $Z$ ensures that each term in the following expression is finite:
\begin{align*}
    \left| G_\theta(\Lambda) \right|
    & \leq \int \left| \rho^2 + (1-2\rho)g_\theta(\Lambda, X) \right| f_2(X) \Gamma(dX) + \{\mathbb{E}|Y|\}^2
    + \var(Z) \\
    & \quad + 2\mathbb{E}_{X \sim \Gamma} \left| (g_\theta(\Lambda, X)-\rho) \mathbb{E}\left[ (Z-\mathbb{E}Z)Y \mid X \right] \right| \\
    & \leq \int f_2(X) \Gamma(dX) + \{\mathbb{E}|Y|\}^2 + \var(Z) + 2\mathbb{E}_{X \sim \Gamma} \left[ \left| \mathbb{E}\left[ (Z-\mathbb{E}Z)Y \mid X \right] \right| \right] \\
    & \leq 2\mathbb{E}Y^2 + \var(Z) + 2 \mathbb{E}|(Z-\mathbb{E}Z)Y| .
\end{align*}
Thus, $G_\theta$ is uniformly bounded, and we denote this bound by the constant $C_G = 2\mathbb{E}Y^2 + \var(Z) + 2 \mathbb{E}|(Z-\mathbb{E}Z)Y|$.

\textbf{Properties of $H_\theta$}

Denote the conditional expectation $\mathbb{E}[Z \mid X=x]$ by $f_Z(x)$. Define
\begin{align*}
    & \quad H_\theta(\Lambda) \\
    &= \mathbb{E}_\theta \left[ \left[ (T_1-\rho) f(X_1) - h_\theta(\Lambda_0) + f_Z(X_1)-\mathbb{E}Z \right]
    \left[ \hat{h}_\theta(\Lambda_1) - (P_\theta\hat{h}_\theta)(\Lambda_0) \right] \mid \Lambda_0=\Lambda \right] \\
    &= \int \left[
    g_\theta(\Lambda, X)[(1-\rho) f(X)+f_Z(X)]
    \hat{h}_\theta(\Lambda+(1-\rho)X) \right. \\
    & \quad \left. + (1-g_\theta(\Lambda, X))[-\rho f(X)+f_Z(X)]
    \hat{h}_\theta(\Lambda-\rho X) \right] \Gamma(dX)
    - [h_\theta(\Lambda)+\mathbb{E}Z](P_\theta\hat{h}_\theta)(\Lambda)
\end{align*}

Then
\begin{align*}
    & \quad |H_\theta(\Lambda)| \\
    & \leq \int C_{h,\alpha} \left[ V^\alpha(\Lambda+(1-\rho)X)+V^\alpha(\Lambda-\rho X) \right] [|f(X)|+|f_Z(X)|] \Gamma(dX) \\
    & \quad + C_{Ph,\alpha} [\mathbb{E}|Y| + \mathbb{E}|Z|] V^{\alpha}(\Lambda) \\
    & \leq C_{h,\alpha} \sqrt{\int \left[ V^{\alpha}(\Lambda+(1-\rho)X)+V^{\alpha}(\Lambda-\rho X) \right]^2 \Gamma(dX)} \\
    & \quad \cdot \sqrt{\int \left[ |f(X)|+|f_Z(X)| \right]^2 \Gamma(dX)} + C_{Ph,\alpha} [\mathbb{E}|Y| + \mathbb{E}|Z|] V^{\alpha}(\Lambda) \\
    & \leq 2C_{h,\alpha} \sqrt{\int \left[ V^{2\alpha}(\Lambda+(1-\rho)X)+V^{2\alpha}(\Lambda-\rho X) \right] \Gamma(dX)} \\
    & \quad \cdot \sqrt{\int f^2(X)\Gamma(dX) + \int f_Z^2(X)\Gamma(dX)} + C_{Ph,\alpha} [\mathbb{E}|Y| + \mathbb{E}|Z|] V^{\alpha}(\Lambda) \\
    & \leq 2C_{h,\alpha}\sqrt{\mathbb{E}Y^2+\mathbb{E}Z^2}
    \sqrt{\frac{\beta_{2\alpha}+b_{2\alpha}}{\iota} V^{2\alpha}(\Lambda)}
    + C_{Ph,\alpha} [\mathbb{E}|Y| + \mathbb{E}|Z|] V^{\alpha}(\Lambda) \\
    &= \left[ 2C_{h,\alpha}\sqrt{\mathbb{E}Y^2+\mathbb{E}Z^2}
    \sqrt{\frac{\beta_{2\alpha}+b_{2\alpha}}{\iota}}
    + C_{Ph,\alpha} [\mathbb{E}|Y| + \mathbb{E}|Z|] \right] V^{\alpha}(\Lambda) ,
\end{align*}
where the second inequality follows from the Cauchy-Schwarz inequality and the fourth inequality follows from the inequality $P_\theta V^{2\alpha} \leq \beta_{2\alpha} V^{2\alpha} + b_{2\alpha}$ when $\alpha \in (0,1/2]$.
Thus, we have shown that for any $\alpha \in (0,1/2]$, $|H_\theta(\Lambda)| \leq C_{H,\alpha} V^{\alpha}(\Lambda)$, where the constant
\begin{equation*}
    C_{H,\alpha}
    = 2C_{h,\alpha}\sqrt{\mathbb{E}Y^2+\mathbb{E}Z^2}
    \sqrt{\frac{\beta_{2\alpha}+b_{2\alpha}}{\iota}}
    + C_{Ph,\alpha} [\mathbb{E}|Y| + \mathbb{E}|Z|] .
\end{equation*}

For the Lipschitz continuity,
\begin{align*}
    & \quad |H_\theta(\Lambda) - H_{\theta^{\prime}}(\Lambda)| \\
    & \leq \int \left[ |\hat{h}_\theta(\Lambda+(1-\rho)X)-\hat{h}_{\theta^{\prime}}(\Lambda+(1-\rho)X)| + |\hat{h}_\theta(\Lambda-\rho X)-\hat{h}_{\theta^{\prime}}(\Lambda-\rho X)| \right] \\
    & \quad \quad [|f(X)|+|f_Z(X)|] \Gamma(dX) \\
    & \quad + \int |g_\theta(\Lambda, X)-g_{\theta^{\prime}}(\Lambda, X)|(|\hat{h}_{\theta^{\prime}}(\Lambda+(1-\rho)X)|+|\hat{h}_{\theta^{\prime}}(\Lambda-\rho X)|) \\
    & \quad \quad [|f(X)|+|f_Z(X)|] \Gamma(dX) \\
    & \quad + [\mathbb{E}{|Y|}+\mathbb{E}{|Z|}] |(P_\theta\hat{h}_\theta)(\Lambda) - (P_{\theta^{\prime}}\hat{h}_{\theta^{\prime}})(\Lambda)| \\
    & \quad + |h_\theta(\Lambda)-h_{\theta^{\prime}}(\Lambda)| (|(P_\theta\hat{h}_\theta)(\Lambda)|+|(P_{\theta^{\prime}}\hat{h}_{\theta^{\prime}})(\Lambda)|) \\
    & \leq \int \left[ (L_{dh,\alpha}+C_{h,\alpha}L_g) \|\theta - \theta^{\prime}\| (V^{\alpha}(\Lambda+(1-\rho)X)+V^{\alpha}(\Lambda-\rho X)) \right] \\
    & \quad \quad [|f(X)|+|f_Z(X)|] \Gamma(dX) \\
    & \quad + [\mathbb{E}{|Y|}+\mathbb{E}{|Z|}] L_{dPh,\alpha} V^{\alpha}(\Lambda) \|\theta-\theta^{\prime}\| + 2 L_h C_{Ph,\alpha} V^{\alpha}(\Lambda) \|\theta-\theta^{\prime}\| \\
    & \leq 2(L_{dh,\alpha}+C_{h,\alpha}L_g) \|\theta - \theta^{\prime}\|
    \sqrt{\int \left[ V^{2\alpha}(\Lambda+(1-\rho)X)+V^{2\alpha}(\Lambda-\rho X) \right] \Gamma(dX)} \\
    & \quad \sqrt{\int f^2(X) \Gamma(dX) + \int f_Z^2(X) \Gamma(dX)} \\
    & \quad + [\mathbb{E}{|Y|}+\mathbb{E}{|Z|}] L_{dPh,\alpha} V^{\alpha}(\Lambda) \|\theta-\theta^{\prime}\| + 2 L_h C_{Ph,\alpha} V^{\alpha}(\Lambda) \|\theta-\theta^{\prime}\| \\
    & \leq 2(L_{dh,\alpha}+C_{h,\alpha}L_g) \sqrt{\mathbb{E}Y^2+\mathbb{E}Z^2}
    \sqrt{\frac{\beta_{2\alpha}+b_{2\alpha}}{\iota}}
    V^{\alpha}(\Lambda) \|\theta - \theta^{\prime}\| \\
    & \quad + [\mathbb{E}{|Y|}+\mathbb{E}{|Z|}] L_{dPh,\alpha} V^{\alpha}(\Lambda) \|\theta-\theta^{\prime}\|
    + 2 L_h C_{Ph,\alpha} V^{\alpha}(\Lambda) \|\theta-\theta^{\prime}\| \\
    & \leq \left[ 2(L_{dh,\alpha}+C_{h,\alpha}L_g)
    \sqrt{\mathbb{E}Y^2+\mathbb{E}Z^2}
    \sqrt{\frac{\beta_{2\alpha}+b_{2\alpha}}{\iota}}
    + [\mathbb{E}{|Y|}+\mathbb{E}{|Z|}] L_{dPh,\alpha} + 2 L_h C_{Ph,\alpha} \right] \\
    & \quad V^{\alpha}(\Lambda) \|\theta-\theta^{\prime}\| ,
\end{align*}
where the third inequality follows from the Cauchy-Schwarz inequality and the fourth inequality follows from the inequality $P_\theta V^{2\alpha} \leq \beta_{2\alpha} V^{2\alpha} + b_{2\alpha}$ when $\alpha \in (0,1/2]$.

Thus, for any $\alpha \in (0,1/2]$, we define the constant
\begin{equation*}
    L_{H,\alpha} = 2(L_{dh,\alpha}+C_{h,\alpha}L_g)
    \sqrt{\mathbb{E}Y^2+\mathbb{E}Z^2}
    \sqrt{\frac{\beta_{2\alpha}+b_{2\alpha}}{\iota}}
    + [\mathbb{E}{|Y|}+\mathbb{E}{|Z|}] L_{dPh,\alpha} + 2 L_h C_{Ph,\alpha} .
\end{equation*}
With this, we have
\begin{equation*}
    |H_\theta(\Lambda) - H_{\theta^{\prime}}(\Lambda)| \leq L_{H,\alpha} \|\theta - \theta^{\prime}\| V^{\alpha}(\Lambda) .
\end{equation*}

While the proofs above require $\alpha \in (0,1/2]$, the resulting inequalities can be extended to all $\alpha \in (0,1]$. For any $\alpha \in (1/2,1]$, we can use the established bounds for $\alpha=1/2$. Since $V(\Lambda) \geq 1$, we have $V^{1/2}(\Lambda) \leq V^{\alpha}(\Lambda)$, which allows us to set $C_{H,\alpha} = C_{H,1/2}$ and $L_{H,\alpha} = L_{H,1/2}$ for this range.

\section{Other Lemmas}
\label{sec_lemma}

The following lemma is based on Lemma 2.3 in (\cite{fortConvergenceAdaptiveInteracting2011}).
\begin{lemma}
    \label{lemma_simultaneous_geometric_ergodicity}
    Assume that for all $\theta \in S_\Theta$, $P_\theta$ is $\pi$-irreducible and aperiodic.
Moreover, there exist some constants $b<\infty, \delta \in (0,1)$, $\beta \in (0,1)$, a probability measure $\nu$ on X and a function $V: \mathrm{X} \to [1,+\infty)$, such that for any $\theta \in S_\Theta$,
    \begin{align*}
        P_\theta V & \leq \beta V+b, \\
        P_\theta(x, \cdot) & \geq \delta \nu(\cdot) \mathbb{I}_{\left\{V \leq c\right\}}(x), \quad
        c := 2 b\left(1-\beta\right)^{-1} .
    \end{align*}
    Then there are some universal constants $C$ and $\gamma$ such that for any $\theta \in S_\Theta$, there exists a probability distribution $\pi_\theta$ such that $\pi_\theta P_\theta=\pi_\theta$, $\pi_\theta(V) \leq b\left(1-\beta\right)^{-1}$ and the inequality
    \begin{equation*}
        \left\|P_\theta^n(x, \cdot)-\pi_\theta\right\|_V \leq L(1-L^{-1})^n V(x)
    \end{equation*}
    holds with $L = C\left\{b \vee \delta^{-1} \vee\left(1-\beta\right)^{-1} \vee c\right\}^\gamma$.
\end{lemma}
\begin{remark}
    There are little differences between this lemma and Lemma 2.3 in (\cite{fortConvergenceAdaptiveInteracting2011}). The computable bound proposed in the first bound expression of Theorem 2.3 in (\cite{meynComputableBoundsGeometric1994}) explains why the constants $C$, $\gamma$ are universal and why the setting of $c$, $L$ can be different from Lemma 2.3 in (\cite{fortConvergenceAdaptiveInteracting2011}).
    The transformation from the computable bound to the expression of $L$ can also refer to the proof of the Lemma 3 in (\cite{saksmanErgodicityAdaptiveMetropolis2010}), which indicates that we can bound the computable bound by a constant multiple of the product of $b$, $\delta^{-1}$, $(1-\beta)^{-1}$ and $c$.
    
    Thus, no matter how we change $V$, as long as $b$, $\beta$ and $\delta$ remains unchanged, the constant $L$ also remains unchanged. Another point need to mention is that the small condition is on a specific small set which is determined by $b$ and $\beta$. The expression is delicate that if we change $P_\theta$ to $P_\theta^d$, then the specific small set remains unchanged for $P_\theta^d$.
\end{remark}

The lemma below is a variant version of Lemma 4.2 in (\cite{fortConvergenceAdaptiveInteracting2011}).
From the proof, we can show that the condition substantially used in the proof is simultaneous geometrically ergodicity. Thus, we only provide this ergodicity as the condition in Lemma \ref{lemma_solution_continuity}.
\begin{lemma}
    \label{lemma_solution_continuity}
    Assume that there exists a positive constant $L>1$, for all $\theta \in S_\Theta$, $P_\theta$ is positive recurrent and simultaneously geometrically ergodic that 
    \begin{equation*}
        \| P_\theta^{n}(x, \cdot) - \pi_\theta \|_V \leq L(1-L^{-1})^n V(x) ,
    \end{equation*}
    for some function $V: \mathrm{X} \to [1,+\infty)$.
    For any $\theta \in S_\Theta$, let $F_\theta: \mathrm{X} \to \mathbb{R}^{+}$ be a measurable function such that $\sup_\theta\left|F_\theta\right|_V < \infty$ and define $\hat{F}_\theta := \sum_{n \geq 0} P_\theta^n\left\{F_\theta-\pi_\theta\left(F_\theta\right)\right\}$.
    For any $\theta, \theta^{\prime} \in S_\Theta$,
    \begin{equation*}
        \left\|\pi_\theta-\pi_{\theta^{\prime}}\right\|_V
        \leq L^2 \left\{\pi_\theta(V)+L V(x)\right\} D_V\left(\theta, \theta^{\prime}\right)
    \end{equation*}
    and
    \begin{equation*}
        \left|P_\theta \hat{F}_\theta-P_{\theta^{\prime}} \hat{F}_{\theta^{\prime}}\right|_V
        \leq \sup_{\theta \in S_\Theta} \left|F_\theta\right|_V
        L^2 \left( L^2 D_V\left(\theta, \theta^{\prime}\right)
        + \left\| \pi_\theta-\pi_{\theta^{\prime}} \right\|_V \right)
        + L^2 \left| F_\theta-F_{\theta^{\prime}} \right|_V .
    \end{equation*}
\end{lemma}

\begin{lemma}
    \label{lemma_convergence_V_as}
    Suppose $\mathbb{E} V(\Lambda_n)$ is finite for any $n \geq 1$ and the inequality $(P_\theta V)(\Lambda) \leq \beta V(\Lambda)+b$ holds for any $\theta \in S_\Theta$ and any $\Lambda \in W_\Gamma$.

    If $P(\theta_n \notin S_\Theta \text{ i.o.}) = 0$, then $V(\Lambda_n) = O_P(1)$ and the summation $\sum_{n=1}^\infty n^{-p} V(\Lambda_n) < \infty$ almost surely for any $p>1$.
\end{lemma}
\begin{proof}
    Denote the variable $V(\Lambda_n)\mathbb{I}(\theta_m, \dots, \theta_{n-1} \in S_\Theta)$ by $V_{m,n}$ with $m \leq n$.
    Thus, for $n>m$,
    \begin{align*}
        \mathbb{E} V_{m,n}
        &= \mathbb{E} \left[
            \mathbb{E} \left[
                V(\Lambda_n)\mathbb{I}(\theta_m, \dots, \theta_{n-1} \in S_\Theta)
                \mid \mathcal{F}_{n-1}
            \right]
        \right] \\
        &= \mathbb{E} \left[
            \mathbb{E} \left[
                V(\Lambda_n)
                \mid \mathcal{F}_{n-1}
            \right]
            \mathbb{I}(\theta_m, \dots, \theta_{n-1} \in S_\Theta)
        \right] \\
        &= \mathbb{E} \left[
            (P_{\theta_{n-1}}V) (\Lambda_{n-1})
            \mathbb{I}(\theta_m, \dots, \theta_{n-1} \in S_\Theta)
        \right] \\
        & \leq \mathbb{E} \left[
            (\beta V+b) (\Lambda_{n-1})
            \mathbb{I}(\theta_m, \dots, \theta_{n-1} \in S_\Theta)
        \right] \\
        & \leq \beta \mathbb{E} V_{m,n-1}+b
    \end{align*}

    By recursion of the above inequality, we can show that for any $n \geq m$,
    \begin{equation*}
        \mathbb{E} V_{m,n}
        \leq \beta^{n-m} \mathbb{E} V_{m,m} + b \sum_{k=0}^{n-m-1} \beta^k
        \leq \beta^{n-m} \mathbb{E} V(\Lambda_m) + \frac{b}{1-\beta} .
    \end{equation*}
    The assumption $\mathbb{E} V(\Lambda_m) < \infty$ implies that for a fixed $m$, $\mathbb{E} V_{m,n}$ is uniformly bounded for all $n \geq m$.

    For any $\epsilon > 0$, since the assumption $P(\theta_n \notin S_\Theta \text{ i.o.}) = 0$, there exists positive integer $N$ such that $P(\theta_n \in S_\Theta \text{ for all } n \geq N) > 1-\epsilon/2$.
    Equivalently, this means $P(V_{N,n} = V(\Lambda_n) \text{ for any } n \geq N) \geq 1-\epsilon/2$.
    We have shown the uniform boundedness of $\mathbb{E} V_{N,n}$ for fixed $N$. Then there exists a positive number $C$ such that $\mathbb{E} V_{N,n} < C$ for any $n \geq N$. From the bound of this expectation, we can show that $P(V_{N,n} \geq \frac{2C}{\epsilon}) < \frac{\epsilon}{2}$.

    Combining the inequality
    \begin{equation*}
        P(V_{N,n} = V(\Lambda_n) \text{ for any } n \geq N) \geq 1-\epsilon/2
    \end{equation*}
    and for any $n \geq N$,
    \begin{equation*}
        P(V_{N,n} \geq \frac{2C}{\epsilon}) < \frac{\epsilon}{2} ,
    \end{equation*}
    we can reach to a conclusion that
    \begin{equation*}
        P(V(\Lambda_n) \geq \frac{2C}{\epsilon}) < \epsilon
    \end{equation*}
    for any $n \geq N$. Combining with finiteness of $\mathbb{E} V(\Lambda_n)$ for $n \in \{1,\dots,N-1\}$, we can show that there exists a positive number $M$ such that $P(V(\Lambda_n) \geq M) < \epsilon$ for any $n \in \mathbb{N}^*$. Thus, $V(\Lambda_n) = O_P(1)$.
    
    Furthermore, the uniform boundedness of $\mathbb{E} V_{N,n}$ allows us to establish the almost sure convergence of the series $\sum_{n=N}^\infty n^{-p} V_{N,n}$.
    By the Monotone Convergence Theorem,
    \begin{equation*}
        \mathbb{E} \left[ \sum_{n=N}^\infty n^{-p} V_{N,n} \right] = \sum_{n=N}^\infty n^{-p} \mathbb{E} [V_{N,n}] < \infty ,
    \end{equation*}
    for $p>1$.
    This implies that the sum $\sum_{n=N}^\infty n^{-p} V_{N,n}$ must be finite almost surely.
    With
    \begin{equation*}
        P(V_{N,n} = V(\Lambda_n) \text{ for any } n \geq N) \geq 1-\epsilon/2 ,
    \end{equation*}
    we can show that
    \begin{align*}
        & \quad P(\sum_{n=N}^\infty n^{-p} V(\Lambda_n) = \infty)
        \leq P(\sum_{n=N}^\infty n^{-p} V(\Lambda_n) \neq \sum_{n=N}^\infty n^{-p} V_{N,n}) \\
        & \leq P(V_{N,n} \neq V(\Lambda_n) \text{ for some } n \geq N)
        \leq \epsilon/2 .
    \end{align*}
    Furthermore, combining with the finiteness of $V(\Lambda_n)$ for $n \in \{1,\dots,N-1\}$, it follows that $P(\sum_{n=1}^\infty n^{-p} V(\Lambda_n) = \infty) \leq \epsilon/2$. Due to the arbitrariness of $\epsilon$, the finiteness is also almost sure.
\end{proof}

The next lemma is adapted from Theorem B.1 in (\cite{fortCentralLimitTheorem2014}).
\begin{lemma}
    \label{lemma_convergence_average_enter}
    For some $a \in [1/2,1)$, let the set $S_\Theta$ satisfy that for any $\theta \in S_\Theta$,
    \begin{equation*}
        P_\theta V \leq \beta V+b ,
    \end{equation*}
    and
    \begin{equation*}
        \| P_\theta^{n}(\Lambda, \cdot) - \pi_\theta \|_{V^a} \leq L_a(1-L_a^{-1})^n V^a(\Lambda) ,
    \end{equation*}
    for some function $V: \mathrm{X} \to [1,+\infty)$.
    
    For any $\theta \in S_\Theta$, let $F_\theta: \mathrm{X} \to \mathbb{R}$ be a measurable function.
    If
    \begin{enumerate}[label=\roman*.]
        \item $P(\theta_n \notin S_\Theta \text{ i.o.}) = 0$.
        \item $\sup_{\theta \in S_\Theta} \pi_\theta(V) < \infty$.
        \item $\sup_{\theta \in S_\Theta} |F_\theta|_V \leq \sup_{\theta \in S_\Theta} |F_\theta|_{V^a} < \infty$.
        \item $V(\Lambda_n) = O_P(1)$ and the summation $\sum_{n=1}^\infty n^{-1/a} V(\Lambda_n) < \infty$ almost surely.
        \item $\frac{1}{N} \sum_{n=1}^{N-1} D_V\left(\theta_{n+1}, \theta_n\right) V(\Lambda_{n+1}) \mathbb{I}(\theta_n \in S_\Theta, \theta_{n+1} \in S_\Theta) \xrightarrow{\mathbb{P}} 0$.
        \item $\frac{1}{N} \sum_{n=1}^{N-1} \left|F_{\theta_{n+1}}-F_{\theta_n}\right|_V V(\Lambda_{n+1}) \mathbb{I}(\theta_n \in S_\Theta, \theta_{n+1} \in S_\Theta) \xrightarrow{\mathbb{P}} 0$.
    \end{enumerate}
    Then,
    \begin{equation*}
        \frac{1}{N} \sum_{n=1}^{N} F_{\theta_n}(\Lambda_n)\mathbb{I}(\theta_n \in S_\Theta)
        -\frac{1}{N} \sum_{n=1}^{N} \pi_{\theta_n}F_{\theta_n}\mathbb{I}(\theta_n \in S_\Theta)
        \xrightarrow{\mathbb{P}} 0 .
    \end{equation*}
\end{lemma}

\begin{proof}
    For any $\theta \in S_\Theta$, define $\hat{F}_\theta := \sum_{n \geq 0} P_\theta^n\left\{F_\theta-\pi_\theta\left(F_\theta\right)\right\}$.

    The expression
    \begin{equation*}
        \frac{1}{N} \sum_{n=1}^{N} F_{\theta_n}(\Lambda_n)\mathbb{I}(\theta_n \in S_\Theta)
        -\frac{1}{N} \sum_{n=1}^{N} \pi_{\theta_n}F_{\theta_n}\mathbb{I}(\theta_n \in S_\Theta)
    \end{equation*}
    can be decomposed as follows:
    \begin{enumerate}
        \item \label{item_proof_lemma_convergence_average_enter_1} $\frac{1}{N} \sum_{n=1}^{N-1} \left[ \hat{F}_{\theta_n}(\Lambda_{n+1})
        -(P_{\theta_n}\hat{F}_{\theta_n})(\Lambda_n) \right]
        \mathbb{I}(\theta_n \in S_\Theta)$.
        \item \label{item_proof_lemma_convergence_average_enter_2} $\frac{1}{N} \sum_{n=1}^{N-1} \left[ \hat{F}_{\theta_{n+1}}(\Lambda_{n+1})
        -\hat{F}_{\theta_n}(\Lambda_{n+1}) \right]
        \mathbb{I}(\theta_n \in S_\Theta, \theta_{n+1} \in S_\Theta)$.
        \item \label{item_proof_lemma_convergence_average_enter_3} $\frac{1}{N} \sum_{n=1}^{N-1} -\hat{F}_{\theta_n}(\Lambda_{n+1})
        \mathbb{I}(\theta_n \in S_\Theta, \theta_{n+1} \notin S_\Theta)$.
        \item \label{item_proof_lemma_convergence_average_enter_4} $\frac{1}{N} \sum_{n=2}^{N} \hat{F}_{\theta_n}(\Lambda_n)
        \mathbb{I}(\theta_n \in S_\Theta, \theta_{n-1} \notin S_\Theta)$.
        \item \label{item_proof_lemma_convergence_average_enter_5} $\frac{1}{N} \left[
        \hat{F}_{\theta_1}(\Lambda_1)\mathbb{I}(\theta_1 \in S_\Theta)
        + \left\{ - \hat{F}_{\theta_N}(\Lambda_N)
        + F_{\theta_N}(\Lambda_N)
        - \pi_{\theta_N}F_{\theta_N} \right\}
        \mathbb{I}(\theta_N \in S_\Theta)
        \right]$.
    \end{enumerate}
    
    We first prove the strong law of large number for the Item \ref{item_proof_lemma_convergence_average_enter_1}.
    We use the Jensen's inequality and the inequality $P_\theta V \leq \beta V+b$ that
    \begin{align*}
        & \quad \mathbb{E} \left[ \left\{\left[ \hat{F}_{\theta_n}(\Lambda_{n+1})-(P_{\theta_n}\hat{F}_{\theta_n})(\Lambda_n) \right]
        \mathbb{I}(\theta_n \in S_\Theta) \right\}^{1/a} \mid \mathcal{F}_n \right] \\
        & \leq 2^{1/a-1} \mathbb{E} \left[ \left| \hat{F}_{\theta_n}(\Lambda_{n+1}) \right|^{1/a} \mathbb{I}(\theta_n \in S_\Theta) 
        + \left| (P_{\theta_n}\hat{F}_{\theta_n})(\Lambda_n) \right|^{1/a}
        \mathbb{I}(\theta_n \in S_\Theta) \mid \mathcal{F}_n \right] \\
        & \leq 2^{1/a-1} \mathbb{E} \left[ \left| \hat{F}_{\theta_n} \right|^{1/a}(\Lambda_{n+1}) \mathbb{I}(\theta_n \in S_\Theta) 
        + (P_{\theta_n}\left|\hat{F}_{\theta_n}\right|^{1/a})(\Lambda_n)
        \mathbb{I}(\theta_n \in S_\Theta) \mid \mathcal{F}_n \right] \\
        & \leq 2^{1/a} \left( L_a^2|F_{\theta_n}|_{V^a} \right)^{1/a} (P_{\theta_n}V)(\Lambda_n) \mathbb{I}(\theta_n \in S_\Theta) \\
        & \leq 2^{1/a} \left( L_a^2|F_{\theta_n}|_{V^a} \right)^{1/a} (\beta+b)V(\Lambda_n) \mathbb{I}(\theta_n \in S_\Theta) .
    \end{align*}

    By the assumption
    \begin{equation*}
        \sum_{n=1}^\infty n^{-1/a} V(\Lambda_n) < \infty
    \end{equation*}
    almost surely, we have
    \begin{equation*}
        \sum_{n=1}^\infty n^{-1/a} V(\Lambda_n)\mathbb{I}(\theta_n \in S_\Theta) < \infty
    \end{equation*}
    almost surely. Therefore,
    \begin{equation*}
        \sum_{n=1}^\infty n^{-1/a}
        \mathbb{E} \left[ \left\{\left[ \hat{F}_{\theta_n}(\Lambda_{n+1})-(P_{\theta_n}\hat{F}_{\theta_n})(\Lambda_n) \right]
        \mathbb{I}(\theta_n \in S_\Theta) \right\}^{1/a} \mid \mathcal{F}_n \right] < \infty .
    \end{equation*}
    By Theorem 2.18 in (\cite{hallMartingaleLimitTheory1980}), we have
    \begin{equation*}
        \frac{1}{N} \sum_{n=1}^{N-1} \left[ \hat{F}_{\theta_n}(\Lambda_{n+1})
        -(P_{\theta_n}\hat{F}_{\theta_n})(\Lambda_n) \right]
        \mathbb{I}(\theta_n \in S_\Theta) \rightarrow 0
    \end{equation*}
    almost surely.

    Next, we prove that Item \ref{item_proof_lemma_convergence_average_enter_2} converges to zero in probability.
    \begin{align*}
        & \quad \left| \frac{1}{N} \sum_{n=1}^{N-1} \left[ \hat{F}_{\theta_{n+1}}(\Lambda_{n+1})
        -\hat{F}_{\theta_n}(\Lambda_{n+1}) \right]
        \mathbb{I}(\theta_n \in S_\Theta, \theta_{n+1} \in S_\Theta) \right| \\
        & \leq \frac{1}{N} \sum_{n=1}^{N-1} \left|
        \hat{F}_{\theta_{n+1}}(\Lambda_{n+1})
        -\hat{F}_{\theta_n}(\Lambda_{n+1}) \right|
        \mathbb{I}(\theta_n \in S_\Theta, \theta_{n+1} \in S_\Theta) \\
        & \leq \frac{1}{N} \sum_{n=1}^{N-1}
        \left[
            \left| (P_{\theta_{n+1}}\hat{F}_{\theta_{n+1}})(\Lambda_{n+1})
            - (P_{\theta_n}\hat{F}_{\theta_n})(\Lambda_{n+1}) \right|
            + \left| F_{\theta_{n+1}}(\Lambda_{n+1}) - F_{\theta_n}(\Lambda_{n+1}) \right| \right. \\
            & \quad \left. + \left| \pi_{\theta_{n+1}} F_{\theta_{n+1}} - \pi_{\theta_n} F_{\theta_n} \right|
        \right]
        \mathbb{I}(\theta_n \in S_\Theta, \theta_{n+1} \in S_\Theta) \\
        & \leq \frac{1}{N} \sum_{n=1}^{N-1}
        \left[
            \left| (P_{\theta_{n+1}}\hat{F}_{\theta_{n+1}})(\Lambda_{n+1})
            - (P_{\theta_n}\hat{F}_{\theta_n})(\Lambda_{n+1}) \right|
            + \left| F_{\theta_{n+1}}(\Lambda_{n+1}) - F_{\theta_n}(\Lambda_{n+1}) \right| \right. \\
            & \quad \left. + \pi_{\theta_n} \left| F_{\theta_{n+1}} - F_{\theta_n} \right|
            + \left| \pi_{\theta_{n+1}} - \pi_{\theta_n} \right| F_{\theta_{n+1}}
        \right]
        \mathbb{I}(\theta_n \in S_\Theta, \theta_{n+1} \in S_\Theta) \\
        & \leq \frac{1}{N} \sum_{n=1}^{N-1}
        \left[
            \sup_{\theta \in S_\Theta} \left|F_\theta\right|_V
            \left(L^4 D_V\left(\theta_{n+1}, \theta_n\right)
            + (L^2+1)\left\|\pi_{\theta_{n+1}}-\pi_{\theta_n}\right\|_V\right) \right. \\
        & \quad \left.
            + \left( L^2+1+\sup_{\theta \in S_\Theta} \pi_\theta(V) \right)
            \left|F_{\theta_{n+1}}-F_{\theta_n}\right|_V
        \right]
        V(\Lambda_{n+1}) \mathbb{I}(\theta_n \in S_\Theta, \theta_{n+1} \in S_\Theta) \\
        & \leq \frac{1}{N} \sum_{n=1}^{N-1}
        \left[
            \sup_{\theta \in S_\Theta} \left|F_\theta\right|_V
            \left(L^4 + L^2(L^2+1)\sup_{\theta \in S_\Theta} \pi_\theta(V)+L^5 V(x)\right)
            D_V\left(\theta_{n+1}, \theta_n\right) \right. \\
        & \quad \left.
            + \left( L^2+1+\sup_{\theta \in S_\Theta} \pi_\theta(V) \right)
            \left|F_{\theta_{n+1}}-F_{\theta_n}\right|_V
        \right] V(\Lambda_{n+1}) \mathbb{I}(\theta_n \in S_\Theta, \theta_{n+1} \in S_\Theta)
    \end{align*}
    
    Because $V(x) < \infty$, $\sup_{\theta \in S_\Theta} \pi_\theta(V) < \infty$ and $\sup_{\theta \in S_\Theta} \left|F_\theta\right|_V < \infty$, we can show that the above formula converges to zero almost surely from
    \begin{equation*}
        \frac{1}{N} \sum_{n=1}^{N-1} D_V\left(\theta_{n+1}, \theta_n\right) V(\Lambda_{n+1}) \mathbb{I}(\theta_n \in S_\Theta, \theta_{n+1} \in S_\Theta) \xrightarrow{\mathbb{P}} 0 ,
    \end{equation*}
    and
    \begin{equation*}
        \frac{1}{N} \sum_{n=1}^{N-1} \left|F_{\theta_{n+1}}-F_{\theta_n}\right|_V V(\Lambda_{n+1}) \mathbb{I}(\theta_n \in S_\Theta, \theta_{n+1} \in S_\Theta) \xrightarrow{\mathbb{P}} 0 .
    \end{equation*}
    
    Next, consider Items \ref{item_proof_lemma_convergence_average_enter_3} and \ref{item_proof_lemma_convergence_average_enter_4}.
    The assumption $P(\theta_n \notin S_\Theta \text{ i.o.}) = 0$ implies that the number of non-zero terms in each summation is finite almost surely.
    Thus, the entire sum is a finite random variable.
    When divided by $N$, it converges to zero almost surely as $N \rightarrow \infty$.  
    
    Finally, for Item \ref{item_proof_lemma_convergence_average_enter_5}, we have
    \begin{equation*}
        \left| F_{\theta_N}(\Lambda_N) \right|
        \leq \sup_{\theta \in S_\Theta} |F_{\theta_n}|_{V^a} V^a(\Lambda_N) ,
    \end{equation*}
    \begin{equation*}
        \left| \hat{F}_{\theta_N}(\Lambda_N)\mathbb{I}(\theta_N \in S_\Theta) \right|
        \leq \sup_{\theta \in S_\Theta} L_a^2 |F_{\theta_n}|_{V^a} V^a(\Lambda_N)
    \end{equation*}
    and
    \begin{equation*}
        \left| \pi_{\theta_N}F_{\theta_N}\mathbb{I}(\theta_N \in S_\Theta) \right|
        \leq \sup_{\theta \in S_\Theta} |F_{\theta_n}|_{V^a} \sup_{\theta \in S_\Theta} \pi_\theta(V^a)
        \leq \sup_{\theta \in S_\Theta} |F_{\theta_n}|_{V^a} \sup_{\theta \in S_\Theta} \pi_\theta(V) .
    \end{equation*}
    Thus, we can show that Item \ref{item_proof_lemma_convergence_average_enter_5} converges to zero in probability from the assumption that $V(\Lambda_N)$ is $O_P(1)$.

    Combining the above convergence results, we obtain
    \begin{equation*}
        \frac{1}{N} \sum_{n=1}^{N} F_{\theta_n}(\Lambda_n)\mathbb{I}(\theta_n \in S_\Theta)
        -\frac{1}{N} \sum_{n=1}^{N} \pi_{\theta_n}F_{\theta_n}\mathbb{I}(\theta_n \in S_\Theta)
        \xrightarrow{\mathbb{P}} 0 .
    \end{equation*}
    This concludes the proof.
\end{proof}

\bibliographystyle{plain}
\bibliography{Library.bib}

\end{document}